%% file: mv21_bj_arxiv.tex
\documentclass{imsart}

\RequirePackage{amsthm,amsmath,amsfonts,amssymb}
\RequirePackage[numbers]{natbib}
\RequirePackage[colorlinks,citecolor=blue,urlcolor=blue]{hyperref}
\RequirePackage{graphicx}

\startlocaldefs

\input{definitions}
\input{packages}
\theoremstyle{plain}

\newtheorem{lemma}{Lemma}
\newtheorem{prop}{Proposition}

\theoremstyle{remark}

\newtheorem{exa}{Example}
\newtheorem{rem}{Remark}
\endlocaldefs

\begin{document}

\begin{frontmatter}
\title{Markov Kernels Local Aggregation for Noise Vanishing Distribution Sampling}
%\title{A sample article title with some additional note\thanksref{t1}}
\runtitle{Locally-weighted MCMC}
%\thankstext{T1}{A sample additional note to the title.}

\begin{aug}
\author[A]{\fnms{Florian} \snm{Maire}\ead[label=e1]{florian.maire@umontreal.ca}},
\author[B]{\fnms{Pierre} \snm{Vandekerkhove}\ead[label=e2]{pierre.vandekerkhove@univ-eiffel.fr}}
%\and
%\author[B]{\inits{T.}\fnms{Third} \snm{Author}\ead[label=e3,mark]{third@somewhere.com}\ead[label=u1,url]{www.foo.com}}
%%%%%%%%%%%%%%%%%%%%%%%%%%%%%%%%%%%%%%%%%%%%%%
%% Addresses                                %%
%%%%%%%%%%%%%%%%%%%%%%%%%%%%%%%%%%%%%%%%%%%%%%
\address[A]{D\'epartement de math\'ematiques et de statistique, Universit\'e de Montr\'eal, Montr\'eal, Canada.
\printead{e1}}

\address[B]{Laboratoire d'analyse et de math\'{e}matiques appliqu\'{e}es,  Universit\'{e} Gustave Eiffel, Champs-sur-Marne, France.
\printead{e2}}
\end{aug}

\begin{abstract}\; A novel strategy that combines a given collection of $\pi$-reversible Markov kernels is proposed. At each Markov transition, one of the available kernels is selected via a state-dependent probability distribution. In contrast to random-scan type approaches that assume a constant (i.e. state-independent) selection probability distribution, the state-dependent distribution is  specified so as to privilege moving according to a kernel which is relevant for the local topology of the target distribution. This approach leverages \textit{paths} or other low dimensional manifolds that are typically present in noise vanishing distributions. Some examples for which we show (theoretically or empirically) that a locally-weighted aggregation converges substantially faster and yields smaller asymptotic variances than an equivalent random-scan algorithm are provided.
\end{abstract}


\end{frontmatter}

\section{Introduction}

Even though Markov chain Monte Carlo (MCMC) methods \citep{brooks2011handbook} are among the most widely used numerical integration algorithms, the long standing practical concern regarding their running time remains fairly open. From a statistical viewpoint, the question is to assess how long should a given MCMC sampler run so that a related Monte Carlo estimator (of an underlying quantity of interest) achieves a prescribed precision level. Early on in the MCMC literature, empirical convergence diagnostics have been proposed (see \cite{cowles1996markov}, \cite{geyer1992practical} and \cite{gilks1996introducing}) and has greatly helped to popularize the use of MCMC in various areas of Science. Addressing such a question theoretically is however highly complex, if possible. Explicit bounds on related asymptotic quantities such as the rate of convergence of the Markov chain (for example in total variation) or its asymptotic variance (for squared integrable functions) may be obtained and, when available, may be used to derive bounds on the mean square error which is closer to the statistical goal, see \cite{latuszynski2013nonasymptotic}. Partially explicit bounds may still be useful to assess scalability or complexity of an MCMC method and can  possibly be used to compare different samplers on that ground. However, those bounds are not always reliable: they can be very conservative especially when the state-space dimension becomes high or in the presence of small isolated local modes. Two of the main methods to obtain total variation bounds, namely the drift and minorization technique  (see \cite{meyn1994computable,rosenthal1995minorization}) and  the approach of directly lower-bounding the Markov chain  $L^2$-spectral gap (see \cite{lawler1988bounds,madras2002markov,jarner2004conductance}) are no exception, see \cite{qin2020limitations}. As a result, they offer little information regarding the estimator itself.

A recent stream of literature has focused on developing new tools for obtaining more precise bounds of convergence, see \cite{qin2019convergence}, \cite{atchade2019approximate}, \cite{yang2017complexity}. A common theme in the works of \cite{yang2017complexity} (on the drift and minorization side) and \cite{atchade2019approximate} (for the spectral gap approach), is to verify the traditional conditions but only on a restricted portion of the state-space called \textit{a large set} in \cite{yang2017complexity}, where, intuitively, the Markov chain mixes well or at least where sharper bounds are easier to derive. In \cite{atchade2019approximate}, it is also desirable that the target distribution concentrates on such a large set. The bounds obtained in \cite{atchade2019approximate} and \cite{yang2017complexity} also have in common to feature an exponentially decreasing term, resembling that of the usual bounds but with different and, typically, much better constants, and an offset term, hopefully small and controlled, resulting from the chain behaviour beyond the large set.

A specific area of applications where this new array of results may be useful is that of sampling from noise vanishing distributions. In this context, the distribution of interest $\pi$ is defined on a high-dimensional set $\Xset$ and its mass concentrates onto some low-dimensional subset $\Zset$. An archetypal example of $\pi$ is given by the law of a random variable $X_\sigma= Z+\sigma\zeta$, $Z$ being a random variable that takes its values on a subset $\Zset\subset\Xset$, $\zeta$ an additive random perturbation and $\sigma>0$ the noise scaling parameter. The subset $\Zset$ is connected and of lower dimension (at least locally) compared to the ambient space $\Xset$ comprising for instance hyperplanes,  manifolds, etc.  This feature characterizes the sparse structure of $\pi$, since the sampling problem is defined on a $d$-dimensional space while, locally and in the limit $\sigma\downarrow 0$, the effective dimension of the sampling space is $d'$, with $d'<d$ (potentially $d'\ll d$). Such a scenario arises in Bayesian statistics, whenever the underlying likelihood model is non-identifiable, or more genuinely in a number of statistical applications including Bayesian inverse problems \citep{knapik2011bayesian}, models involving variables with strong nonlinear relationships \citep{givens1996local} and deterministic simulation models used in Ecology \citep{duan1992effective,bates2003bayesian}, Demography \citep{raftery2010estimating,raftery1995inference}, see also \cite{poole2000inference} and Cosmology \citep{van2009geometry,tempel2014filamentary}.

In this paper, we consider the situation where a family of Markov transition kernels $P_1,P_2,\ldots,P_n$ is given and wonder how to aggregate them in an efficient manner, having in mind (i) the aforementioned recent techniques for deriving quantitative bounds of convergence and (ii) applications to noise vanishing distribution sampling. An efficient aggregation should here be understood as a mixture kernel ${P}_\omega=\sum_{k=1}^n\omega_kP_k$ ($\omega$ is a probability distribution on $\{1,\ldots,n\}$) which leads to a particularly efficient sampling on $\Zset$, perhaps at the expense of poor performance on the complement $\Xset\backslash\Zset$. Recall that a transition from ${P}_\omega$ can be obtained as follows: while at state $X_t$, draw a r.v. $I\sim\omega$ and, conditioned on $I$, draw $X_{t+1}\sim P_I(X_t,\cdot)$.  Central to our approach is the use of state-dependent weights for selecting one of the $n$ available kernels at random, i.e. given $X_t$, $\omega\equiv\omega(X_t)$. This contrasts with the {random-scan} type strategies in which each kernel is assigned a selection probability that does \textit{not} depend on the Markov chain state. In the context of noise vanishing distributions, one can see why state-dependent weights are appealing with the following example.

\begin{exa}\label{ex0}
Let $\Xset=[-2,2]^2$, $\Zset=\{(x_1,x_2)\in[-1,1]^2\,:\,x_1=0\;\mathrm{or}\;x_2=0\}$ and consider sampling from $\pi_\sigma$, the law of the r.v. $X_\sigma=Z+\sigma\zeta$ where $\sigma\in(0,1)$, $\zeta$ is a uniform r.v. on $[-1,1]$ and $Z$ is a r.v. taking values on $\Zset$ and whose distribution, known up to a normalizing constant, cannot be sampled from. To sample from $\pi_\sigma$, assume that $n=2$ random walk Metropolis-Hastings Markov kernels are available: $P_1$ allows ``large'' jumps in the $x_1$ direction and ``small'' jumps in the $x_2$ direction and conversely for $P_2$. By symmetry, the optimal random-scan strategy is to select one of the two kernels uniformly at random. But when $\sigma$ vanishes, this strategy is clearly not efficient: for any state $x\in\Xset\backslash[-\sigma,\sigma]^2$, only one of the two kernels ($P_1$ or $P_2$) is likely to be adapted to the local geometry of $\pi_\sigma$. Hence, in the noise vanishing regime the random-scan Markov chain $\{X_t\}$ satisfies $\proba(X_t=X_{t+1}\,|\,X_t\in\Xset\backslash[-\sigma,\sigma]^2)\geq 1/2$, which is clearly not desirable. This fact is likely to not hold for a locally-weighted weighting mechanism that would give a (much) larger selection probability to $P_1$ whenever the Markov chain is at a state $x$ such that $|x_2|>\sigma$ and to a (much) larger one to $P_2$ whenever $|x_1|>\sigma$. This example easily generalizes to $d$ dimensions with $\Xset=[-2,2]^d$ and a collection of $n=d$ Markov kernels constructed similarly as for $d=2$. In the noise vanishing regime, the random-scan satisfies $X_t=X_{t+1}$ with probability greater than $1-1/d$. In such a problem and when $\sigma\to 0$, one can hope that the convergence of a Markov chain featuring an adequate locally-weighted strategy to select from $P_1,\ldots,P_d$ would scale better with $d$ than a random-scan strategy, possibly offering dimension free convergence rate.
\end{exa}

We first study the feasibility of such a locally-weighted strategy in a general context. As we shall see, several ways to aggregate Markov kernels with state-dependant weights are possible and several weight functions can be considered. Since weight functions are designed so that a locally-weighted strategy mixes fast on $\Zset$, the potential slow exploration of $\Xset\backslash\Zset$ slows down the resulting Markov chain convergence, all the more when $\sigma$ is not ``sufficiently'' small. As a consequence, establishing theoretical results showing the superiority of a locally-weighted strategy over a random-scan strategy is challenging, if true, because of that asymptotic artifact. However, the aforementioned recent techniques which break down the convergence bounds into an exponentially decaying term and an additive offset seem particularly well suited for showing the non-asymptotic relevance of such locally-weighted aggregation strategies. In fact, we believe that those bounds could even be optimised by mean of locally informed weights, but we leave these aspects of the question for future research. This is worth mentioning since, as MCMC algorithms are only used for a finite runtime, non-asymptotic bounds (e.g. bounds on the MCMC estimator mean square error) are probably the most useful criterions for choosing a sampler, hence motivating algorithms such as the ones presented in this paper. We indeed show through a series of examples (empirical or theoretical), the benefit of locally-weighted strategies over random-scan equivalent and present how they can be implemented in practice.

\paragraph*{Notation}
 For any $c\in\rset$ and $q\in\nset$, $\mathbf{c}_q$ is the constant vector $(c,c,\ldots,c)\in\rset^q$. Consider a measurable space $(\Xset,\Xalg)$. For all $x\in\Xset$, $\delta_x$ is the dirac distribution on the singleton $\{x\}$ and for all $A\subset \Xset$, $\mathds{1}_A$ is the indicator function of set $A$. As is standard in the Markov chain Monte Carlo literature,  a Markov chain defined on $(\Xset,\Xalg)$ is generically denoted $\{X_t\}:=\{X_t\,:\,t\in\nset\}$. For any Markov kernel $P$ and any measurable function $f:\Xset\to\rset$, $Pf:=\esp\{ f(X_1)|X_0=\cdot\}$ and for any probability measure $\mu$ on $(\Xset,\Xalg)$, $\mu P^t$ is the law of $X_t$, where $X_0\sim \mu$. Moreover, $\Pr_\mu$, $\esp_\mu$, $\var_\mu$ refer to the corresponding operator related to an unambiguous Markov chain with initial distribution $\mu$. With some abuse of notation, replacing $\mu$ by $x$ for any $x\in\Xset$ in this notation refers to the case where $\mu=\delta_x$. For a signed measure $\nu$ on $(\Xset,\Xalg)$, $\|\nu\|$ denotes the total variation of $\nu$ and the space of square $\nu$-integrable real-valued functions is denoted by $\mathcal{L}^2(\nu):=\{f:\Xset\to\rset\,:\,\int f^2\rmd \nu<\infty\}$.

\section{Related work}
\label{sec:2}

Our work revolves around three topics: sampling from noise vanishing distributions using a local target dependent information to aggregate Markov kernels. Here is a very brief review of the existing literature on those subjects which motivates our work.

Research on efficient MCMC sampling for locally low-dimensional distributions can be traced back at least to the seminal work of \cite{girolami2011riemann}. First and second order Monte Carlo methods based on Langevin of Hamiltonian dynamics have proven very efficient to address those sampling problems when they incorporate geometric information on the model at hand such as the first and second derivatives of the target, the Fisher information, a Riemannian metric or specific knowledge on the manifold (its tangent space, a forward operator it stems from, etc.), see \cite{byrne2013geodesic,livingstone2014information,zappa2018monte,au2020manifold} and the references therein. These methods are undoubtedly state of the art whenever one has access to the specifics they require. We make the assumption that not much is known on the target distribution beyond its unnormalized density and that a collection of $n$ Markov kernels, which taken individually might lead to a poor algorithm, are available.

While not being restricted to them, those kernels can be thought of as Gaussian random walk Metropolis-Hastings kernels (MH kernels in the sequel) with different covariance matrices. Our work is thus perhaps more closely related to the  idea of a MH algorithm with position-dependent proposal covariance matrix studied in \cite{livingstone2015geometric} and to the locally-balanced proposal construction of \cite{zanella2020informed}. The connection between the aim of the latter, while being practically restricted to countable state space applications and our work is particularly obvious and will be further discussed.   Even though the locally-weighted strategy presented in this paper specifies a \textit{proper} Markov chain, readers familiar with the literature on adaptive MCMC will also see a resemblance with the Regional Adaptive MH of \cite{craiu2009learn} and to a certain degree to the adaptive Gibbs sampler of \cite{latuszynski2013adaptive}, see also \cite{levine2006optimizing}.

Regarding the aggregation of Markov kernels, it was shown in  \cite{roberts1997geometric} and \cite{roberts1998two} that a mixture of kernels (with state-independent weights)  inherits under mild assumptions the properties of the individual kernels. As we shall see, things are quite different when using state-dependent mixture weights since, if done inadequately, the aggregated kernel may be transient. The question of \textit{optimal} selection probability is known to be a challenging one  and \cite{andrieu2016random} is probably the only literature entry on this topic. The author carries out a thorough exploration of the hybrid Gibbs case, with $n=2$ kernels, and compares the random-scan and the deterministic-update Gibbs sampler, according to their asymptotic variance. While the question, which our work does not pretend addressing,  is even more involved for state-dependent weights, this paper suggests that, in the noise vanishing regime, locally-weighted strategies with reasonable weight function scale better than state-independent ones. Finally, we take note of the very recent work of \cite{herschlag2020non} in which a non-reversible MH algorithm based on a locally-weighted mixture of proposal distributions is derived. This is a particular construction that we also consider here but the analyses and applications carried out in the two papers are pushed in different directions.

\section{Locally-weighted aggregation}
\label{sec:3}
Let $\pi$ be a probability distribution on some measurable space $(\Xset,\Xalg)$ that we are interested in sampling from.
We consider a collection of $n$ $\pi$-reversible Markov kernels defined on $(\Xset,\Xalg)$
$$
\Pfrak:=\{P_1,P_2,\ldots,P_n\}
$$
and call any Markov chain that moves at each transition according to one of these kernels selected at random an aggregation of $\Pfrak$. More precisely, let $\omega$ be a probability on $\{1,2,\ldots,n\}$, that is $\omega=(\omega_1,\ldots,\omega_n)\in\Delta_{n-1}$ where $\Delta_{n-1}$ is the $(n-1)$-simplex and denote by $P_\omega:=\sum_{i=1}^n\omega_iP_i$ the hybrid kernel aggregating $\Pfrak$.  To sample from $X_{t+1}|X_t\sim P_\omega(X_t\,,\,\cdot\,)$, one may simply draw an index $I\sim \omega$ independently of $X_t$ and then, conditionally on $(I,X_t)$, draw $X_{t+1}\sim P_I(X_t,\,\cdot\,)$. The random-scan Gibbs sampler follows this type of transition mechanism with $n=d$, $P_i(x,\cdot)=\pi(\cdot\,|\,x_{-i})\delta_{x_{-i}}$ and $\omega$ typically set as the uniform distribution over $\{1,\ldots,n\}$. As is well known, $P_\omega$ is also $\pi$-reversible. Indeed, for any $(A,B)\in\Xalg\otimes\Xalg$, we have
\begin{multline}
\label{eq:omega_const}
\int_A\pi(\rmd x) P_\omega(x,B)=\int_{A} \pi(\rmd x)\sum_{i=1}^n\omega_i P_i(x,B)
=\sum_{i=1}^n\omega_i\int_{A} \pi(\rmd x) P_i(x,B)\\
=\sum_{i=1}^n\omega_i\int_{B} \pi(\rmd x) P_i(x,A)=\int_B\pi(\rmd x) P_\omega(x,A)\,.
\end{multline}
We refer to \cite{roberts1997geometric} for more details on conditions under which $P_\omega$ inherits other theoretical properties shared by the Markov kernels in $\Pfrak$.

In this section, we generalize the construction of hybrid kernels to allow dependance of the selection probability $\omega$ on the current state of the Markov chain. To avoid confusions, we will denote generically a state-dependent probability by $\varpi:\Xset\to\Delta_{n-1}$, that is $\varpi\in\Delta_{n-1}^{\Xset}$ and, by contrast, a state-independent probability by $\omega\in\Delta_{n-1}$.

It can be checked that $P_\varpi$ (whose transition writes $X_{t+1}|X_t\sim P_{\varpi(X_t)}(X_t\,,\,\cdot)$) is not necessarily $\pi$-reversible as the second equality in \eqref{eq:omega_const} does not hold, in general. To highlight this aspect, let us illustrate with the following example in which  switching from $P_\omega$ to $P_\varpi$ changes the stationary distribution.
\begin{exa}
Let $\Xset=\{1,2\}$, $a\in(0,1)$ and consider $n=2$ kernels defined as
$$
P_1=\begin{pmatrix}
      a & 1-a\\
      1-a & a
    \end{pmatrix}\qquad
P_2=\begin{pmatrix}
      1-a & a\\
      a & 1-a
    \end{pmatrix}    \,.
$$
It can be readily checked that $P_1$ and $P_2$ are reversible with respect to the uniform distribution on $\Xset$, denoted $\mathcal{U}(\Xset)$. As a consequence, setting $\omega=\{b,1-b\}$ for any $b\in(0,1)$, we have that  $P_{\omega}$ is also reversible with respect to $\mathcal{U}(\Xset)$. For any $\eps\in[0,b]$, define the function ${\varpi}_\eps:\Xset\to\Delta_{n-1}$ for any $x\in\{1,2\}$ by
$$
{\varpi}_\eps(x):=\{b-\eps+2\eps\mathds{1}_{\{x=2\}},1-b+\eps-2\eps\mathds{1}_{\{x=2\}}\}\,.
$$
It can be checked that $P_{\varpi_\eps}$ is reversible with respect to the distribution $\pi_\eps$ defined such that
$$
\pi_{\eps}(X=1)=\frac{1}{1+\frac{a+(b-\eps)(1-2a)}{a+(b+\eps)(1-2a)}}\,,
$$
which coincides with the $\mathcal{U}(\Xset)$ if and only if $\eps=0$, i.e. in case of state-independent selection probability ($\varpi_0=\omega$). Depending on $(a,b)$, $P_{\omega}$ and $P_{\varpi_\eps}$ may have other further differences: taking $a=1/4$ and $b=0$, one can check that $P_{\varpi_\eps}$ produces i.i.d. samples from $\pi_\eps$ while $P_{\omega}$ produces dependent samples from $[1/2\,,1/2]$ (if started from that distribution).
\end{exa}

It is easy to find examples where $P_\varpi$ does not even admit a stationary distribution despite $P_1,\ldots,P_n$ all being $\pi$-reversible. There is in fact a simple way to ``correct'' the algorithm so as to retrieve the $\pi$-invariance from $P_1, P_2,\ldots,P_n$ by introducing an accept-reject step. We refer to this type of algorithm as a \textit{locally-weighted} aggregation of kernels in $\Pfrak$. To further emphasize this correct construction, we denote by $\Past_\varpi$ the resulting locally-weighted Markov kernel, by contrast with the uncorrected kernel $P_\varpi$.
\begin{algorithm}
\caption{Locally-weighted MCMC, transition $X_{t}\to X_{t+1}$}
\label{alg1}
\begin{algorithmic}[1]
\Require $X_t=x\in\Xset$
\State draw $I\sim\varpi(x) \rightsquigarrow i$
\State propose $\tX\sim P_i(x,\cdot)\rightsquigarrow \tx$ and set $X_{t+1}=\tx$ with probability
\begin{equation}
\label{eq1}
\alpha_i(x,\tx)=1\wedge\frac{\varpi_i(\tx)}{\varpi_i(x)}
\end{equation}
and set $X_{t+1}=x$ otherwise.
\end{algorithmic}
\end{algorithm}
The locally-weighted Markov chain is described at Algorithm \ref{alg1} and its transition kernel writes:
\begin{multline}
\label{eq2}
\Past_\varpi(x,A)=\sum_{i=1}^n\varpi_i(x)\left\{\int_A P_i(x,\rmd y)\alpha_i(x,y)+\delta_x(A)\left(1-r_i(x)\right)\right\}\,,\\
r_i(x):=\int_{\Xset} P_i(x,\rmd y)\alpha_i(x,y)\,.
\end{multline}

\begin{prop}
\label{prop2}
Assume that for all $i\in\{1,\ldots,n\}$, $P_i$ is $\pi$-reversible, then for any choice of function $\varpi:\Xset\to\Delta_{n-1}$, $\Past_\varpi$ is $\pi$-reversible.
\end{prop}

\begin{proof}
Let $\rho$ be the measure on $\Xalg\otimes\Xalg$ defined as $\rho(A,B):=\int_A\pi(\rmd x)\Past_\varpi(x,B)$ and $H:\Xset^2\to\rset$ a measurable test function. $\Past_\varpi$ is $\pi$-reversible if and only if for any $(X,Y)\sim \rho$, $\esp H(X,Y)=\esp H(Y,X)$. We thus have
\begin{multline*}
\esp H(X,Y)
=\sum_{i=1}^n\iint_{\Xset}H(x,y)\pi(\rmd x)[P_i(x,\rmd y)\left\{\varpi_i(x)\wedge \varpi_i(y)\right\}\\
+\delta_x(\rmd y)\varpi_i(x)\left(1-r_i(x)\right)]\\
=\iint_{\Xset}H(x,y)\pi(\rmd y)\bigg\{\sum_{i=1}^n\varpi_i(y)P_i(y,\rmd x)\left\{1\wedge\frac{\varpi_i(x)}{\varpi_i(y)}\right\}\\
+\delta_y(\rmd x)\sum_{i=1}^n\varpi_i(y)\left(1-r_i(y)\right)\bigg\}\\
=\iint_{\Xset}H(x,y)\pi(\rmd y)\Past_\varpi(y,\rmd x)=\iint_{\Xset}H(y,x)\pi(\rmd x)\Past_\varpi(x,\rmd y)
=\esp H(Y,X)\,,
\end{multline*}
where the second equality follows from the $\pi$-reversibility of $P_i$ and the symmetry of the measure $\pi(\rmd x)\delta_x(\rmd y)$ on $\Xalg\otimes\Xalg$.
\end{proof}

\noindent Since $\pi$-reversible Markov kernels are necessarily $\pi$-invariant, $\Past_\varpi$ is also $\pi$-invariant.

\begin{rem}
For state-independent selection probability $\omega$, $P_\omega=\Past_\omega$ since $\alpha_i=1$ for all $i\in\{1,\ldots,n\}$. We thus see $\Past_\varpi$ as a generalization of the usual hybrid kernel, beyond the case of state-independent selection probability.
\end{rem}

The usefulness of a state-dependent selection probability to construct hybrid kernels is easy to grasp in certain situations (see Example \ref{ex0}), but  one may wonder what can be formally said on the efficiency (asymptotic variance, convergence time, etc.) of $\Past_\varpi$  relatively to $P_\omega$. A more precise question would be which is more efficient between $\Past_\varpi$ for the best $\varpi\in\Delta_{n-1}^\Xset$ (in a certain context) and $P_\omega$ for the best $\omega\in\Delta_{n-1}$ (in the same context): indeed the accept-reject step of $\Past_\varpi$ introduces a compromise between the two strategies and makes the question seriously challenging. In particular, it can be noted that $P_\omega$ and $\Past_\varpi$ do not admit a Peskun ordering \citep{peskun1973optimum}, in all generality.

The following section introduces a noise vanishing distribution (on a discrete state space) which concentrates on certain edges of a $d$-dimensional hypercube. In this context, we are able to answer to the previous question when efficiency is measured with respect to both mixing time and spectral quantities, in the somewhat degenerate case of zero-noise scenario ($\sigma=0$). More precisely, we show that $\Past_\varpi$ with the best state-dependent probability is $d/2$ better than $P_\omega$ with the best state-independent probability, for both criteria of performance. Beyond that degenerate case, we were not able to carry out calculations explicitly and resort to computer-based numerical results to compare the two kernels convergence when $\sigma>0$.

\section{Analysis of a canonical example}
\label{sec:4}

\begin{exa}
\label{ex1}
Let $\Xset=\{1,\ldots,m\}^d$ where $d\geq 2$ and $m\geq 3$ be the $d$-dimensional discrete hypercube with edge length $m$.  Define the subset $\Zset\subset \Xset$ that comprises the connected edges $\Eset_1,\Eset_2,\ldots,\Eset_d$ defined as follows:
\begin{multline}
\Zset:={\bigcup}_{i=1}^d\Eset_i\,,\\
\Eset_1:=\left\{x\in\Xset\;\big|\;x_{2}=\cdots=x_d=1\right\}\,,\Eset_d:=\left\{x\in\Xset\;\big|\;x_{1}=x_2=\cdots=x_{d-1}=m\right\}\,,\\
\Eset_i:=\left\{x\in\Xset\;\big|\;x_{1}=\cdots=x_{i-1}=m\quad\text{and}\quad x_{i+1}=\cdots=x_d=1\right\}\,,\quad 1<i<d\,.
\end{multline}
Finally, define for any $\sigma\in[0,1]$, the distribution $\pi_\sigma$ on $\Xset$ as the mixture
$$
\pi_\sigma=(1-\sigma)\mathcal{U}(\Zset)+\sigma\mathcal{U}(\Xset\backslash \Zset)\,,
$$
with $|\Xset|=m^d$, $|\Zset|=(d-1)(m-1)+m$. The distribution $\pi_0$ is illustrated graphically at Figure \ref{fig:hypercube:1}, with $d=3$ and $m=10$. Even though sampling i.i.d. draws from $\pi_\sigma$ is straightforward,  we consider for illustrative purpose aggregating  the collection of Gibbs kernels $P_1,\ldots,P_n$ with $n=d$ and where $P_i(x,\rmd y)=\pi_\sigma(\rmd y_i\,|\,x_{-i})\delta_{x_{-i}}(\rmd y_{-i})$, $i\in\{1,\ldots,d\}$.\footnote{Similar examples for which sampling i.i.d. draws from $\pi_\sigma$ is not straightforward exist, for instance whenever the path $\Zset$ is unknown.} By symmetry, the best state-independent selection probability $\omega$ is given by $\omega=(1/d,1/d,\ldots,1/d)$.  By contrast, there are several relevant state-dependent selection probability distributions  $\varpi(x)=(\varpi_1(x),\ldots,\varpi_d(x))$ given for example by
\begin{align}
\label{eq:omegas}
&\varpi_i(x)=\left\{
\begin{array}{cc}
(1-\sigma)\frac{\mathds{1}_{\{x\in\Eset_i\}}}{\sum_{\ell=1}^d\mathds{1}_{\{x\in\Eset_\ell\}}}+\sigma/{d}\,,&x\in\Zset\\
1/d\,,& x\not\in\Zset
\end{array}
\right.\\
&\quad\mathrm{or}\quad\nonumber\\
&\varpi_i(x)\propto\sum_{\ell=1}^m\pi_\sigma(x_1,\ldots,x_{i-1},\ell,x_{i+1},\ldots,x_d)\nonumber\,,
\end{align}
for $i\in\{1,\ldots,d\}$. Intuitively, those choices of $\varpi$ are designed so as to pick the ``right'' full conditional distribution(s), i.e. the one(s) that move(s) the chain on the same edge but, contrarily to  $P_\omega$, at a different state (with high probability).
\end{exa}

\begin{figure}
\centering
 \includegraphics[scale=0.7]{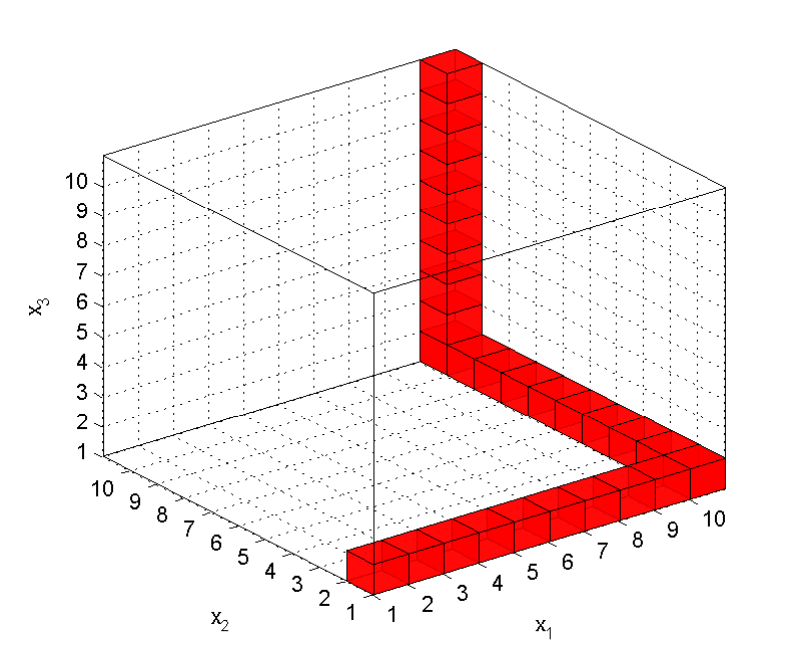}%\includegraphics[scale=0.7]{./figure/hypercube/TV_p01_th_2.eps}
\caption{Example \ref{ex1} with $m=10$, $d=3$ and $\sigma=0$. The subset $\Zset$ are the states in red and all the mass of $\pi$ is concentrated on $\Zset$.
\label{fig:hypercube:1}}
\end{figure}

\paragraph*{The noise-free case: $\sigma=0$.}
Let $\sigma=0$ in Example \ref{ex1}. First, with probability of at least $1-2/d$, $P_\omega$  selects a kernel $P_i$ that prevents the chain to move as at least $1-2/d$ full posterior distributions $\pi_0(\,\cdot\,|\, X_{-i})$ have their probability mass concentrated exclusively on $X_i$. Hence, when $d$ is large, the Markov chain hardly moves. Second, taking any of the choice for $\varpi$ in Eq. \eqref{eq:omegas} (in fact the two options coincide in the large $d$ limit and $\sigma=0$), we expect $\Past_\omega$ to be more robust or even to benefit from the absence of noise as it generates a chain roaming swiftly on $\Zset$. It can also be checked that the first choice of $\varpi$ in Eq. \eqref{eq:omegas} dominates the second one in the Peskun ordering sense and, in the noise-free regime, is in fact the best state-dependent selection probability. In the sequel, we  consider $\Past_\varpi$ with this choice of selection probability. The intuition that $\Past_\varpi$ is more efficient than $P_\omega$ is made rigourous in the following propositions which suggest that $\Past_\varpi$ is typically $d/2$ more efficient  than $P_\omega$,  both in a non-asymptotical and asymptotical way.

Before stating those results, recall the definition of a coupling time associated with a  Markov kernel $P$. Let $\{X_t,X'_t\}$ be a discrete time process defined on $(\Xset\times\Xset,\Xalg\otimes\Xalg)$ such that marginally $\{X_t\}$ and $\{X'_t\}$ are both a Markov chain with transition kernel $P$ with initial distribution $\mu$ and $\pi$, respectively. The coupling time of the joint process $\{X_t,X_t'\}$ is the random variable $\tau$ defined as $\tau:=\inf_{t\in\nset}\{X_t=X'_t\}$. The coupling time $\tau$ is indeed a time characteristic to the Markov chain speed of convergence since the coupling inequality (see \eg \cite{rosenthal1997faithful}) states that for all $t\in\nset$,
$$
\|\Pr\{X_t\in\,\cdot\,\}-\pi\|\leq \Pr\{\tau>t\}\,.
$$
As a consequence, results on the mixing time (such as Proposition \ref{prop:hypercube:1}) are related to the non-asymptotic convergence performance of the Markov chain. By contrast, the two other results (such as Propositions \ref{prop:hypercube:2} and \ref{prop:hypercube:3})  are related to the asymptotic performance of the Markov chain, and in particular  to its asymptotic variance and its asymptotic exponential rate of convergence. Considering a $\pi$-reversible kernel $K$, we recall that  the absolute spectral gap of $K$, denoted $\gap(K)\in[0,1]$, defined formally in the appendix, can be used to characterize the asymptotic efficiency of $K$. Indeed, it is well known, see for instance \cite{rosenthal2003asymptotic}, that for any square integrable function $f$ and any initial distribution $\mu_0$
$$
\var(K,f)\leq \frac{2}{\gap(K)}\var_\pi f(X)\,,\qquad \lim_{t\to\infty}\frac{1}{t}\log\|\mu_0K^t-\pi\|=\log(1-\gap(K))\,,
$$
where $\var(K,f):=\lim_{t\to\infty}t\var\left[(1/t)\sum_{i=1}^t f(X_i)\right]$ is the asymptotic variance of MCMC estimator of $\esp_\pi f(X)$ and the variance is here taken with respect to a Markov chain $\{X_t\}$ started at stationarity and with transition kernel $K$. In other words, the larger the absolute spectral gap, the more efficient the Markov chain.
\begin{prop}
\label{prop:hypercube:1}
In the context of Example \ref{ex1} with $\sigma=0$, the expected coupling time of the $P_\omega$ (denoted $\tau$) is $d/2$ times larger than $\Past_\varpi$ (denoted $\tau^\ast$), when both Markov chains are initialized at state $\mathbf{1}_d\in\Xset$, \ie
\begin{equation}
\label{eq:prop1}
\esp_{\mathbf{1}_d}(\tau)= \frac{d}{2}\esp_{\mathbf{1}_d}(\tau^\ast)\,.
\end{equation}
\end{prop}

\begin{prop}
\label{prop:hypercube:2}
In the context of Example \ref{ex1} with $\sigma=0$ and $d$ even, the absolute spectral gap of $P_\varpi^\ast$ and $P_\omega$ satisfy
$$
\mathrm{Gap}(P_\varpi^\ast)= \frac{d}{2} \mathrm{Gap}(P_\omega)\,.
$$
\end{prop}
\begin{prop}
\label{prop:hypercube:3}
In the context of Example \ref{ex1} with $\sigma=0$, for any function $f:\Xset\to\rset$
$$
\var(P_\varpi^\ast,f)\leq \frac{2}{d}\var(P_\omega,f)+\left(\frac{2}{d}-1\right)\var_{\pi_0} f(X)\,.
$$
\end{prop}
The proof of those propositions can be found in Supplementary Material (Sections \ref{proof1}, \ref{proof2} and \ref{proof3}, respectively). The first two proofs are based on an equivalent representation of the Markov chains which are ``folded'' onto a simpler state space. This representation makes easier the construction of couplings and the analysis of their spectral properties.
\begin{rem}
\label{rem0}
The factor $d/2$ in the Propositions \ref{prop:hypercube:1} and \ref{prop:hypercube:2} can be interpreted as follows: since $\pi$ is uniform on $\Zset$, the convergence of both Markov chains (starting from one extremity of the filament) is characterized by the speed at which they traverse the hypercube vertices that belong to $\Zset$, \eg $(10,1,1)$ and $(10,10,1)$ for the case illustrated in Figure \ref{fig:hypercube:1}. While at one of those vertices, the relative speed at which $P_\omega$ moves to one of the two adjacent edges compared to $\Past_{\varpi}$ is $2/d$ since ``only'' two choices of direction may lead to such a transition. Those vertices can be seen as bottleneck for $\Past_\varpi$ when compared to $P_\omega$. Consider the alternate definition of $\Zset$ in the case $d=3$ with $\Zset:=\Eset_1\cup\Eset_2\cup \Eset_3'$ where $\Eset_3':=\{x\in\Xset\,|\,x_1=m,x_2=1,x_3\in(1,m)\}$. Thus, the state $(m,1,1)$ connects the three subspaces of dimension one. Hence, $\Past_\varpi$ and $P_\omega$ are equally efficient to jump to any edge while at this state. This reflects in the quantitative factor of  Eq.\eqref{eq:prop1} which drops to $d/3=1$.
\end{rem}

These first results confirm the intuition that, in a situation such as Example \ref{ex1}, a state-dependent distribution $\varpi$ that incorporates geometric information of $\pi$ to draw the updating direction of a Gibbs sampler can speed up the Markov chain convergence. Again, we stress that obtaining those theoretical results is eased by the fact that the mass of $\pi$ is here concentrated on the filamentary path $\Zset$, \ie $\sigma=0$.

\paragraph*{The noise-vanishing case $\sigma>0$ and discontinuity at $\sigma=0$.}
Considering the situation where $\sigma>0$ is certainly interesting but makes the whole analysis more involved. In particular, the useful folded representation of the Markov chains in the noise-free case is no longer available. However, whenever $m^d$ is not too large for modern days computers, one can obtain numerical convergence results by calculating the transition matrices $P_\omega$ and $\Past_\varpi$ and using that the total variation of a signed measure satisfies $\|\delta_{\mathbf{1}_d}P^t-\pi_\sigma\|=(1/2)\sum_{x\in\Xset}\left|\delta_{\mathbf{1}_d}P^t(x)-\pi_\sigma(x)\right|$. Figure \ref{fig:hypercube:2} shows the convergence of the Markov chains in the case $d=5$, $m=4$ and noise levels $\sigma\in\{0,10^{-5},10^{-2},10^{-1}\}$. Both columns show the same plots, in different scales. While the linear Y-axis scale allows to assess the speed of convergence of the Markov chains in the transient phase, the logarithmic scale outlines the latter in asymptotic regime. As a check, the case $\sigma=0$ at the first row of Figure \ref{fig:hypercube:2} illustrates Propositions \ref{prop:hypercube:1} and \ref{prop:hypercube:2}: $\Past_\varpi$ dominates ${P}_\omega$ both asymptotically and non-asymptotically and is in fact $d/2=2.5$ times more efficient. Indeed, the $d/2$-slowed version of $\Past_\varpi$ (that is, the kernel $(2/d)\Past_\varpi+(1-2/d)\mathrm{I}$ which only moves according to $\Past_\varpi$ w.p. $2/d=0.4$ and stays put otherwise) converges at the same rate as $P_\omega$.  Moreover, looking at the next rows, the fact that $P_\varpi^\ast$ is $2.5$ more efficient than $P_\omega$ remains true in the transient phase but less so when $\sigma$ increases. Indeed, as $\sigma$ increases, the initial convergence phase is hardly affected as long as $\sigma$ is not too large, i.e. $\sigma\leq 0.01$. It is however never true in the asymptotic regime, which leads to the main point of the analysis. Looking at the asymptotic rate (right column), a shortcoming of $\Past_\varpi$ is exposed: after exploring swiftly $\Zset$, it converges far more slowly on $\Xset\backslash\Zset$. The transition to this slow convergence regime occurs when $\delta_{\mathbf{1}_d}{{P_\varpi^{\ast\,t}}}$ approximates the conditional distribution $\pi_\sigma(\,\cdot\,|\Zset)$. The black dashed lines on the RHS plots of Figure \ref{fig:hypercube:2} indicate $\sigma=\pi_\sigma(\Xset\backslash \Zset)$. Indeed, note that
\begin{multline}
\label{eq:tv_cond}
\|\pi_\sigma-\pi_\sigma(\cdot\,|\,\Zset)\|
=(1/2)\sum_{x\in\Xset}\left|\pi_\sigma(x)-\pi_\sigma(x)\mathds{1}_{\{x\in\Zset\}}/\pi_\sigma(\Zset)\right|
\\
=(1/2\pi_\sigma(\Zset))\sum_{x\in\Zset}\left|\pi_\sigma(\Zset)\pi_\sigma(x)-\pi_\sigma(x)\right|
+(1/2)\sum_{x\in\Xset\backslash\Zset}\left|\pi_\sigma(x)\right|\\
=(1/2)(1-\pi_\sigma(\Zset))+(1/2)\pi_\sigma(\Xset\backslash\Zset)=\sigma\,.
\end{multline}
Hence, in terms of convergence, after transforming rather quickly the initial measure $\delta_{\mathbf{1}_d}$ to the conditional measure $\pi_\sigma(\,\cdot\,|\,\Zset)$, the locally-weighted chain hit a bottleneck when pushing $\pi_\sigma(\,\cdot\,|\,\Zset)$ to $\pi_\sigma$ and will eventually (how fast depends on the noise level $\sigma$) be outpaced by $P_\omega$. This fact may seem surprising since, by definition of $\varpi$, the kernel $P_\varpi^\ast(x,y)=P_\omega(x,y)$ for each $(x,y)\in\Xset\backslash\Zset$. In other words, when considering the restriction of $\pi_\sigma$ to $\Zset$, $P_\varpi^\ast$ is $d/2$ times faster than $P_\omega$ and when considering its restriction to $\Xset\backslash\Zset$, $P_\varpi^\ast=P_\omega$. Yet, $P_\omega$ is asymptotically better than $P_\varpi^\ast$ when $\sigma>0$.

The reason is that the weight function $\varpi$ Eq. \eqref{eq:omegas} coerces ``too strongly'' the chain on $\Zset$, yet without preventing convergence of  $\delta_x P_\varpi^{\ast\,t}$ to $\pi$, for each $x\in\Xset$. Indeed, in the small noise regime $\sigma>0$, selecting a kernel taking the chain out of $\Zset$ while in $\Zset$ occurs w.p. $\sigma/d $ which vanishes as $\sigma\downarrow 0$. The same probability also holds when the chain visits $\Xset\backslash\Zset$ and try to reach $\Zset$. Such a behaviour is a characteristic of meta-stable Markov processes which are well studied in molecular dynamics and known for their poor convergence properties. The fact that the convergence on $\Xset\backslash \Zset$ slows down with $\sigma$ can be seen in the right column of Figure \ref{fig:hypercube:2} (the slope of plateaus decreases with $\sigma$, $\sigma>0$). Letting $\rho_\varpi$ be the exponential convergence rate of $P_\varpi^\ast$, one may surmise based on those plots that there exists $\text{c}_{m,d}<0$ such that $\log \rho_\varpi\sim \text{c}_{m,d}\sigma$ as $\sigma\downarrow 0$, for fixed $d$ and $m$. This would suggest that $\sigma\mapsto \text{Gap}(P_\varpi^\ast)$ is not continuous at $\sigma=0$ since $\lim_{\sigma\downarrow 0}\text{Gap}(P_\varpi^\ast)=0$ while, from Proposition \ref{prop:hypercube:2}, $\text{Gap}(P_\varpi^\ast)>0$ at $\sigma=0$.

Motivated by the fact that, since for any $(i,x)\in\{1,\ldots,d\}\times \Xset$, $\omega_i(x)=1/d$, a slight variation of $\varpi$ defined by
\begin{equation}\label{eq:varpi_und}
\underline{\varpi}_i(x):\propto\varpi_i(x)\vee 1/d^2\,,\qquad (i,x)\in\{1,\ldots,d\}\times \Xset
\end{equation}
may allow to reach a tradeoff between the ``griddy'' locally-weighted kernel $P_\varpi^\ast$ which learns efficiently $\pi_\sigma(\,\cdot\,|\,\Zset)$ and the ``na\"ive'' random-scan kernel $P_\omega$ which, by contrast with $P_\varpi^\ast$, scales with $\sigma$ to learn $\pi_\sigma(\,\cdot\,|\,\Xset\backslash\Zset)$. Compared to $P_\varpi^\ast$, Figure \ref{fig:hypercube:2} shows that $P_{\underline{\varpi}}^\ast$ is, as expected, mildly slower to transform $\delta_{\mathbf{1}_d}$ to $\pi_\sigma(\,\cdot\,|\,\Zset)$ but deals much better with the asymptotic regime (as shown at the right column).

In line with the results of \cite{roberts1997geometric}, this example shows that the  design of $\varpi\in\Delta_{n-1}^\Xset$ should guarantee that each weight $\varpi_i(x)$ does not vanish to zero with a parameter of the model (\eg $\sigma$). However, as we shall see at Section \ref{sec:6} with the applications, the fact that $P_{\underline{\varpi}}^\ast$ is possibly less efficient asymptotically than $P_\omega$ is typically irrelevant since for noise vanishing distributions this regime is hardly reach for a reasonable runtime. These observations call for further non-asymptotical analysis which could be eased by modern tools such as the pseudo-spectral gap \citep{atchade2019approximate} or large sets \citep{yang2017complexity}, discussed in the introduction section. We leave these considerations for future work. In a slightly different, yet related, stream of ideas, we conclude this section by mentioning that a recent result in \cite[Theorem 1]{gagnon2021nonreve} can be readily applied to the context of Example \ref{ex1} with $\sigma>0$.  This states that for any $\epsilon>0$ and any bounded function $f:\Xset\to\rset$, there exists a \textit{tolerable} noise level $\sigma_0$, such that for any $\sigma<\sigma_0$,
$$
\var(f,P_{\underline{\varpi}}^\ast)\leq \frac{2}{d}\var(P_\omega,f)+\left(\frac{2}{d}-1\right)\var_{\pi_\sigma} f(X)+\frac{4}{d}\epsilon\,.
$$

\begin{figure}
\centering
\includegraphics[scale=0.68]{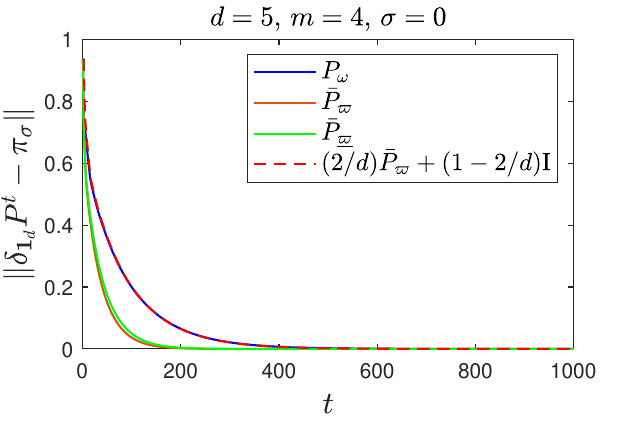}\includegraphics[scale=0.68]{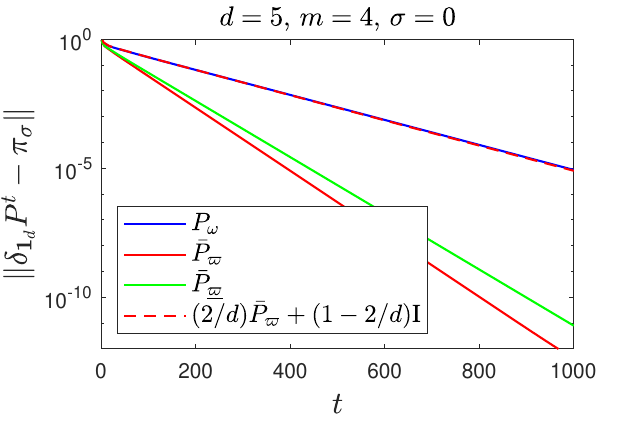}

\includegraphics[scale=0.68]{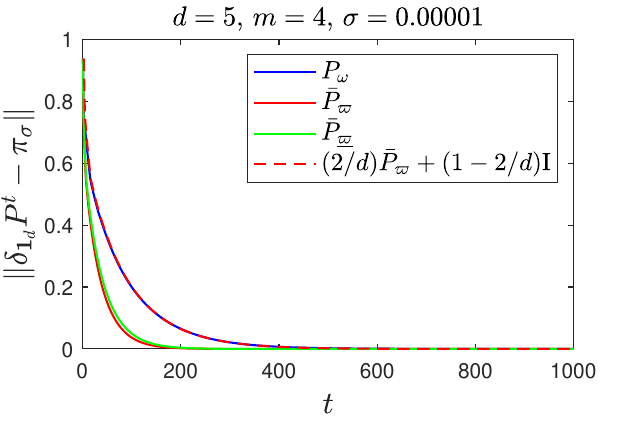}\includegraphics[scale=0.68]{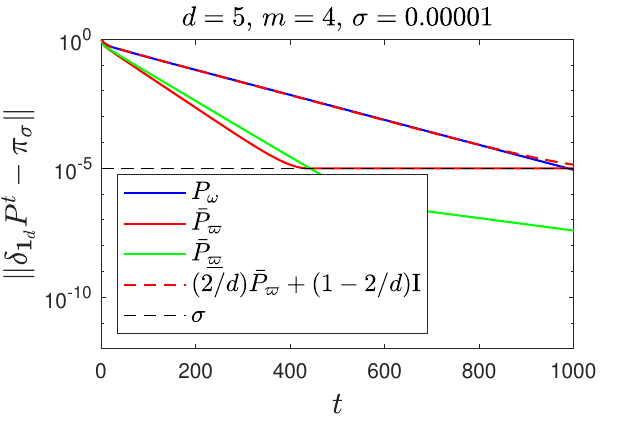}

\includegraphics[scale=0.68]{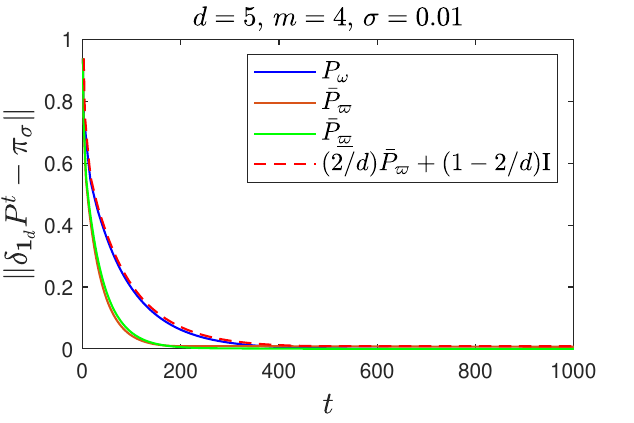}\includegraphics[scale=0.68]{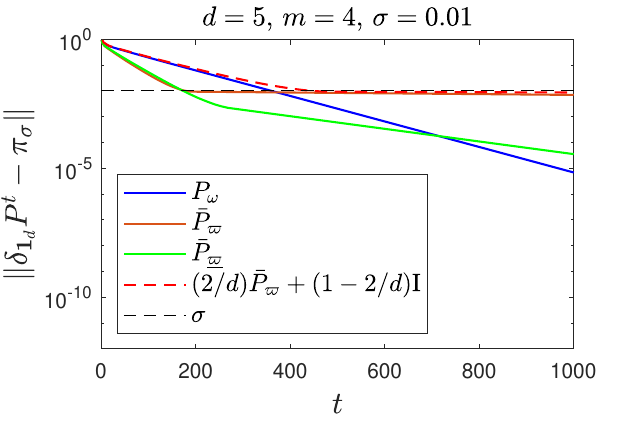}

\includegraphics[scale=0.68]{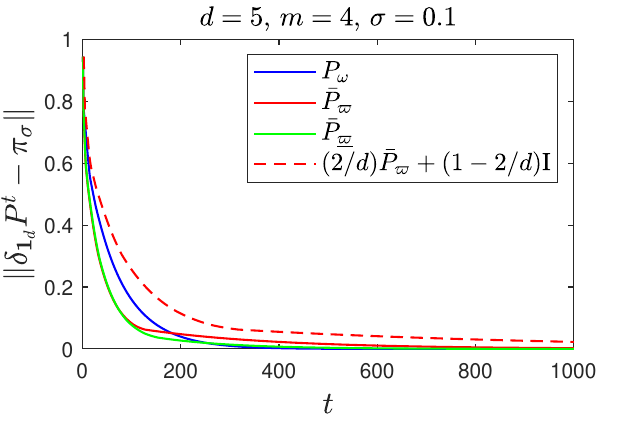}\includegraphics[scale=0.68]{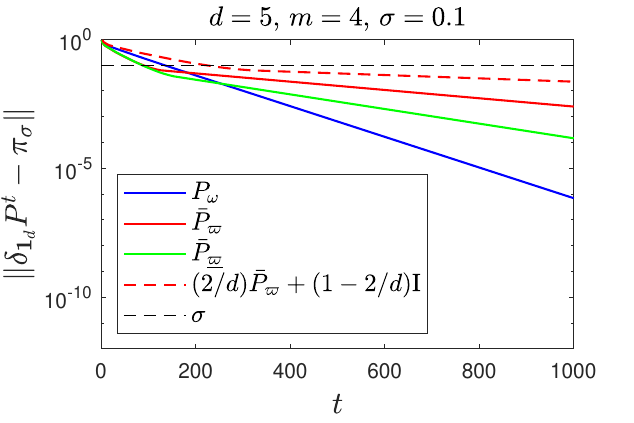}

\caption{Examples \ref{ex1}. Convergence from $\delta_{\mathbf{1}_d}$ to $\pi_\sigma$ of random-scan ($P_\omega$) and locally-weighted ($P_\varpi^\ast$) Markov chains. The two columns show the same plots but with different Y-axis scales. The LHS shows that the transient phase is much faster for $P_\varpi^\ast$ than $P_\omega$ while the RHS shows that $P_\varpi^\ast$ is not efficient in asymptotic regime. It is however very efficient to converge from $\delta_{\mathbf{1}_d}$ to $\pi_\sigma(\cdot\,|\,\Zset)$ (indeed, from Eq. \eqref{eq:tv_cond} $\|\pi_\sigma-\pi_\sigma(\cdot\,|\,\Zset)\|=\sigma$ which is shown in black dotted line on the RHS). This motivates constructions of weights (such as $\underline{\varpi}$, Eq.\eqref{eq:varpi_und}) which do not coerce as much the locally-weighted chain to $\Zset$.\label{fig:hypercube:2}}
\end{figure}

Before moving on to applications, we introduce a second version  of the locally-weighted MCMC algorithm (Algorithm \ref{alg2}), more efficient than Algorithm \ref{alg1} and which can be defined when $\Pfrak$ is a collection of Metropolis-Hastings kernels. Such a situation is often encountered in Bayesian statistics and motivates Algorithm \ref{alg2}.

\section{A locally-weighted MCMC algorithm for \allowbreak Metropolis-Hastings kernels}
\label{sec:5}
In this section, the collection of kernels $\Pfrak$ comprises exclusively Metropolis-Hastings kernels, \ie for all $i\in\{1,\ldots,n\}$, there exist an absolutely continuous Markov kernel $Q_i$, functions $\beta_i:\Xset^2\to(0,1]$ and $\varrho_i:\Xset\to[0,1)$, such that for all $A\in\Xalg$,
\begin{equation}
\label{eq4}
P_i(x,A)=\int_A Q_i(x,\rmd y)\beta_i(x,y)+\delta_x(A)(1-\varrho_i(x))\,,
\end{equation}
where $\beta_i:\Xset\times\Xset\to (0,1)$ is the acceptance probability defined as
$$
\beta_i(x,y)=1\wedge \frac{\pi(y)Q_i(y,x)}{\pi(x)Q_i(x,y)}\qquad \text{and}\qquad \varrho_i(x)=\int_\Xset Q_i(x,\rmd y)\beta_i(x,y)\,.
$$
By construction $P_i$ is $\pi$-reversible, see for example \cite{tierney1998note}. In this particular case, the locally-weighted algorithm (Alg. \ref{alg1}) can equivalently be rewritten as follows: draw $I\sim \varpi(X_t)$ and $X\,|\,I\sim Q_I$ and set $X_{t+1}=X$ with probability
\begin{equation}\label{Eq:new_acc_rate}
{\beta}_i^\ast(X_t,X)=\beta_i(X_t,X)\left(1\wedge \frac{\varpi_i(X)}{\varpi_i(X_t)}\right)
\end{equation}
and set $X_{t+1}=X_t$ otherwise, that is with probability $1-\rho_i(X_t)$. Experienced readers will recognize that the acceptance probability defined at Eq. \eqref{Eq:new_acc_rate} is suboptimal since it is the product of the MH acceptance probability with that of the locally-weighted correction factor. This motivates the introduction of a second locally-weighted Markov chain $\{X_t\}$, only relevant when all the kernels $P_1,\ldots,P_n$ fall into the framework of Eq. \eqref{eq4}.  Denoting by $\bP_\varpi$ the transition kernel of this Markov chain, we define $\bP_\varpi$ as
\begin{multline*}
\bP_\varpi(x,\rmd y)=\sum_{i=1}^n\varpi_i(x)Q_i(x,\rmd y)\left(1\wedge \frac{\pi(y)Q_i(y,x)\varpi_i(y)}{\pi(x)Q_i(x,y)\varpi_i(x)}\right)\\
+\delta_x(\rmd y)\left(1-\sum_{i=1}^n\varpi_i(x)\varrho_i(x)\right)
\end{multline*}
and the transition mechanism $X_t\to X_{t+1}$ is described at Algorithm \ref{alg2}.

\begin{algorithm}
\caption{A second locally-weighted MCMC for MH kernels, transition $X_{t}\to X_{t+1}$}
\label{alg2}
\begin{algorithmic}[1]
\Require $X_t=x\in\Xset$
\State draw $I\sim\varpi(x) \rightsquigarrow i$
\State propose $\tX\sim Q_i(x,\cdot)\rightsquigarrow \tx$ and set $X_{t+1}=\tx$ with probability
\begin{equation}
\label{eq1_bis}
\bar\beta_i(x,\tx):=1\wedge \frac{\pi(\tx)Q_i(\tx,x)\varpi_i(\tx)}{\pi(x)Q_i(x,\tx)\varpi_i(x)}
\end{equation}
and $X_{t+1}=x$ otherwise.
\end{algorithmic}
\end{algorithm}

\begin{prop}
 The Markov transition kernel  $\bP_\omega$ is $\pi$-reversible.
\end{prop}

\begin{proof}
We consider the distribution whose density is given by $\bar\pi(i,x):=\varpi_i(x)\pi(x)$. Let $K$ be the transition kernel of the Markov chain $\{(X_t,I_t),\,t\in\nset\}$ defined by the MH kernels with proposal $(I,X)|(I_t,X_t)\sim \varpi_I(X_t)\otimes Q_I(X_t,X)$. Then $K$ is $\bar\pi$-reversible by construction and thus for each $(i,i')\in\{1,\ldots,n\}$ and any $(A,B)\in\Xalg\otimes\Xalg$,
$$
\int_A \bar\pi(i,\rmd x)K((i,x),(i',B))=\int_B \bar\pi(i',\rmd x)K((i',x),(i,A))
$$
but noting that $K((i,x),\cdot)$ does not depend on $i$, and integrating out over all the possible $i$ and $i'$ on both sides we have that
$$
\int_A \pi(\rmd x)\sum_{i'=1}^n K(x;(i',B))=\int_B\pi(\rmd x)\sum_{i=1}^n K(x;(i,B))
$$
and the proof is completed by noting that $\sum_{i'=1}^n K(x;(i',B))=\bP_\varpi(x,B)$.
\end{proof}

In Algorithm \ref{alg1}, a proposal $\tX$ can be rejected (1) because of the non-zero diagonal mass of $P_i$ or (2) because of the correction step necessary to keep the locally-weighted algorithm $\pi$-invariant. In contrast, a proposal $\tX$ in Algorithm \ref{alg2} faces only one accept/reject step. This naturally induces a Peskun ordering between the Markov kernels $P_\varpi^\ast$ and $\bP_\varpi$ which yields the following result which indicates that when $P_1,\ldots,P_n$ are MH kernels, the locally-weighted MCMC of Algorithm \ref{alg2} should be preferred to Algorithm \ref{alg1}, when the efficiency is measured by the asymptotic variance.

\begin{prop}
\label{prop:peskunMH}
For any square integrable function  $f\in\Ltwo(\pi)$ and any $\varpi\in\Delta_{n-1}^\ast$,  we have
\begin{equation}
\label{eq5}
\var(f,P_\varpi^\ast)\geq \var(f,\bP_\varpi)\,.
\end{equation}
\end{prop}

\begin{proof}
Note that for all $x\in\Xset$ and any $A\in\Xalg$,
\begin{enumerate}[(i)]
\item the Markov subkernels associated to $P_\varpi^\ast$ and $\bP_\varpi$ write
\begin{equation*}
\left\{
\begin{array}{l}
P_\varpi^\ast(x,A\backslash\{x\})=\sum_{i=1}^{n}\varpi_i(x)\int_A Q_i(x,\rmd y) \beta_i^\ast(x,y)\,,\\
\bP_\varpi(x,A\backslash\{x\})=\sum_{i=1}^{n}\varpi_i(x)\int_A Q_i(x,\rmd y) \bar\beta_i(x,y)
\end{array}
\right.
\end{equation*}
\item for all $i\in\{1,\ldots,n\}$ and for $(x,y)\in \Xset^2$,
\begin{multline*}
\beta_i^\ast(x,y)=\left\{1\wedge \frac{\pi(y)Q_i(y,x)}{\pi(x)Q_i(x,y)}\right\}\left\{1\wedge\frac{\varpi_i(y)}{\varpi_i(x)}\right\}
\\\leq
1\wedge \frac{\pi(y)Q_i(y,x)\varpi_i(y)}{\pi(x)Q_i(x,y)\varpi_i(x)}=\bar\beta_i(x,y)\,,
\end{multline*}
since for any positive real numbers $(a,b)$, $(1\wedge a)(1\wedge b)<1\wedge ab$.
\end{enumerate}
Combining (i) and (ii), we obtain that $\bP_\varpi$ dominates $P_\varpi^\ast$ in the Peskun ordering sense. Since $P_\varpi^\ast$ and $\bP_\varpi$ are both $\pi$-reversible, the inequality \eqref{eq5} follows by applying Theorem 4 from \cite{tierney1998note}.
\end{proof}

\begin{rem}
Let $\Xset$ be countable. Then, the locally-balanced MCMC algorithm proposed in \cite{zanella2020informed} can be seen as a particular instance of Algorithm \ref{alg2}. Assume that each $x\in\Xset$ has exactly $n$ neighbors $\mathsf{N}(x)=\{\mathsf{N}_1(x),\ldots,\mathsf{N}_n(x)\}\subset \Xset$ to which transition from state $x$ is possible. Let, for each $i\in\{1,\ldots,n\}$, $Q_i(x,\,\cdot\,)=\delta_{\mathsf{N}_i(x)}$. Then setting $\varpi_i(x)=g(\pi(\mathsf{N}_i(x))/\pi(x))/\sum_{\ell=1}^n g(\pi(\mathsf{N}_\ell(x))/\pi(x))$ for any function $g:\rset^+\to\rset^+$ such that $g(t)=tg(1/t)$ in Algorithm \ref{alg2} defines a locally-balanced algorithm in the sense of \cite{zanella2020informed}. The motivation behind the two samplers is similar: they are both especially relevant when the noise-level in vanishing-noise distribution decreases as they are able to exploit sparse and filamentary structures in $\pi_\sigma$. The mean to reach that goal is however different: a single proposal $Q$ is turned into a locally-balanced one in \cite{zanella2020informed} by the transformation $Q\mapsto g(\pi(y))Q(x,y)/\sum_{z\in\Zset}g(\pi(z))Q(x,z)$ while, in Algorithm \ref{alg2}, a uniform mixture of proposals is turned into a locally-weighted one by the transformation $(1/n)\sum_{i=1}^nQ_i(x,y)\mapsto \sum_{i=1}^n \varpi_i(x)Q_i(x,y)$. The locally-balanced algorithm is appealing since for several important models, its acceptance ratio is, by the very design of the locally-balance proposal, shown to tend to one when $d\to\infty$. The unfortunate byproduct is that this construction also makes simulating from a locally-balanced proposal only feasible when $\Xset$ is countable.
\end{rem}

\begin{rem}
It is possible to construct yet another Markov kernel based on the same proposal scheme (i) $I\sim \varpi(X_t)$ and (ii) $X|I\sim Q_I(X_t,\cdot)$ of Algorithm \ref{alg2} but using the acceptance probability defined by
$$
\tilde\beta_i(X_t,X):=1\wedge \frac{\pi(X)\sum_{j=1}^n \varpi_j(X)Q_j(X,X_t)}{\pi(X_t)\sum_{j=1}^n \varpi_j(X_t)Q_j(X_t,X)}\,.
$$
This choice specifies a \textit{meta}-Metropolis-Hastings Markov chain which is thus $\pi$-reversible. In the case of position-independent weight function, $\omega\in\Delta_{n-1}$ we know from \cite[Proposition 5]{tierney1998note} that the later Markov chain is more efficient, again from the asymptotic variance perspective, than Algorithm \ref{alg2}. However, no such thing can be proven when $\varpi\in\Delta_{n-1}^\Xset$ and since $\tilde\beta_i$ is typically more expensive to calculate than $\bar\beta_i$ as it requires evaluating $2n$ proposals at each iteration, we do not consider this algorithm any further.
\end{rem}

\section{The kernel selection probability function}
\label{sec5:weightfunction}
To implement a locally-weighted algorithm (Alg. \ref{alg1} or Alg. \ref{alg2}) a choice has to be made concerning the  weight function $\varpi\in\Delta_{n-1}^\Xset$. Contrarily to Example \ref{ex1}, a reasonable weight function is usually not apriori known. There are many possible heuristics and the one we follow in this paper is to set for $x\in\Xset$ the kernel selection probability function as
\begin{equation}
\label{eq:varpi}
\varpi_i(x)\propto \esp_i\left[g(\pi(\tX))\,\bigg|\,x\right]\,,\qquad i\in\{1,\ldots,n\}\,,
\end{equation}
where the conditional expectation is w.r.t. $Q_i(x,\,\cdot\,)$ and $g:\rset^+\to\rset^+$. Essentially, the choice of $g$ modulates the sampling strategy and in particular controls the tradeoff between attempting easy moves that are likely to be accepted and more risky ones. With $g=g_0$ the constant function, $\varpi_i(x)=\omega_i=1/n$ for any $x\in\Xset$ and all $i\in\{1,\ldots,n\}$. But taking $g$ as the identity function, higher weights are assigned to kernels leading, on average, to states whose density is higher than states that would be generated using $g_0$. This feature is exacerbated when $g$ is the square function or any other faster growing function. On the contrary, if $g$ is the inverse function, proposed states have, on average, a lower density than states that would be generated using $g_0$. In this paper, we consider $g(x)=x$ which is sensible in the context of noise vanishing distributions where one is first and foremost interested in efficient sampling on the subspace or manifold containing the bulk of the probability mass. Of course, the weight $\varpi_i(x)$ is all the more informative given that $\var_i[g(\pi(\tX))\,|\,x]$ is small. This is the case when the proposal distributions $Q_i$ have a low variance and/or  a low-dimensional support, which usually requires $n$ to be large.

Apart from certain situations in which the state space is countable, pointwise evaluation of $\varpi_i(x)$ in Eq. \eqref{eq:varpi} is rarely feasible. When the proposal distributions $Q_i$ are random walk kernels, i.e. have the form $Q_i(x,\rmd y)=R_i(y-x)\rmd y$, the following approximation can be considered
$$
\widehat{\varpi}_i^{(L)}(x)=\frac{1}{L}\sum_{\ell=1}^n \pi(x+ Z_\ell^{(i)})\,,\qquad (Z_1^{(i)},\ldots,Z_L^{(i)})\sim_{iid}R_i\,,\qquad i\in\{1,\ldots,n\}\,,
$$
for some $L\in\nset$. For example, if $Q_i(x,\cdot)=\norm(x,\Sigma_i)$, then $R_i=\norm(0,\Sigma_i)$. It can be readily checked that the particles $(Z_1,\ldots,Z_n)$ are exogenous random variables as they do not depend on the past history of the Markov chain nor from its current state. Readers familiar to the Multiple-Try Metropolis algorithm \citep{liu2000multiple} will see a resemblance with the resulting algorithm. By replacing $\varpi$ by $\widehat{\varpi}^{(L)}$ in $\bP_\varpi$, we essentially turn a time-homogeneous Markov chain into a time-inhomogeneous one, since each transition is carried out by a kernel which is parameterized by the $2n$ sets of particles $(Z_1,\ldots,Z_L)$ which help approximating $\varpi_1(x),\ldots,\varpi_n(x)$, necessary to draw $I$, and $\varpi_{1}(\tX),\ldots,\varpi_n(\tX)$, necessary to compute the acceptance probability $\bar\beta_i(x,\tilde{X})$. Since the resulting kernel does not depend on a particular choice for the weight function $\varpi$, we simply denote it by $\bP_L$ where $L\in\nset$ is the number of particles used to estimate $\varpi$.

The question whether or not the time-inhomogeneous Markov chain inherits the same convergence properties, LLN and CLT from the time-homogeneous Markov chain is hard to answer in all generality. However, since the kernels vary only through the set of particles used to approximate $\varpi$ and that they are all $\pi$-reversible MH kernels, we expect this approximation step to not deteriorate significantly the convergence of the Markov chain. For more details on convergence of time-inhomegeneous Markov chain, see \cite{douc2004quantitative} and \cite{saloff2007convergence} and for CLT see \cite{sethuraman2004martingale}. Another option is to use the same set of particles for each $Q_i$, throughout the algorithm so as to retrieve a time-homogeneous chain.

\section{Numerical Examples}
\label{sec:6}
In this section, we consider three sampling problems defined on general state spaces. For each, the bulk of the probability mass of the distribution of interest displays sparse and manifold-like features. Since the full conditional distributions of the models considered are not straightforward, a Gibbs sampler cannot be implemented. With some abuse of notations, we will keep the random-scan terminology to refer to the algorithm which is usually known as Metropolis-within-Gibbs with state-independent selection weight. Our goal is to compare such a random-scan procedure with a corresponding locally-weighted one. In all examples, we will consider mixing a collection of proposal kernels rather than $\pi$-reversible kernels directly, hence the locally-weighted procedure will be as specified at Algorithm \ref{alg2}. For each example, we specify the collection of proposal kernels $Q_1,Q_2,\ldots,Q_n$, the random-scan weight $\omega\in\Delta_{n-1}$ and the weight function $\varpi:\Xset\to \Delta_{n-1}$ used by the locally-weighted algorithm. When constructing a reasonable weight function is too involved, we use the technique detailed in Section \ref{sec5:weightfunction} to estimate the local weights, on the fly.

Random-scan and locally-weighted strategies are compared according to their time to convergence from some initial distribution $\pi_0$ as well as through the asymptotic variance of the empirical average of some test functions over the sample path of both Markov chains, started at stationarity. More precisely:
\begin{itemize}
\item The convergence in distribution is assessed by estimating the Kullback-Leibler divergence (KL) between $\pi$ and $\pi_0 P^t$, that is the distribution of the Markov chain after $t$ transitions and started at $\pi_0$. The KL divergence is estimated using a nearest neighbor entropy estimator, developed in  \cite{chauveau2013smoothness} and \cite{chauveau2020entropy}.
\item The asymptotic variance of a Markov chain $P$ used to estimate $\pi f$ is estimated by simulating a population of MCMC estimators of $\{\widehat{\pi f}\}_1,\{\widehat{\pi f}\}_2,\ldots$  obtained through the simulation in parallel of i.i.d. Markov chains started at $\pi$ when possible, or from a proxy of $\pi$ otherwise, for $T$ iterations. The time horizon $T$ is problem specific and is set large enough so that the MCMC estimators reach their asymptotic normal regime. The asymptotic variance $\var(P,f)$ is thus estimated by
    $$
    T\widehat{\var}\left(\{\widehat{\pi f}\}_1,\{\widehat{\pi f}\}_2,\ldots\right)\,,
    $$
    where $\widehat{\var}$ denotes the empirical variance operator.
\end{itemize}
Studying simultaneously the transient and  stationary regimes of MCMC algorithms is important for applications in Statistics and Machine Learning which typically focus on the Root Mean Square Error (RMSE) of estimators.

\subsection{Mixture of Gaussians}
This example introduces a non-straightforward mixture of Gaussian distribution which can be seen as an extension of Example \ref{ex1} to general state spaces. In particular, the distribution is parameterized by the dimension of the state space $d$ and a noise level $\sigma$.  With such a synthetic example, it is possible to study the algorithms in noise vanishing regime $\sigma \downarrow 0$ and/or in high dimensional settings $d\uparrow \infty$.

\begin{exa}
\label{ex7}
Let $\pi_{d,\sigma}$ be the $d$-dimensional mixture of $d$ Gaussian distributions, parameterized by a noise level $\sigma>0$ and defined as
\begin{multline}
\label{eq:ex7}
\pi_{d,\sigma}=\frac{1}{d}\sum_{i=1}^d\norm_d(\mu_i,\Sigma_i)\,,\\
\mathrm{with}\quad \mu_1:=\0_d\quad\mathrm{and\;for\; all}\quad
 i\in\{2,\ldots,d\}\,,
\mu_{i}:=\mu_{i-1}+\gamma_i\,,\quad \\
\mathrm{with}\quad \gamma_i=(\gamma_i(1),\ldots,\gamma_i(d))^\mathrm{T}\,,\quad \gamma_i(j)=\frac{1}{\sigma}\Phi^{-1}(0.9)\left({\mathds{1}}_{\{j=i-1\}}+{\mathds{1}}_{\{j=i+1\}}\right)\,,\\
i\in\{1,\ldots,d\}\,,\quad \Sigma_i=\mathrm{diag}([\1_{i-1}\T\,, 1/\sigma^2\,,\1_{d-i}\T])\,,
\end{multline}
with the convention that for any real number $c$, $\mathbf{c}_d\in\rset^d$ is the constant vector equal to $c$ and $\mathbf{c}_0=\{\emptyset\}$ and $\Phi$ is the standard normal cdf. This definition implies that, for any $d\in\nset$, $\pi_{d,\sigma}$ features a more pronounced sparse and filamentary structure  as the noise level $\sigma$ decreases, see Figure \ref{fig:cross}. Even though sampling i.i.d. draws from $\pi_{d,\sigma}$ is straightforward, the random-scan algorithm $P_\omega$ and the locally-weighted algorithms (more precisely the kernel $\bP_\varpi$ for Algorithm \ref{alg2}  and the kernel $\bP_L$ where the weight function $\varpi$ is estimated on the fly with $L$ particles, see Section \ref{sec5:weightfunction}) are implemented on a computer to produce (dependent) draws from $\pi_{d,\sigma}$.  Given that this example may be seen as a continuous counterpart to the sampling problem of Example \ref{ex3},  one may reasonably expect  $\bP_\varpi$  to dominate $P_\omega$, at least when defined with a relevant choice of weight function $\varpi$. Moreover, the performance of $\bP_\varpi$, relative to $P_\omega$, should increase as $\pi_{d,\sigma}$ gets more sparse and filamentary, i.e.  as $d$ increases and $\sigma$ decreases. The magnitude of improvement remains to be assessed empirically and it is particularly interesting to study how it compares to the theoretical results obtained for Example \ref{ex1} in the noise-free regime at Propositions \ref{prop:hypercube:1}, \ref{prop:hypercube:2} and \ref{prop:hypercube:3}. On the practical side, it is interesting to study how $\bP_L$ compares to $P_\omega$ since, by contrast to $\bP_\varpi$, both kernels require the same initial knowledge on $\pi_{d,\sigma}$ to be implemented. Another relevant question is to assess whether estimating the weights does not cause an overwhelming computational burden to $\bP_L$, relatively to $P_\omega$. Elements of response to those questions based on numerical experiments are provided.
\end{exa}

\begin{figure}[h]
\centering
\includegraphics[scale=0.65]{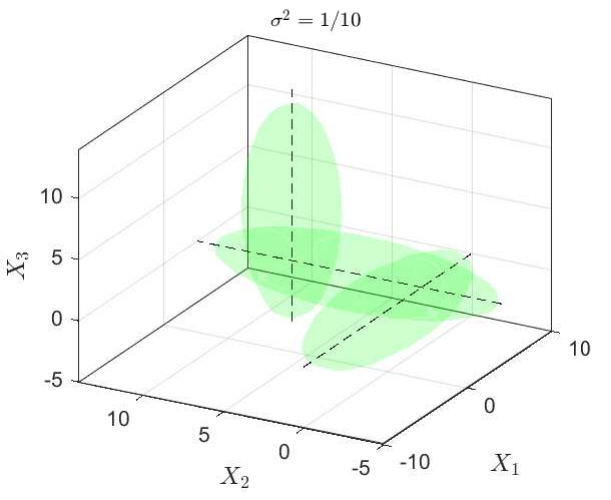}\includegraphics[scale=0.65]{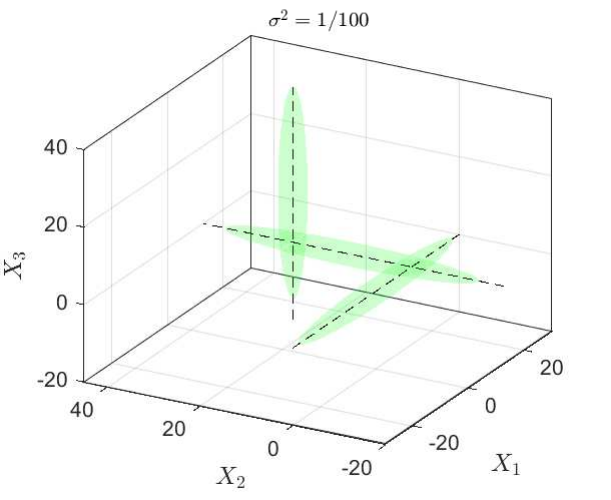}
\includegraphics[scale=0.65]{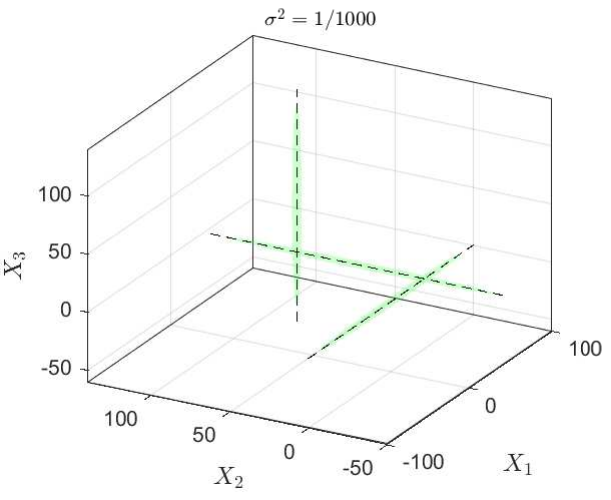}
\caption{Example \ref{ex7}. Representation of the distribution $\pi_{3,\sigma}$, for three parameters $\sigma^2\in\{0.1,0.01,0.001\}$.  The three-dimensional ellipsoids cover 95\% of the probability mass of each mixture component of $\pi_{3,\sigma}$. The dashed line represents the direction of the eigenvector associated to the largest eigenvalue of the covariance matrix of each component. Note that the axis scale is different in each plot.  \label{fig:cross}}
\end{figure}

\begin{table}
\centering
\begin{tabular}{ccccc}
\hline
$d=3$ & \multicolumn{4}{c}{$\var(\bP_\varpi,f)/\var(P_\omega,f)$} \\
 \hline
 $\sigma^2$   & $10^{-1}$ & $10^{-2}$ & $10^{-3}$ & $10^{-4}$ \\
    \hline
$f_1$& $0.56$, $0.57*$& $0.37$, $0.54*$ & $0.22$, $0.47*$ & $0.13$, $0.47*$ \\
$f_2$ & $0.71$, $0.65*$ & $0.40$, $0.57*$ & $0.25$, $0.47*$& $0.16$, $0.47*$ \\
$f_3$ & $0.72$, $0.66*$ & $0.40$, $0.57*$& $0.25$, $0.47*$ & $0.16$, $0.47*$\\
\hline
\end{tabular}

\begin{tabular}{ccccc}
\hline
$d=3$ &   \multicolumn{2}{c}{$\var(\bP_{10},f)/\var(P_\omega,f)$} &  \multicolumn{2}{c}{$\var(\bP_{100},f)/\var(P_\omega,f)$}\\
 \hline
 $\sigma^2$    &$10^{-1}$ & $10^{-2}$ & $10^{-1}$ & $10^{-2}$\\
    \hline
$f_1$& $0.65$ & $0.35$ & $0.65$ & $0.34$\\
$f_2$ &   $0.69$ & $0.40$ & $0.70$ & $0.40$\\
$f_3$ &   $0.68$ &$0.40$& $0.70$ &  $0.40$\\
\hline
\end{tabular}
\vspace{.5cm}

\begin{tabular}{cccccccccc}
\hline
$d=5$ & \multicolumn{4}{c}{$\var(\bP_\varpi,f)/\var(P_\omega,f)$} &  \multicolumn{3}{c}{$\var(\bP_{10},f)/\var(P_\omega,f)$} &  \multicolumn{2}{c}{$\var(\bP_{100},f)/\var(P_\omega,f)$}\\
 \hline
 $\sigma^2$   & $10^{-1}$ & $10^{-2}$ & $10^{-3}$ & $10^{-4}$ &$10^{-1}$ & $10^{-2}$ &$10^{-3}$ & $10^{-1}$ & $10^{-2}$\\
    \hline
$f_1$& $0.35*$ & $0.34*$ &  $0.32*$ & $0.30*$ & $0.61*$ & $0.40*$ & $0.34*$ & $0.52*$ & $0.40*$ \\
$f_2$ & $0.44*$ & $0.33*$ & $0.32*$  & $0.29*$ & $0.62*$ & $0.44*$ & $0.35*$ &$0.61*$ & $0.42*$\\
$f_3$ & $0.48*$ & $0.34*$ & $0.33*$ & $0.30*$ & $0.61*$ & $0.45*$ &$0.35*$ & $0.61*$ & $0.42*$\\
\hline
\end{tabular}

\vspace{.5cm}
\begin{tabular}{cccccccccc}
$d=10$ & \multicolumn{4}{c}{$\var(\bP_\varpi,f)/\var(P_\omega,f)$} &  \multicolumn{3}{c}{$\var(\bP_{10},f)/\var(P_\omega,f)$}
 &  \multicolumn{2}{c}{$\var(\bP_{100},f)/\var(P_\omega,f)$}\\
 \hline
 $\sigma^2$   & $10^{-1}$ & $10^{-2}$ & $10^{-3}$ & $10^{-4}$ &$10^{-1}$ & $10^{-2}$ & $10^{-4}$ & $10^{-1}$ & $10^{-2}$\\
    \hline
$f_1$& $0.22*$ & $0.18*$ &  $0.17*$ & $0.17*$ & $0.64*$ & $0.33*$ & $0.15*$ &  $0.44*$&  $0.28*$\\
$f_2$ & $0.25*$ & $0.18*$ &  $0.18*$  & $0.18*$ & $0.61*$ & $0.35*$& $0.16*$ & $0.44*$ &  $0.31*$\\
$f_3$ & $0.43*$ & $0.19*$ &  $0.19*$ & $0.18*$ & $0.66*$ & $0.34*$& $0.16*$ & $0.44*$ &   $0.29*$\\
\hline
\end{tabular}
\caption{Example \ref{ex7}. Ratio of asymptotic variances for three test functions between a locally-weighted kernel  $\bP_\varpi$ with $\varpi$ defined at Eq. \eqref{eq:ex8:weights} or $\bP_L$ for $L\in\{10,100\}$ and the random-scan kernel $P_\omega$ with $\omega=(1/d,\,\cdots\,,\,1/d)\T$. Test functions were defined as $f_1(x)=\mathbf{1}_{\{x_1<0\}}$, $f_2(x)=\|x\|_2^2$ and $f_3(x)=x_2^2$. Results for different model parameters $(d,\sigma^2)\in \{3,5,10\}\times\{1/10,1/100,1/1000,1/10000\}$ are provided. Entries with an asterisk indicate that the variance parameter for the random-scan proposal kernel was set to that tuned by the locally-weighted preliminary run (see paragraph \textit{Proposal kernels} at page 15). The smaller the ratio the better the locally-weighted kernel for estimating $\pi f_i$ ($i\in\{1,2,3\}$), relatively to the random-scan one. All asymptotic variances were estimated using replicas of Markov chains sufficiently long so as to enter a central limit theorem regime, for all three test functions. All chains were started in stationary regime.\label{tab:ex7:avar}}
\end{table}

\begin{figure}
\centering
\includegraphics[scale=0.5]{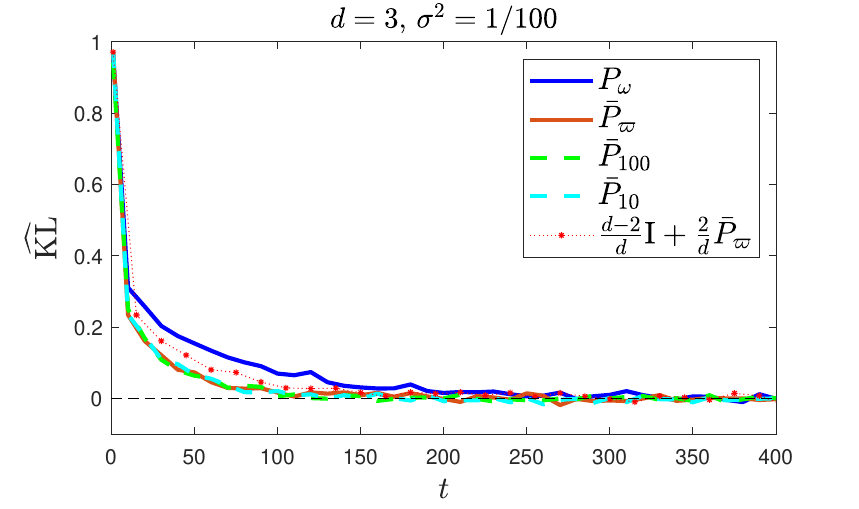}\includegraphics[scale=0.5]{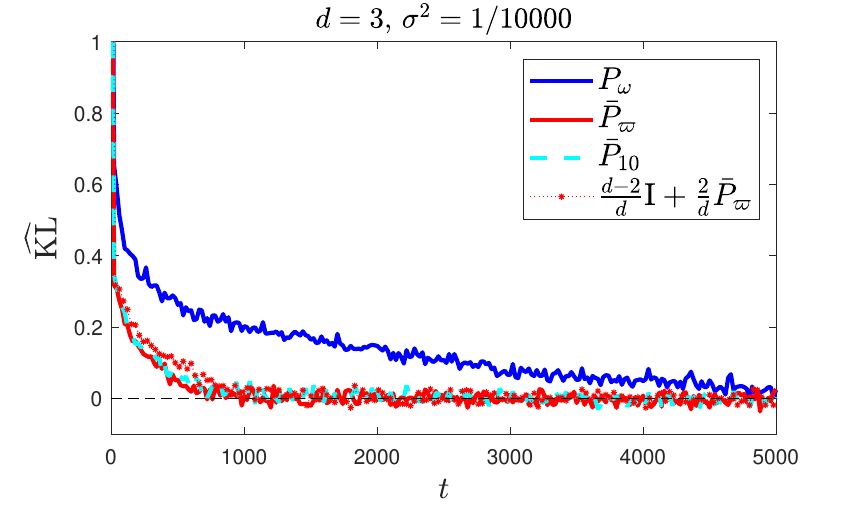}

\includegraphics[scale=0.5]{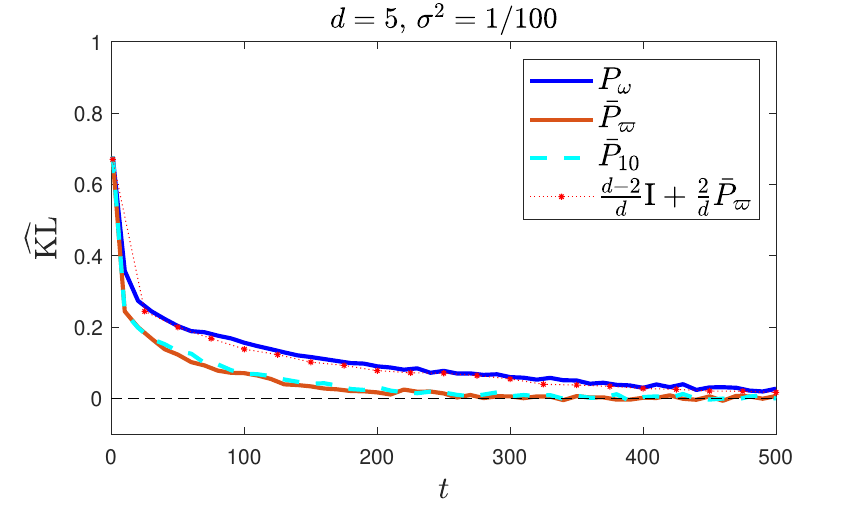}\includegraphics[scale=0.5]{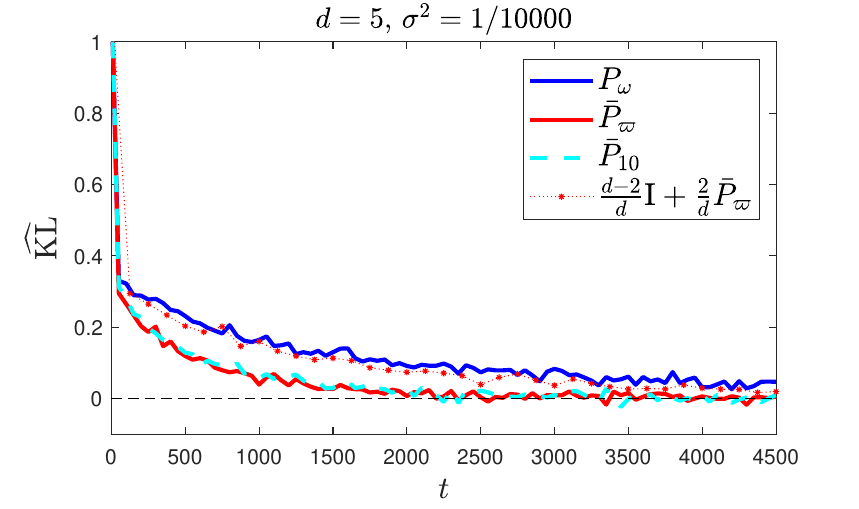}

\includegraphics[scale=0.5]{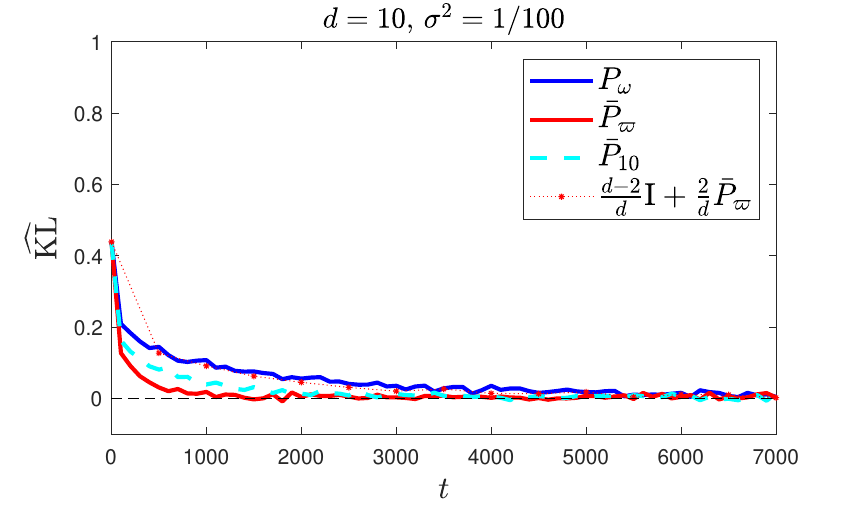}\includegraphics[scale=0.5]{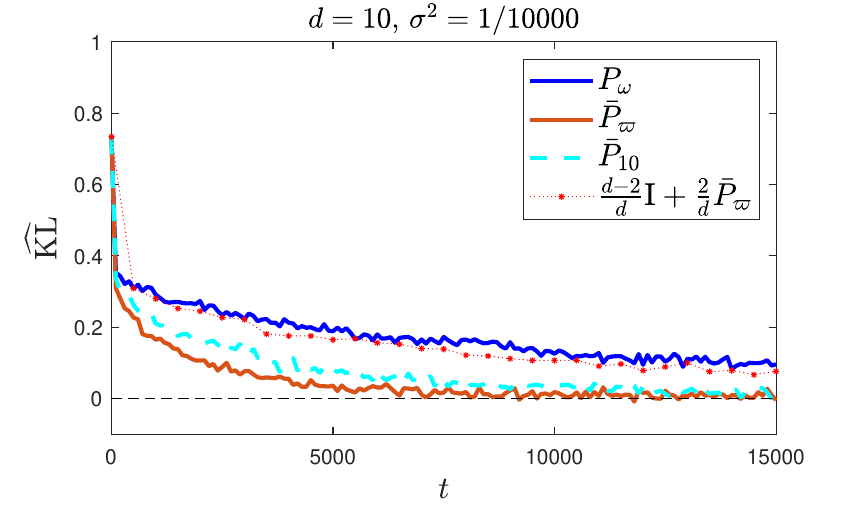}
\caption{Example \ref{ex7}. Distributional convergence of $\{\pi_0 P^t\}$ from $\pi_0=\norm_d(\mu_1,\mathrm{Id}_d)$ towards $\pi_{d,\sigma}$ for $P\in\{P_\omega,\bP_\varpi,\bP_{10}\}$ and a lazy version of $\bP_\varpi$ which only attempts moving using $\bP_\varpi$ with probability $2/d$. The plots estimate $\mathrm{KL}(\pi_0P^t,\pi_{d,\sigma})$, that is  the Kullback-Leibler divergence between $\pi_0 P^t$ and $\pi_{d,\sigma}$ as a function of $t$, for different parameters $(d,\sigma)$. Estimating the local weights (i.e. using $\bP_L$ instead of $\bP_\varpi$) slows down very mildly the convergence of the locally-weighted strategy. For $d=3$, the random-scan proposal kernel variance was adaptively tuned using an adaptive random-scan MCMC algorithm while for $d>3$ the random-scan proposal kernel variance was tuned with the variance parameter estimated from an adaptive locally-weighted MCMC algorithm $\bP_\varpi$. We note that when the same proposal variance is used for the locally-weighted and random-scan chains (as for $d=5$ and $d=10$), the lazy chain $\pi_0\left[(1-2/d)\text{I}+d/2 \bP_\varpi\right]^t$ offers a relatively tight lower bound for $\pi_0 P_\omega^t$, which is in line with the theory of Example \ref{ex1}, see Prop. \ref{prop:hypercube:1}. \label{fig:ex7_cv}}
\end{figure}

\paragraph*{Available proposal kernels.}
Given the symmetry of $\pi_{d,\sigma}$ (see Figure \ref{fig:cross}), the proposal kernels are single-site update Gaussian kernel with a certain variance parameter $\upsilon>0$, which depends on $\sigma$. Therefore, the collection of proposal kernels $Q_1,\ldots,Q_d$ (thus $n=d$) is defined such that
$$
(i,x)\in\{1,\ldots,d\}\times \Xset\,,\qquad Q_i(x,\rmd y)=\phi_1(\rmd y_i\,|\,x_i,\upsilon)\prod_{j\neq i}\delta_{x_j}(\rmd y_j)\,,
$$
where $\upsilon>0$ is the random-walk variance parameter and for any $d\in\nset$, $\phi_d(\,\cdot\,|\,\mu,\Sigma)$ is the $d$-dimensional Gaussian pdf with mean parameter $\mu$ and covariance $\Sigma$. For each noise level parameter $\sigma$ considered, the parameter $\upsilon$ was set for the three methods $P_\omega,\bP_\varpi,\bP_L$ independently using a preliminary adaptive run for each algorithm tuning $\upsilon$ on the fly until the average acceptance probability stabilizes between $0.3$ and $0.4$. However, for $d\geq 5$ we found that the resulting tuned value $\upsilon$ for random-scan ($P_\omega$) did not allow to reach the asymptotically normal regime of the MCMC estimator, in a reasonable time frame.\footnote{ This observation is explained as follows. For any $X_t$ in an area of non-negligible probability and each $i\in\{1,\ldots,d\}$,  with probability $1-1/d$, $\tilde{X}\sim Q_i(X_t,\cdot)$ should be such that $|\tilde{X}_i-X_{t,i}|\ll1$ in order to have a chance of being accepted. Indeed with probability $1-1/d$, $X_t$ belongs to one of those components of the Gaussian mixture which is not stretched along the $i$-th axis and which therefore requires a relatively small perturbation to remain in a non-negligible probability region of the state-space. As a result, the adaptive strategy sets $\upsilon$ to a relatively small value which does not allow the random-scan chain to traverse the state-space efficiently, all the more that $\sigma$ is low.} Thus, for those situations, the random-scan was implemented using the parameter $\upsilon$ adaptively tuned by the locally-weighted adaptive run $\bP_\varpi$. To emphasize this nuance, results obtained in that setup were reported with an asterisk at Tables \ref{tab:ex7:avar}.

\paragraph*{Weight function.}

By symmetry of $\pi_{d,\sigma}$, the optimal selection probability for the random-scan algorithm is uniform on $\{1,\ldots,d\}$, i.e. $\omega\propto\mathbf{1}_d$.

The locally-weighted weight function should give larger probability to update the component of the Markov chain corresponding to the one or two direction(s) along which $\pi_{d,\sigma}$ varies very slowly (relatively to the other directions), i.e. the direction along which $\pi_{d,\sigma}$ stretches out on. More formally, it is defined as
\begin{equation}
\label{eq:ex8:weights}
(i,x)\in\{1,\ldots,d\}\times \Xset\,,\qquad\varpi_i(x):\propto\sqrt{\phi_d(x\,|\,\mu_i,\Sigma_i)+\frac{1}{d^4}\max_{1\leq i\leq d}\phi_d(x\,|\,\mu_i,\Sigma_i)}\,,
\end{equation}
where the symbol $\propto$ means that for all $x\in\Xset$, the entries $\{\varpi_i(x)\}_{i=1}^d$ are normalized, i.e. $\sum_{i=1}^d \varpi_i(x)=1$. The rationale behind Eq. \eqref{eq:ex8:weights} is that the first term is likely to be negligible for all but one direction (or two whenever $x\in\Xset$ is in an area where two components are roughly equally likely). Without the second term of order $d^{-4}$, the chain would  be nearly reducible since for each $x\in\Xset$ at most two components have a reasonable chance to be updated. The second term inside the square root is set so as to give a chance to also update the other components. By contrast with $P_\omega$ which gives those components a $1/d$ chance,  $\bP_\varpi$  assigns those components a selection chance which scales with $1/d^2$. This can be seen as a tradeoff between the random-scan strategy and a deterministic one which would, in the limit $\sigma\to 0$, systematically pick the one (or two) components on which $\pi_\sigma$ locally stretches out on, thereby failing the sampler for $\sigma>0$.

\paragraph*{Results}

\begin{itemize}
\item\textbf{In terms of asymptotic efficiency}: Table \ref{tab:ex7:avar} reports the ratio of asymptotic variances of the MCMC estimator of $\pi {f}_i$, for three test functions, between a locally-weighted MCMC ($\bP_\varpi$ or $\bP_L$) and its random-scan counterpart $P_\omega$. The test functions are defined as $f_1(x)=\mathbf{1}_{\{x_1<0\}}$, $f_2(x)=\|x\|_2^2$ and $f_3(x)=x_2^2$, for all $x\in\rset^d$. Results point out that a locally-weighted algorithm allows to reduce significantly the asymptotic variance characterizing the random-scan strategy. As expected and in line with the analysis carried out at Example \ref{ex1}, $\bP_\varpi$ outperforms $P_\omega$  more and more significantly when $\sigma$ decreases and $d$ increases. In particular, for each $d\in\{3,5,10\}$, the ratios $\var(f_i,P_\omega)/\var(f_i,\bP_\varpi)$ ($i\in\{1,2,3\}$) tend to stabilize to a limiting value as $\sigma$ decreases and, more quantitatively, it can be checked that the limiting value is always less than $2/d$, as is the case for Example \ref{ex1} in the small noise regime (see Propositions \ref{prop:hypercube:1} and \ref{prop:hypercube:3}). Recall that the proposal kernel variance plays a huge role in those results. It may be noted that, for $d=3$, results obtained by the random-scan kernel tuned with the proposal kernel variance estimated on the fly with an adaptive rule are very poor compared to those obtained by the same mechanism but using the locally-weighted kernel. For reasons mentioned in the footnote of page 13, when $d>3$ the adaptively tuned proposal variance obtained by the random-scan is too small to observe convergence of the random-scan MCMC estimator in a reasonable time frame and results are obtained by comparing $P_\omega$ and $\bP_\varpi$ both fitted with the locally-weighted proposal variance. We thus believe that even though random-scan is outperformed by the locally-weighted kernel in each scenario, results obtained with $d\in\{5,10\}$ are actually flattering for the random-scan approach. Indeed, without the locally-weighted exploration strategy the proposal variance of $P_\omega$ could simply not have been tuned efficiently. Perhaps the most encouraging observation for a practical locally-weighted implementation, is that replacing $\bP_\varpi$ by $\bP_L$ hardly deteriorates the improvement, at least  if the noise level is sufficiently low, relative to $d$. As expected,  the performance of $\bP_L$ increases with $L$ the number of auxiliary particles but it is interesting to note that in the small noise regime few particles are indeed necessary to mimic the canonical choice $\varpi$ of  Eq. \eqref{eq:ex8:weights}: for instance $\bP_\varpi$ and $\bP_{10}$ are comparable when $d=5$, $\sigma^2=1/1000$ and $\bP_{10}$ even dominates $\bP_\varpi$ when $d=10$, $\sigma^2=1/10000$. Indeed, the particular choice of weight function is by no mean optimal.

\item \textbf{In terms of convergence time}: An empirical convergence analysis is carried out and Figure \ref{fig:ex7_cv} reports the results. While it is difficult to estimate the asymptotic rate of convergence from such experiments, those results may allow to gain insight about the transient phase, and thus the mixing time, of both random-scan and locally-weighted Markov chains. All the MCMC algorithms were initialized with $\pi_0=\norm_d(\mu_1,\mathrm{Id}_d)$ and Figure \ref{fig:ex7_cv} illustrates how fast the Markov chains traverse through the different ``edges'' of the distribution.  Interestingly, in the small noise regime and when both locally-weighted and random-scan are tuned with the same proposal variance, the convergence is sped up by a factor $d/2$ when using the locally-weighted algorithm instead of the random-scan.  This observation may be linked to the results obtained in the context of Example \ref{ex1} and Proposition \ref{prop:hypercube:1} in particular. However, we surmise that for each $d$ there exists a certain noise threshold $\sigma_d^\ast$ above which there is no benefit in using a locally-weighted algorithm, see Figure \ref{fig:ex_8_cv_2} for $d=3$ and $\sigma_3^\ast\approx 1/10$. Of course, $\sigma_d^\ast$ should increase with $d$.  Finally, regarding computational considerations, we have the following ordering on the computational cost $\text{cost}(P_\omega)<\text{cost}(\bP_\varpi)<\text{cost}(\bP_{L})$, the later being an increasing function of $L$, see Table \ref{tab:cpu} for quantitative comparisons. Figure \ref{fig:ex_8_cv_3} shows that when $\pi_{d,\sigma}$ is not sufficiently sparse and filamentary, the convergence of $\bP_{10}$ is slower (in wall-clock time) than $P_\omega$ (see the case $d=3$, $\sigma^2=1/100$). However, we stress that with a fully parallel architecture, the computational cost of $\bP_L$ is equivalent to $\bP_1$. The main conclusion of this experiment is that in the small-noise regime, the practical locally-weighted algorithm $\bP_L$ scales much better than its random-scan counterpart with the dimension $d$ of this sampling problem, taking into account the computational time and noting that, by contrast to $\bP_\varpi$, both $\bP_L$ and $P_\omega$ do not need any additional information on $\pi_{d,\sigma}$.
\end{itemize}

\begin{figure}
\centering
\includegraphics[scale=0.5]{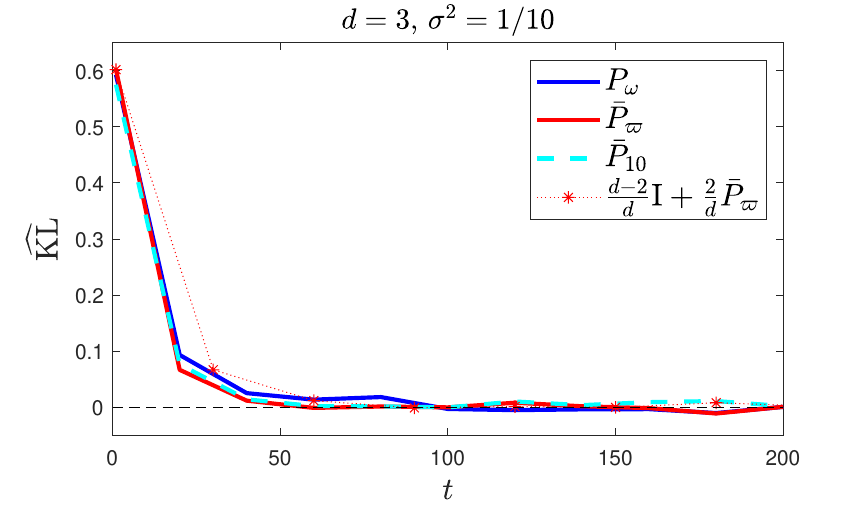}\includegraphics[scale=0.5]{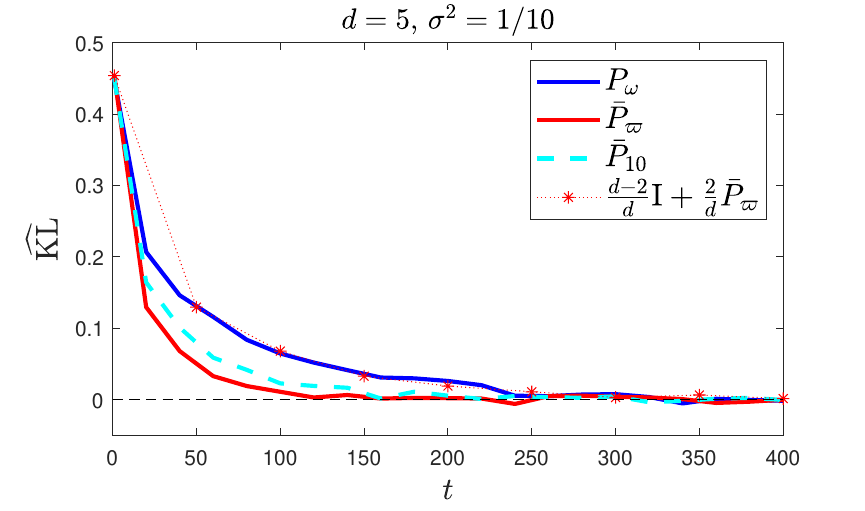}
\includegraphics[scale=0.5]{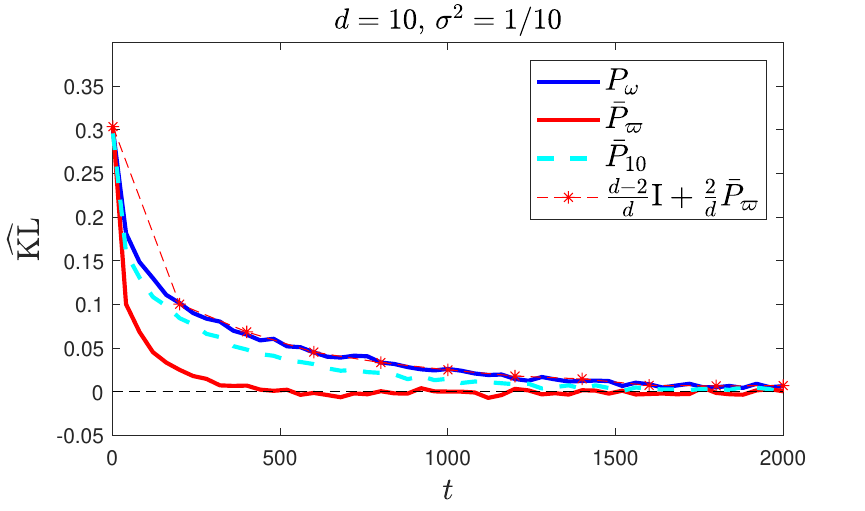}
\caption{Example \ref{ex7}. Distributional convergence of $\{\pi_0 P^t\}$ from $\pi_0=\norm_d(\mu_1,\mathrm{Id}_d)$ towards $\pi_{d,\sigma}$ for a given noise level $\sigma^2=1/10$ when $d$ increases. The proposal variance of $P_\omega,\bP_\varpi$ and $\bP_{10}$ were tuned independently according to a preliminary adaptive MCMC run for each kernel and for each $d\in\{3,5,10\}$. We see that the improvement of $\bP_\varpi$ relative to $P_\omega$ increases with $d$. However, while $\bP_{10}$ is similar to $\bP_\varpi$ for $d=3$, as $d$ increases, $\bP_{10}$ does not keep up with $\bP_\varpi$ and is in fact closer to $P_\omega$ when $d=10$. This means that $L=10$ particles are not enough to retain the locally-weighted advantage and that the estimated weight function is essentially uniform.   \label{fig:ex_8_cv_2}}
\end{figure}

\begin{table}
\centering
\begin{tabular}{ccccc}
\hline
  $d$ & $\mathcal{T}(\bP_{\varpi})/\mathcal{T}(P_{\omega})$ &$\mathcal{T}(\bP_{1})/\mathcal{T}(P_{\omega})$& $\mathcal{T}(\bP_{10})/\mathcal{T}(P_{\omega})$ & $\mathcal{T}(\bP_{100})/\mathcal{T}(P_{\omega})$\\
  \hline
  $3$ &  $1.60$ & $2.73$ & $2.88$ & $3.52$\\
  $5$ &  $1.50$ & $2.29$ & $2.45$ & $6.21$\\
  $10$ & $1.39$ & $1.98$ & $3.37$ & $9.63$\\
  \hline
\end{tabular}
\caption{Example \ref{ex7}. Comparison of the runtime of a single iteration of  the locally-weighted algorithms $\mathcal{T}(\bP_{\varpi})$, $\mathcal{T}(\bP_{1})$, $\mathcal{T}(\bP_{10})$ and $\mathcal{T}(\bP_{100})$, relatively to a single iteration of  the random-scan $\mathcal{T}(P_{\omega})$.\label{tab:cpu}}
\end{table}

\begin{figure}
\centering
\includegraphics[scale=0.5]{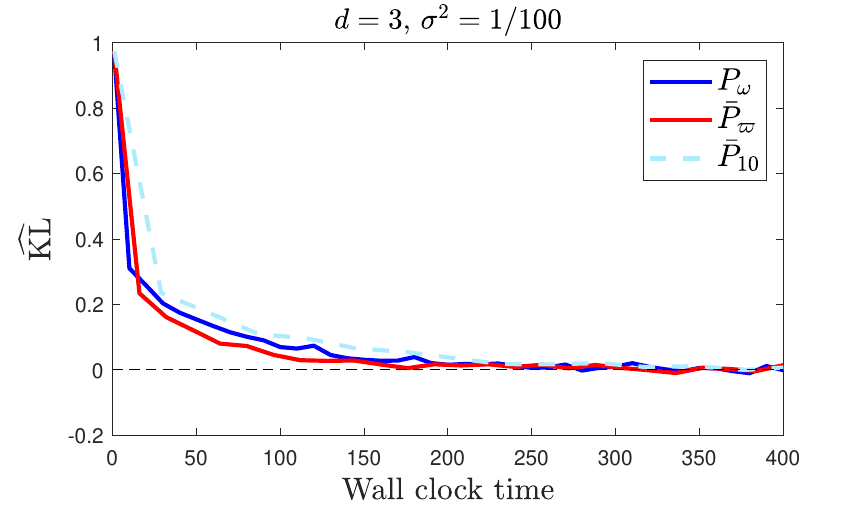}\includegraphics[scale=0.5]{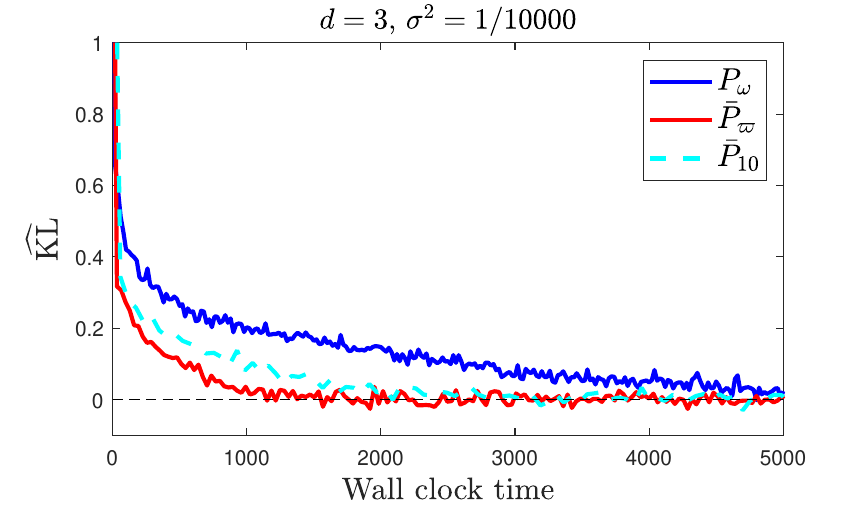}
\includegraphics[scale=0.5]{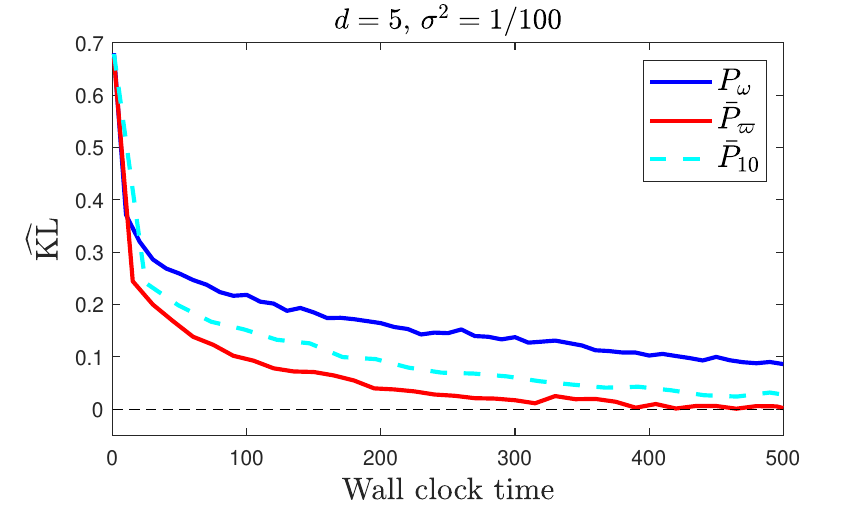}\includegraphics[scale=0.5]{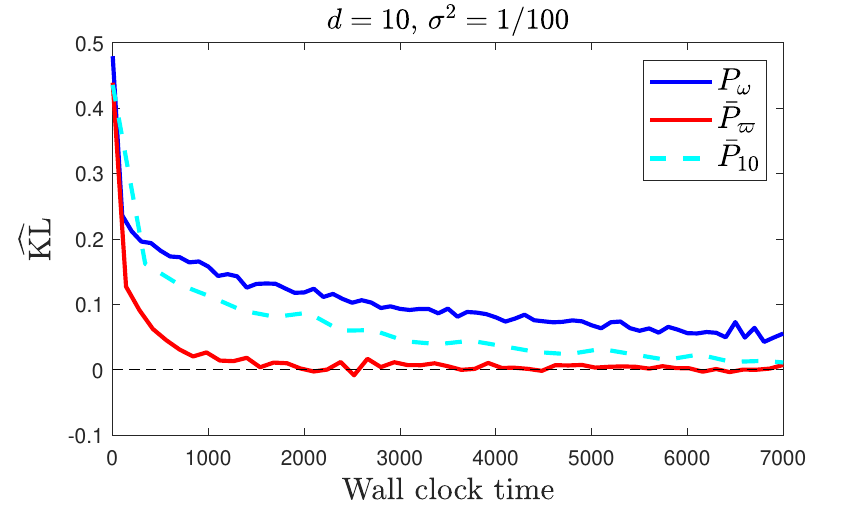}
\caption{Example \ref{ex7}. Distributional convergence of $\{\pi_0 P^t\}$ from $\pi_0=\norm_d(\mu_1,\mathrm{Id}_d)$ towards $\pi_{d,\sigma}$ in time-normalized experiments. In those experiments, $\bP_\varpi$ is always best but it gets better relatively to $P_\omega$, when $\sigma$ diminishes while $d$ is constant  (first row)  and when $d$ increases while $\sigma$ is constant (second row). The same is not true for $\bP_{10}$: first, it may be worse than $P_\omega$ if the ratio $d/\sigma$ is not large enough (upper left-hand corner). Second, when $d$ increases with $L$ and $\sigma$ constant, it is not clear whether or not the improvement of $\bP_L$ relative to $P_\omega$ increases (second row).
\label{fig:ex_8_cv_3}}
\end{figure}

\subsection{Sound source localization}
\begin{exa}\label{ex9}
This example is borrowed from \cite{forbes2021approximate}. We aim at inferring the unknown position $\theta=(\theta_1,\theta_2)\in\Theta=[-a,a]^2$  $(a>0)$ of the source of a sound based on the signal recorded by two pairs of microphones: one pair at $\mu_1=(-b,0)$ and $-\mu_1=(b,0)$ and the other one at $\mu_2=(0,-b)$ and $-\mu_2=(0,b)$ for some $b\in(0,a)$. The prior on $\theta$ is the uniform distribution on $\Theta$. To localize the position of the source, the Interaural Time Difference (ITD) is used so that the observed data is the recorded ITD. For a given pair of microphones, the ITD is the time lag between the detection of the sound signal by both microphones of that pair. In our setup, we assume that the ITD is measured in equal probability with either pairs. Hence knowing that sound waves propagate in the air at a constant speed denoted by $c_s$, the ITD (for the pair $(\mu_1,-\mu_1)$ and given that the sound source location is $\theta$ ) is simply given by
$$
\text{ITD}_1(\theta)=\frac{1}{c_s}\bigg|\|\theta-\mu_1\|_2-\|\theta+\mu_1\|_2\bigg|
$$
and similarly for the pair $(\mu_2,-\mu_2)$, see \cite{forbes2021approximate} for more details. To account for measurement errors in ITD, we assume that it is observed under a small additive Student noise and the observed data model is
$$
Y=\text{ITD}_J(\theta)+\sigma\varepsilon\,,\qquad J\sim\text{Unif}\{1,2\}\,,\quad \varepsilon\sim\text{Student}(\nu)\,.
$$
In our experiment we have used the parameters $a=1$, $b=0.22$, $\sigma^2=10^{-5}$, $\nu=3$ and $Y$ was drawn from the likelihood model using the true location $\theta^\ast=(0.75,0.25)$ (the outcome of that draw was $Y=0.132)$. The corresponding posterior distribution denoted $\pi$ is illustrated at Figure \ref{fig:sound1}. In such a setting the mass of the posterior distribution concentrates around one-dimensional manifolds as $\sigma$ goes to zero and thus we expect a locally-weighted MCMC to perform better than a random-scan approach, provided that the directions supporting the manifold can be identified.
\end{exa}

\begin{figure}
\centering
\includegraphics[scale=0.25]{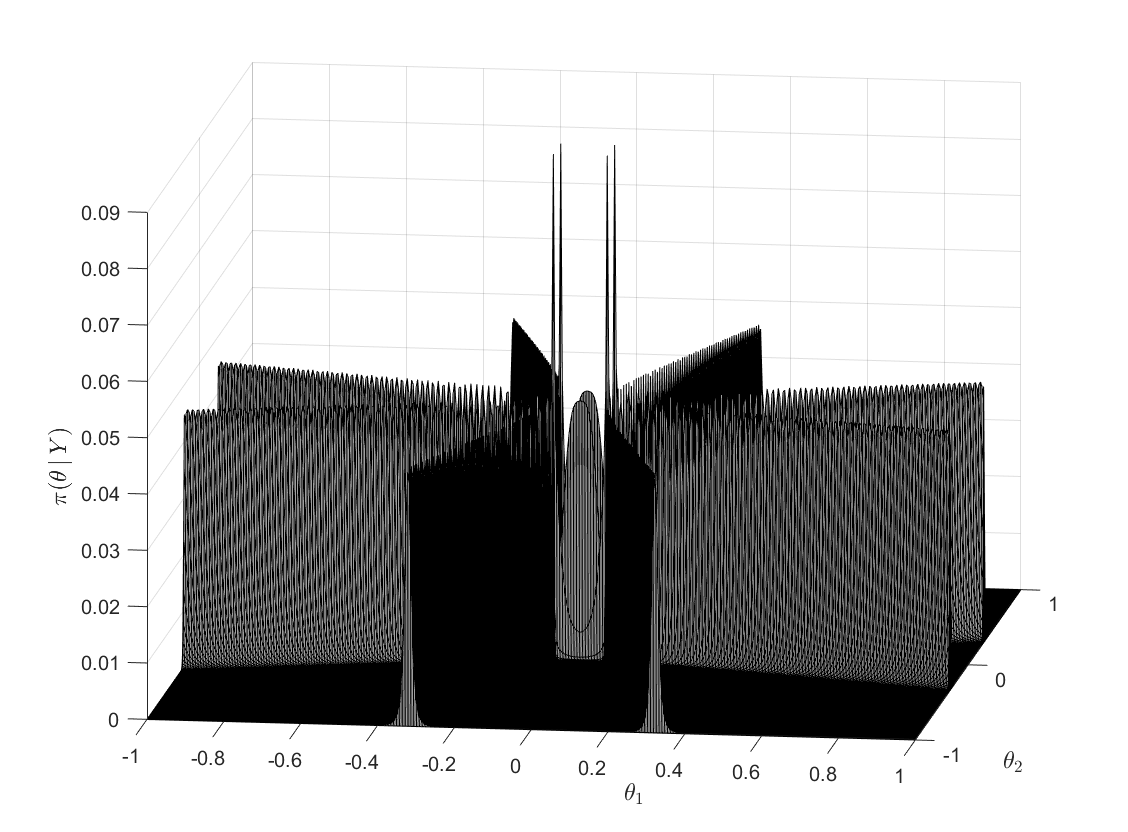}\includegraphics[scale=0.5]{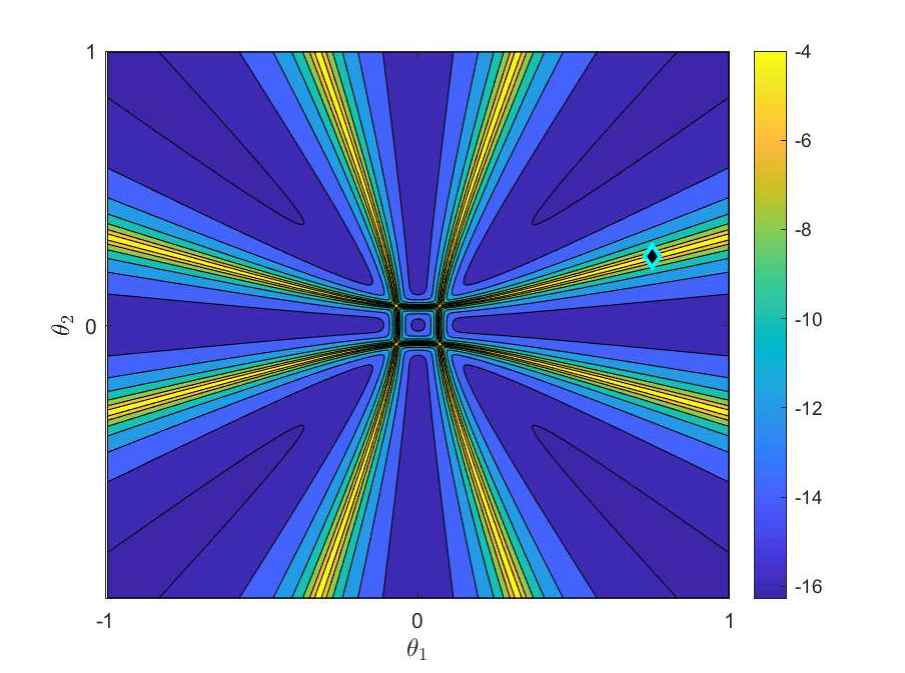}
\caption{Example \ref{ex9}. Illustration of the posterior distribution $\pi(\theta\,|\,Y)$ (left) and contour lines of the function $\theta\mapsto\log\pi(\theta\,|\,Y)$, the black diamond gives the location of the true source $\theta^\ast$  (right).\label{fig:sound1}}
\end{figure}

\paragraph*{Setup.} We do not consider here the locally-weighted kernel $\bP_\varpi$ as defining a relevant weight function $\varpi$ appears difficult without further knowledge on $\pi$. We thus compare the random-scan kernel $P_\omega$ to the locally-weighted kernel $\bP_L$, for several values of $L$. To identify the manifold directions, a na\"ive random-walk Metropolis-Hastings was first run and reasonable directions were obtained by doing a local linear regression on the resulting samples: four directions were identified which, combined with the two canonical ones, resulted in the definition of $n=6$ proposal kernels available to both the random-scan and locally-weighted samplers. The proposal variance of each proposal $Q_i$ was adaptively tuned for both $P_\omega$ and $\bP_L$, as in Example \ref{ex7}. Furthermore, the selection probability $\omega$ of the random-scan algorithm was set to the uniform distribution on $\{1,2,\ldots,6\}$.

\paragraph*{Results.} Figure \ref{fig:ex_8_KLcv} shows that even when taking into considerations the additional computational complexity of $\bP_L$ relatively to $P_\omega$, the locally-weighted algorithms converge faster than the random-scan approach, in the KL divergence sense. In particular, $\bP_{10}$ is more than 3 times faster than $P_\omega$. We note that the large computational budget of $\bP_{1000}$ makes it barely faster than $P_\omega$. Table \ref{tab:cpu:ex8} provides the relative computational complexity $\mathcal{T}(\bP_L)/\mathcal{T}(P_\omega)$ for several $L$, where for any Markov kernel $P$, $\mathcal{T}(P)$ is the computational cost of one iteration of $P$. In particular $\bP_{10}$ is about 3 times more expensive to run than $P_\omega$, implying that for a given runtime one would obtain 3 times more samples by using $P_\omega$ rather than $\bP_{10}$. Hence if computation is not an issue, $\bP_{10}$ does converge roughly 10 times as fast as $P_\omega$ in the KL divergence sense, rather than $3.3$ shown in the time normalized experiment of Figure \ref{fig:ex_8_KLcv}.

Table \ref{tab:ex_8_avar} provides an insight of the asymptotic variance of the locally-weighted and random-scan samplers for three test functions, $f_1(x)=\1_{\{x_1<0.2\}}$, $f_2(x)=\|x\|_2^2$ and $f_3(x)=x_2^2$. For each $L\in\{1,10,100,1000\}$, the left subcolumn gives an estimation of ratio \allowbreak $\var(\bP_L,f)/\var(P_\omega,f)$ and the right subcolumn adjusts those ratios by taking into account the different computational complexity of each kernel. Indeed, if a time budget $\tau$ is available, then for any kernel Markov $P$ the number sample produced  is given by $T(P)=\lfloor\tau/\mathcal{T}(P)\rfloor$. Thus the variance of the MCMC estimator obtained with a time budget $\tau$ is, provided that $\tau$ is large enough, roughly equal to $\var(P,f)\mathcal{T}(P)/\tau$. The right subcolumns are thus given by $\var(\bP_L,f)\mathcal{T}(\bP_L)/\var(P_\omega,f)\mathcal{T}(P_\omega)$, i.e. the product of each left subcolumn of this Table by the corresponding entry of Table \ref{tab:cpu:ex8}. This shows that, without considering the computational aspect of the samplers, $\bP_L$ dominates $P_\omega$ for each function considered and, even though $\bP_{1000}$ hardly provides any improvement over $\bP_{100}$, the larger $L$ the better. Taking into account the computational time, $\bP_1$, $\bP_{10}$ and $\bP_{100}$ offer better results than $P_\omega$ and more precisely $\bP_{10}>\bP_{1}>\bP_{100}$, while $\bP_{1000}$ may even be worse than $P_\omega$ for some functions due to its overwhelmingly large computational budget. Again, we stress that the computational efficiency of $\bP_L$ could be greatly improved using a parallel implementation for the weight function learning step and thus results given at Figure \ref{fig:ex_8_KLcv} and in the right subcolumns of Table \ref{tab:cpu:ex8} are in fact flattering for the random-scan approach.

\begin{figure}
\centering
\includegraphics[scale=0.6]{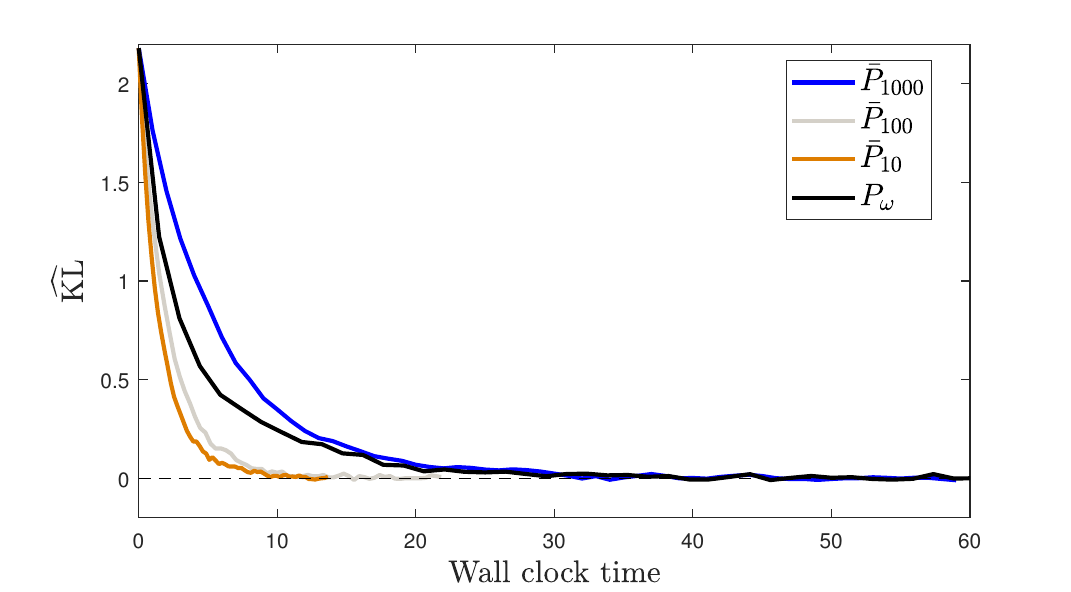}
\caption{Example \ref{ex9}. Convergence of the random-scan kernel $P_\omega$ and several locally-weighted algorithm $\bP_L$, in function of the wall clock time, starting from $\nu=\text{unif}([-a\,;\,a]^2)$.\label{fig:ex_8_KLcv}}
\end{figure}

\begin{table}
\centering
\begin{tabular}{cccc}
\hline
$\mathcal{T}(\bP_{1})/\mathcal{T}(P_{\omega})$& $\mathcal{T}(\bP_{10})/\mathcal{T}(P_{\omega})$ & $\mathcal{T}(\bP_{100})/\mathcal{T}(P_{\omega})$& $\mathcal{T}(\bP_{1000})/\mathcal{T}(P_{\omega})$\\
  \hline
  $2.26$ & $2.40$ & $4.25$ & $16.64$\\
  \hline
\end{tabular}
\caption{Example \ref{ex9}. Comparison of the runtime of a single iteration of  the locally-weighted algorithms $\mathcal{T}(\bP_{1})$, $\mathcal{T}(\bP_{10})$, $\mathcal{T}(\bP_{100})$ and $\mathcal{T}(\bP_{1000})$, relatively to a single iteration of  the random-scan $\mathcal{T}(P_{\omega})$.\label{tab:cpu:ex8}}
\end{table}

\begin{table}
\begin{tabular}{ccccc}
\hline
 & \multicolumn{2}{c}{$\var(\bP_1,f)/\var(P_\omega,f)$}    & \multicolumn{2}{c}{$\var(\bP_{10},f)/\var(P_\omega,f)$} \\\hline
 $f_1$ &$0.019$ & $0.047$ & $0.013$& $0.038$ \\
 $f_2$ &$0.185$ & $0.464$ & $0.117$& $0.340$  \\
 $f_3$ & $0.096$& $0.241$ & $0.071$& $0.207$ \\
 \hline
\end{tabular}
\begin{tabular}{ccccc}
\hline
 &    \multicolumn{2}{c}{$\var(\bP_{100},f)/\var(P_\omega,f)$}  & \multicolumn{2}{c}{$\var(\bP_{1000},f)/\var(P_\omega,f)$}\\\hline
 $f_1$ & $0.010$& $0.051$ & $0.010$& $0.204$ \\
 $f_2$ & $0.094$& $0.480$ & $0.091$& $1.856$\\
 $f_3$ &  $0.068$& $0.348$ & $0.071$& $1.451$\\
 \hline
\end{tabular}
\caption{Example \ref{ex9}. Ratio of the asymptotic variance of several locally-weighted MCMC estimator of $\pi f$ for three test functions, $f_1(x)=\1_{\{x_1<0.2\}}$, $f_2(x)=\|x\|_2^2$ and $f_3(x)=x_2^2$, relative to that of the random-scan MCMC estimator. For each $L$, the left subcolumn gives the estimation of the asymptotic variance ratios while the right subcolumn gives the ratio of the MCMC estimator in time-normalized experiments. \label{tab:ex_8_avar}}
\end{table}

\subsection{Messenger RNA-based transfection studies using GFP}
\begin{exa}
In molecular biology, messenger RNA (mRNA) can be used to deliberately transfer genes to a cell, a step known as transfection. Understanding the delivery and the kinetics of transfection has important implications in medical biotechnology, including COVID-19 vaccines. The green fluorescent protein (GFP), derived from the jellyfish, is a reliable marker which allows to visualize spatial and temporal patterns of gene expression in vivo, see \cite{soboleski2005green} and \cite{leonhardt2014single} for further details on the use of GFP to monitor mRNA-based transfection. The statistical model proposed in \cite{leonhardt2014single} is a two-state ordinary differential equation with unknown functions $m$ and $G$ giving the number of mRNA and GFP molecules, respectively. This system of ODE is parameterized by $\kappa$ the translation rate of mRNA to GFP, $\beta$ and $\delta$ the degradation rate of GFP and mRNA, respectively, such that:

\begin{equation}
\left\{
\begin{array}{l}
\dfrac{\rmd G}{\rmd t}=\kappa m-\beta G\,, \vspace{0.25cm}\\
\dfrac{\rmd m}{\rmd t}=-\delta m\,.
\end{array}
\right.
 \label{eq:exRNA}
 \end{equation}
The observed data $\{Y_1,Y_2,\ldots,Y_r\}$ are obtained from $r$ repeated experiments, each of which consists of $p=180$ GFP measurements at regular time intervals over a period of 30 days. In particular, for each experiment $i\in\{1,\ldots,r\}$, $Y_i=(Y_{i,1},\ldots,Y_{i,p})$ and $Y_{i,j}$ is the number of GFP molecules measured in experiment $i$ at time-step $\tau_j$, with $\tau_p=30$ days. The statistical model is given by
\begin{multline*}
(i,j)\in\{1,\ldots,r\}\times\{1,\ldots,p\}\,,\qquad Y_{i,j}=G(\tau_j)+\sigma\varepsilon_{i,j}\,,\\
\qquad (\varepsilon_{1,1},\ldots,\varepsilon_{1,p},\ldots,\varepsilon_{r,p})\sim_{\iid}\norm(0,1)\,.  
\end{multline*}
The objective is to estimate the model parameters $\beta$, $\delta$ and $\kappa$, conditionally on $\{Y_1,\ldots,Y_r\}$ in order to infer the mRNA-based transfection procedure. In this particular case, Eq. \eqref{eq:exRNA} can be solved so that $G$ is known. The data used in this Example were first published in \cite{frohlich2018multi} \footnote{The data are available at \url{https://doi.org/10.5281/zenodo.1228898} under the license CC BY-SA 4.0. The data consist in the first 20 GFP trajectories extracted from the file \texttt{code/project/data/20160427\_mean\_eGFP.xlsx} which were normalized so as to range in $[0,1]$.}.
 Furthermore, we fix the noise parameter in Eq. \eqref{eq:exRNA} to $\sigma=5$ and proceed with the estimation of the other parameters in log-scale, i.e  $X:=[\log\beta,\,\log\delta,\,\log\kappa]$. The prior distribution on $X$ is Gaussian with mean parameter $[2\,,2\,,2]$ and covariance matrix $10\mathrm{I}_3$ with the constraint that $X_i<0$ for $i\in\{1,2,3\}$ since all parameters are rate parameters.
Two conditional posterior distributions are represented as heat map at Figure \ref{fig:exRNA_1}. \label{ex3}
\end{exa}

\paragraph*{Setup.}
It is well known that in such ODE-based models, identifiability issues regarding the degradation rate parameters may arise. To take this into account, in addition of $Q_0$ defined as a Gaussian proposal kernel with diagonal covariance matrix $\Delta_0$, which can be seen as a ``na\"ive'' MH proposal kernel, we consider the following collection of proposal kernels. For any even $n$ and $i\in\{1,\ldots,n\}$, let $Q_{i,n}$ be the Gaussian proposal kernel with covariance matrix given by $\Delta_i H_{i,n} H_{i,n}\T$ where $\Delta_i$ is a scaling diagonal matrix and $H_{i,n}$ the orthogonal projection matrix onto the space spanned by $[0\,,0\,,1]$ and $[-\sin(\vartheta_{i,n}),\,\cos(\vartheta_{i,n}),\,0]$ with $\vartheta_{i,n}=\pi(i-1)/n$. In other words, while $Q_0$ perturbs the current state in a three dimensional spaces, each $Q_{i,n}$ operates on a specific two-dimensional plane, capitalizing on the correlation structure of $\beta$ and $\delta$ (or the lack thereof). In particular, for any even $n$, $Q_{1,n}$ only allows perturbations of parameters $\beta,\kappa$ while $Q_{n/2,n}$ only allows perturbations of parameters $\delta,\kappa$. Figure \ref{fig:exRNA_2} illustrates how those degenerate Gaussian kernels effectively discretize the three dimensional space into a collection of two dimensional planes which appear more adapted to explore the posterior geometry. As $n$ increases, the hyperbolic shape typical to the conditional posterior of $(\log\beta,\log\kappa)$ (see the LHS of Figure \ref{fig:exRNA_1}) could be reasonably well approximated by a mixture of those Gaussian kernels. To summarize, the case $n=0$ corresponds to the na\"ive Random Walk Metropolis--Hastings algorithm, while for any even number $n>0$, $n+1$ proposal kernels are available $Q_0,Q_{1,n},Q_{2,n},\ldots,Q_{n,n}$. In the following, we compare how the random-scan and locally-weighted kernels exploit this family of proposals.

\paragraph*{Weights.} For now, assume that $n$ is fixed and dependance on $n$ is made implicit in notations. Given the symmetry of the problem, a uniform weight is set for the random-scan kernel $P_\omega$, i.e. $\omega=(\omega_0,\ldots,\omega_n)$ with $\omega_i=1/(n+1)$ for all $i\in\{1,\ldots,n\}$.  A numerical analysis of the posterior distribution shows that when $|x_1|$ (resp. $|x_2|$) gets larger, the posterior variability of $X_2$ (resp. $X_1$) given $X_1=x_1,X_3$ (resp. $X_2=x_2,X_3$) decreases significantly to present a Gaussian look centered between $-2$ and $-3$ but that the corresponding posterior densities decreases much more slowly. As a consequence, one would seek to assign a larger probability to the two kernels that update only $X_1$ (resp. only $X_2$) when $|X_1|$ (resp. $|X_2|$) is very large. While such an implementation is impossible for random-scan samplers, a locally-weighted algorithm allows it by setting the weight function $\varpi(x)=(\varpi_0(x),\ldots,\varpi_n(x))$ as
\begin{equation}\label{eq:exRNA:weight}
\varpi_i(x)\propto \left\{\begin{array}{cc}
1&\text{if } i\in\{0,\ldots,n\}\backslash\{1,n/2\}\,,\\
1\vee \left(x_1+2.5\right)^2&\text{if }i=1\,,\\
1\vee \left(x_2+2.5\right)^2&\text{if }i=n/2\,.\\
\end{array}\right.
\end{equation}
While not arbitrary, this choice of weight function is not optimal in any sense and may be improved upon. In particular, it can be noted that the locally-weighted feature only appears when the Markov chain visits or attempts visiting a tail region defined as $\{x:|x_1-(-2.5)|>1 \;\text{or}\;|x_2-(-2.5)|>1 \}$ as otherwise the selection probability is uniform. Nevertheless, as simple as it is, this choice of weight function leads, as we shall see, to a significant improvement over random-scan strategies.

\paragraph*{Tuning parameters.}
There are two tuning parameters: the number of degenerate proposal kernels $n$ and the proposal scaling parameters $\Delta_{0},\Delta_1,\ldots,\Delta_n$. The scaling parameters were tuned on the fly using an adaptive MCMC strategy until the acceptance rate of each kernel stabilizes between 30\% and 40\%. Different values of $n$ were considered with $n\in\{0,2,4,6,10,14,20\}$. For each algorithm the value of $n$ retained is the one that maximizes the Expected Squared Jump Distance (ESJD). As shown at Figure \ref{fig:exRNA_esjd},  $n=2$ is best (on average) for random-scan and $n=10$ is best (on average) for the locally-weighted strategy.

\paragraph*{Results.} Two types of results are reported: Figure \ref{fig:exRNA_KL} compares the convergence properties of locally-weighted (with $n=10$), random-scan (with $n=2$) and na\"ive MH (with $n=0$) when the initial distribution is set as the prior distribution. In Table \ref{tab:exRNA_var}, the asymptotic performance of those three algorithms are reported, for four test functions. Both illustrations show the significant gain in statistical efficiency offered by the above-specified locally-weighted strategy, which could be even further optimized by a more sophisticated choice of weight function $\varpi$, over a sensible random-scan alternative.

\begin{figure}[h]
\centering
\includegraphics[scale=0.3]{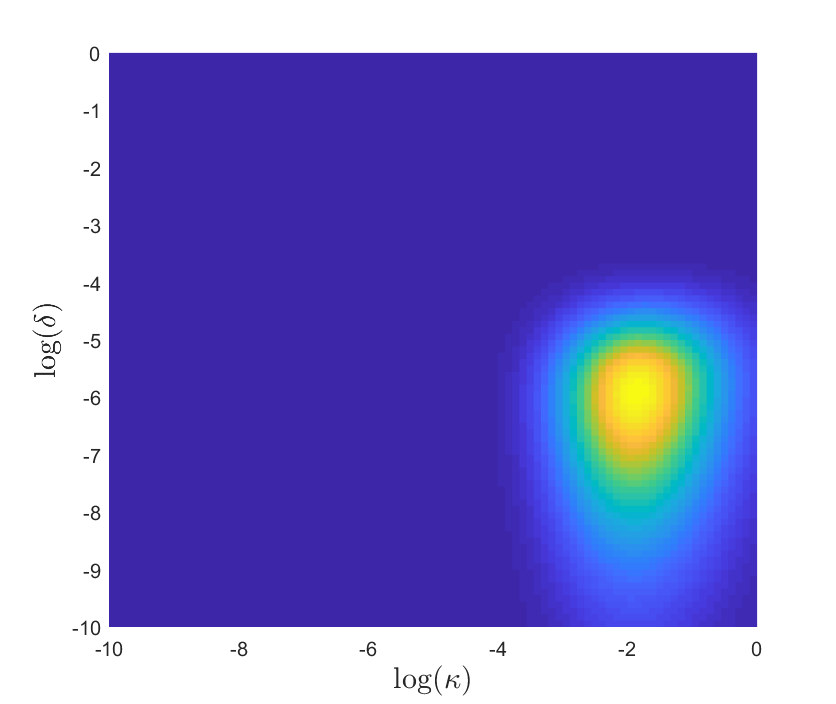}\includegraphics[scale=0.3]{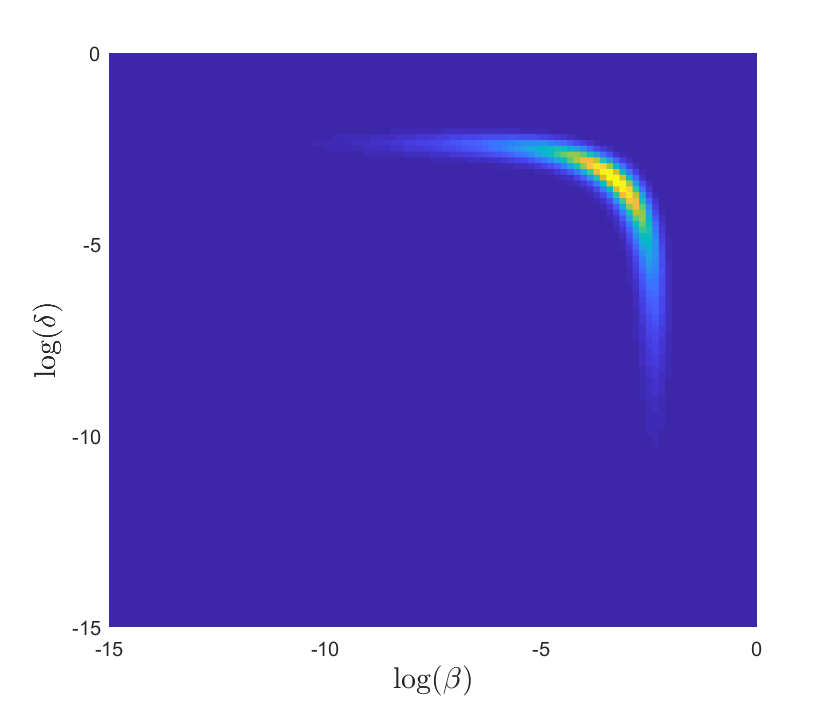}
\caption{Example \ref{ex3}.  Heat map of two conditional posterior distributions, $\pi(\log\beta,\log\kappa\,|\,Y,\log\delta=-3.51)$ and $\pi(\log\beta,\log\delta\,|\,Y,\log\delta=-2.30)$.}
\label{fig:exRNA_1}
\end{figure}

\begin{figure}[h]
%\centering
\hspace{-1cm}
\includegraphics[scale=0.35]{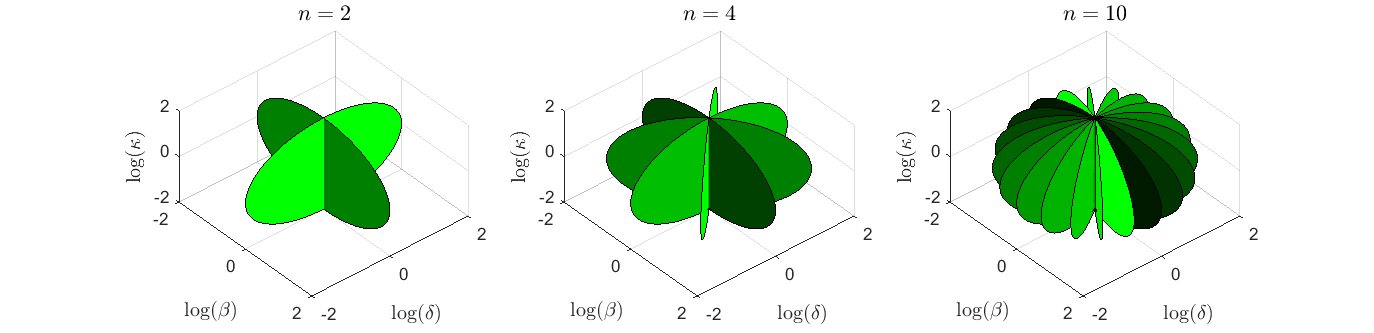}
\caption{Example \ref{ex3}.  Representation of the support of the proposal kernels $Q_{1,n}(\mathbf{0}_3,\cdot),\ldots,Q_{n,n}(\mathbf{0}_3,\cdot)$, for $n\in\{2,4,10\}$ by mean of disks (the support of those degenerate Gaussian kernels is unbounded).}
\label{fig:exRNA_2}
\end{figure}

\begin{figure}
\hspace{-3cm}
\includegraphics[scale=0.67]{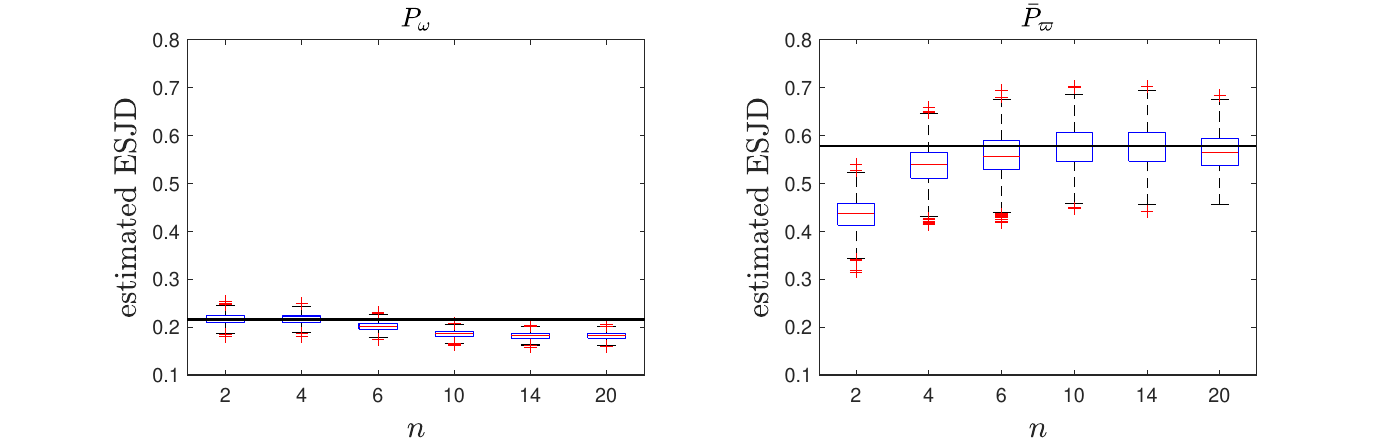}
\caption{Example \ref{ex3}. Expected Squared Jumping Distance (ESJD) for the two chains, random-scan on the LHS and locally-weighted on the RHS. Results are based on the replication of 10,000 Markov chains each with time horizon $T=10,000$, started from the prior distribution. The jumps size of $P_\omega$ (resp. $\bar{P}_\varpi$) is (on average) maximized for $n=4$ (resp. for $n=10$). Note that the choice of $n=4$ (resp. $n=14$) leads to similar results but the smaller value is retained for parsimony. We note that the average jumps size achieved by $\bar{P}_\varpi$ (with $n=10$) is on average 2.7 larger than that of $P_\omega$ (with $n=2$). By comparison, the ESDJ achieved for $n=0$ i.e. only $Q_0$ is used (both algorithms coincide in this case),  is about $14\%$. \label{fig:exRNA_esjd}}
\end{figure}

\begin{figure}
\centering
\includegraphics[scale=0.72]{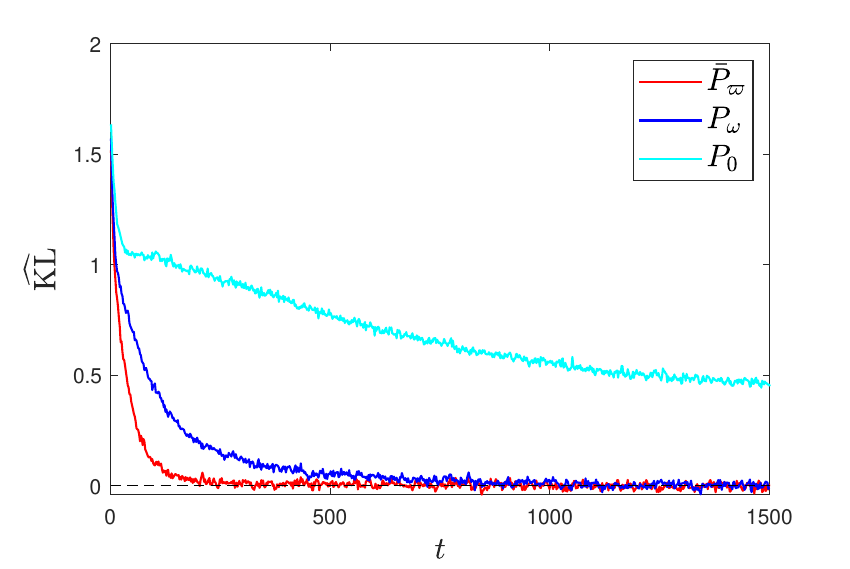}
\caption{Example \ref{ex3}.\label{fig:exRNA_KL} Illustration of the convergence of the locally-weighted $\bar{P}_\varpi$ (with $n=10$), the random-scan ${P}_\omega$  (with $n=2$) and the na\"ive MH $P_0$ (with $n=0$) from the prior distribution to the posterior, as estimated by the Kullback-Leibler divergence between $\pi$ and $\text{prior}\times P^t$ with $P\in \{\bar{P}_\varpi,P_\omega,P_0\}$. Results obtained from $3,000$ independent replications of the three Markov chains. Empirically, locally-weighted reaches at $t\approx 300$ a neighborhood of $\pi$ after which no visible improvement can be seen on this plot. This same neighborhood is reached at $t\approx 1050$ by random-scan. The na\"ive MH remains far off the stationary regime even after $t=1,500$. }
\end{figure}

\begin{table}
\centering
\begin{tabular}{cccccc}
\hline
& $f_1(x)=x_1$ &  $f_2(x)=x_3$ & $f_3(x)=\|x\|$ & $f_4(x)=\1_{\{x_1<-10\}}$ \\
\hline
${\var(\bP_\varpi,f)}/{\var(P_\omega,f)}$ & $0.857$ & $0.300$ & $0.353$ & $0.198$ \\
${\var(\bP_\varpi,f)}/{\var(P_0,f)}$ & $0.448$ & $0.068$ & $0.140$ & $0.074$\\
\hline
\end{tabular}
\caption{Example \ref{ex3}. First row shows the asymptotic variance ratio between locally-weighted $(n=10)$ and random-scan $(n=2)$, for four test functions. Second row shows the asymptotic variance ratio for the same functions between locally-weighted $(n=10)$ and the na\"ive MH kernel $(n=0)$. Results were obtained from $1,000$ replications of the three Markov chains started (approximately) at stationarity with time horizon $T=10,000$. It was checked that for that time horizon, all MCMC estimates had reached their normal asymptotic regime, for each test function. \label{tab:exRNA_var}}
\end{table}

\section{Conclusion}

A general way of aggregating $n$ reversible Markov kernels (Algorithm \ref{alg1}) or $n$ proposal kernels (Algorithm \ref{alg2}) using a state-dependent selection probability $\varpi\in\Delta_{n-1}^\Xset$ has been proposed. While designing a meaningful $\varpi$ may seem difficult when little prior information on $\pi_\sigma$ is available,  our work has shown that a locally-weighted algorithm using particle-based weights, which does not require any additional knowledge on $\pi_\sigma$ (see Section \ref{sec5:weightfunction}), outperforms significantly a random-scan counterpart, even in time-normalized experiments. This gain can be further increased by considering parallel computing environment for the weights calculation.

In the same range of ideas, considering modern day statistical models where the likelihood function evaluation is computationally very expensive, MCMC methods such as the locally-weighted algorithms proposed in this paper, the locally-balanced \citep{zanella2020informed} or even the Multiple-Try algorithm can take advantage of a parallel implementation. Indeed, taking for simplicity the case of Example \ref{ex1} where $\Xset=\rset^d$ with $n=d$ Gibbs kernels $P_1,\ldots,P_d$ and assuming a noise vanishing context where for each $x\in\Xset$, at most two kernels are adapted to the local geometry of $\pi_\sigma$, parallelization implies that those methods would need $\mathcal{O}(1)$ likelihood evaluations to move significantly on $\pi_\sigma$. This contrasts with uninformed methods that would need $\mathcal{O}(d)$ likelihood evaluations to do so. Thus, in addition to the obvious greater statistical efficiency, those methods also benefit from a better computational efficiency in such a scenario.

Another interesting aspect of our work is to question the search for asymptotical optimality in the design of MCMC samplers. Indeed, while we have seen (Section \ref{ex3}) that the locally-weighted approach can be less efficient than its random-scan counterpart in the asymptotic regime, the former can reach a very good approximation of $\pi_\sigma$ much faster than the latter. The noise vanishing distributions provides a typical class of sampling problems where this phenomenon occurs, sometimes exhibiting a dramatic improvement  as shown by some examples studied in this paper. New tools such as the pseudo-spectral gap \citep{atchade2019approximate} or new analyses based on the drift and minorization approach on large sets \citep{yang2017complexity} may help to develop more precise non-asymptotical error bounds which would formalize that aspect explicitly.

\section{Supplementary Material}
\subsection{Proof of Proposition \ref{prop:hypercube:1}}
\label{proof1}
\begin{proof}
We first recall some basic notions related to discrete Markov chains coupling. Let $\pi$ be a distribution on $(\Xset,\Xalg)$ and two $\pi$-invariant Markov chains $\{X_t\}:=\{X_t,\,t\in\nset\}$ and $\{X'_t\}:=\{X'_t,\,t\in\nset\}$ with the same transition matrix $P$. A joint process $\{\Gamma_t\}:=\{(X_t,X'_t)\}$ defined on $(\Xset\times\Xset,\Xalg\otimes\Xalg,\proba)$ is referred to as a coupling of $\{X_t\}$ and $\{X'_t\}$ if $\{\Gamma_t\}$ admits $\{X_t\}$ and $\{X'_t\}$ as marginal distributions. Defining the coupling time $\tau(\Gamma)$ as
$$
\tau(\Gamma):=\inf_{t\in\nset}\{X_t=X'_t\}\,,
$$
a useful property of coupled Markov chains, arising from the coupling inequality states that:
\begin{equation}
\label{eq:coupl}
\|P^t(x,\,\cdot\,)-P^t(y,\,\cdot\,)\|\leq \proba_{x,y}\{\tau>t\}\,,
\end{equation}
where $\proba_{x,y}$ is the probability distribution generated by the simulation of the coupled Markov chain $\{\Gamma_t\}=\{X_t,X'_t\}$ started at $\Gamma_0=(x,y)$. In Eq. \eqref{eq:coupl}, we have used the shorthand notation $\tau$ for $\tau(\Gamma)$, noting however that a coupling time is relative to a specific coupling. Since we have
\begin{equation}
\label{eq:coupl2}
\sup_{x\in\Xset}\|P^t(x,\,\cdot)-\pi\|\leq \sup_{(x,y)\in\Xset^2}\|P^t(x,\,\cdot)-P^t(y,\,\cdot)\|\,,
\end{equation}
combining Eqs. \eqref{eq:coupl} and \eqref{eq:coupl2} shows that the coupling time distribution characterizes the Markov convergence. In particular, using Markov inequality, we have
$$
\sup_{x\in\Xset}\|P^t(x,\,\cdot)-\pi\|\leq \frac{1}{t}\sup_{(x,y)\in\Xset^2}\esp_{x,y}(\tau)\,,
$$
where $\esp_{x,y}$ is the expectation under $\proba_{x,y}$.

In this proof, for any quantity $\alpha$ relative to the random-scan Gibbs sampler, the equivalent quantity related to the locally-weighted algorithm (Alg. \ref{alg1}) will be denoted as $\alpha^\ast$. In particular, let $\proba^\ast$ be the probability distribution generated by Algorithm \ref{alg1} and $\esp^\ast$ be the expectation operator under $\proba^\ast$. The dependance on $\omega$ and $\varpi$ is omitted.

Without loss of generality, we order $\Xset$ such that the states $\{x_1,\ldots,x_{1+d(m-1)}\}$ correspond to the filament (\ie $\Zset$). We notice that the transition matrices $M$ and $\Mast$ corresponding respectively to the random-scan and the locally-weighted sampler (Alg. \ref{alg1}) satisfy in this case:
\begin{equation}
M=
\begin{bmatrix}
P & 0 \\
A & B
\end{bmatrix}
\qquad \text{and}\qquad \Mast=
\begin{bmatrix}
\Past & 0 \\
A^\ast & B^\ast
\end{bmatrix}\,,
\end{equation}
and clearly $\Zset$ is an absorbing state. Assuming that both Markov chains start in $\Zset$, it is thus sufficient to analyse only the transition matrices $P$ and $\Past$ which are essentially the restriction of the Markov chains to $\Zset$. Let $\{X_t\}$ and $\{X^\ast_t\}$ be the two Markov chains generated by $P$ and $\Past$ respectively.

The first step of the proof consists in projecting the Markov chains $\{X_t\}$ and $\{X^\ast_t\}$ onto a smaller state space by lumping some states from $\Zset$ together. Let us write $\Zset$ as $\Zset=\{\Vset_1,\Eset_1,\Vset_2,\Eset_2,\ldots,\Eset_d,\Vset_{d+1}\}$ where $\Vset_k$ and $\Eset_k$ are respectively the $k$-th vertex and the $k$-th edge of the hypercube that belongs to $\Zset$ such that $\Vset_k\cap\Eset_k=\{\emptyset\}$. The folded representation of the Markov chain $\{X_t\}$ with transition kernel $P$ is the discrete time process $\{Y_t\}$ defined on $\Yset=\{1,\ldots,2d+1\}$ as follows: if there is $k\in\{1,\ldots,d+1\}$ such that $X_t=\Vset_k$, set $Y_t=2k-1$ or if there is $k\in\{1,\ldots,d\}$ such that $X_t\in\Eset_k$, set $Y_t=2k$. In other words, $\{Y_t\}$ inherits the vertices from $\{X_t\}$ but aggregates together into a unique state, the states that are in between two consecutive vertices. The same mapping allows to define $\{Y_t^\ast\}$ as the \textit{folded} version of the locally-weighted Markov chain $\{X^\ast_t\}$. An illustration of the folded Markov chains $\{Y_t\}$ and $\{Y_t^\ast\}$ is given in Figure \ref{fig:representation}, in the case where $d=3$. In the following, we refer to as $Q$ (resp. $Q^\ast$) the transition matrix of $\{Y_t\}$ (resp. $\{Y_t'\}$).

\begin{figure}
\centering
\fbox{
\begin{tikzpicture}[,->,>=stealth',shorten >=1pt,auto,node distance=2.2cm,
                    thick,main node/.style={circle,draw,font=\sffamily\Large\bfseries}]
  \node[main node] (1) {1};
  \node[main node] (2) [below right of=1] {2};
  \node[main node] (3) [above right of=2] {3};
  \node[main node] (4) [below right of=3] {4};
  \node[main node] (5) [above right of=4] {5};
  \node[main node] (6) [below right of=5] {6};
  \node[main node] (7) [above right of=6] {7};
  \path[every node/.style={font=\sffamily\small}]
    (1) edge [bend left] node {$\alpha$} (3)
        edge [bend left] node[below] {$\bar{\beta}$} (2)
        edge [loop left] node {} (1)
    (2) edge [bend right] node[right] {$\alpha$} (3)
        edge [bend left] node[below] {$\bar{\alpha}$} (1)
        edge [loop below] node {} (2)
    (3) edge [bend right] node [below] {$\beta$} (2)
        edge node [below] {$\alpha$} (1)
        edge [bend left] node[below] {$\beta$} (4)
        edge [bend left] node[above] {$\alpha$} (5)
        edge [loop above] node {} (3)
    (4) edge [bend left] node {$\alpha$} (3)
        edge [bend right] node [right] {$\alpha$} (5)
        edge [loop below] node {} (4)
    (5) edge [loop above] node {} (5)
        edge node [below] {$\alpha$} (3)
        edge [bend left] node[below] {$\beta$} (6)
        edge [bend left] node[above] {$\alpha$} (7)
        edge [loop above] node {} (5)
        edge [bend right] node [below] {$\beta$} (4)
    (6) edge [bend left] node {$\alpha$} (5)
        edge [bend right] node [right] {$\bar{\alpha}$} (7)
        edge [loop below] node {} (6)
    (7) edge node {$\alpha$} (5)
        edge [bend right] node[below] {$\bar{\beta}$} (6)
        edge [loop right] node {} (7);
\end{tikzpicture}}
\caption{Projection on the folded space of the random-scan and locally-weighted Markov chains sampling from $\pi$, in the case where $d=3$. The odd states correspond to vertices and the even ones to the aggregated states between two vertices. For the random-scan, the transition probabilities of the \textit{folded} Markov chain $\{Y_t\}$ are $\alpha=\bar{\alpha}=1/dm$ and $\beta=\bar{\beta}=\{1-\frac{2}{m}\}/d$. For the locally-weighted algorithm, the transition probabilities of $\{Y^\ast_t\}$ are $\alpha^\ast=1/2m$, $\bar\alpha^\ast=1/m$, $\beta^\ast=1/2-1/m$ and $\bar\beta^\ast=1-2/m$. For each state, the self loop indicate the probability to stay put, which equals one minus the sum of outwards probabilities. \label{fig:representation}}
\end{figure}

The second step is to define a coupling for the two folded Markov chains $\{Y_t\}$ and $\{Y^\ast_t\}$. For simplicity, we only present the coupling for $\{Y_t\}$ but the same approach is used for $\{Y^\ast_t\}$. Since there is an order on $\Yset$, we consider the reflection coupling presented in Algorithm \ref{alg:coupling} that exploits the symmetry of the Markov chain. Clearly since $U\sim\unif(0,1)$ implies that $1-U\sim\unif(0,1)$, the marginal chains satisfy $Y_t\sim Q^t(Y_0,\,\cdot\,)$ and $Y_t'\sim Q^t(Y_0',\,\cdot\,)$ and the resulting discrete time process $\{(Y_t,Y_t')\}$ jointly defined is a coupling of $\{Y_t\}$ and $\{Y_t'\}$. The coupling introduced in Algorithm \ref{alg:coupling}, allows to derive the expected coupling time, \ie the time at which the two Markov chains $\{Y_t\}$ and $\{Y_t'\}$ coalesce. By symmetry, the Markov chains coalesce necessarily when $Y_\tau=Y_\tau'=d+1$. Therefore, denoting by $\esp_0^\diamond$ the expectation under the coupling $\{(Y_t,Y_t')\}_t$ on $(\Yset\times\Yset,\Yalg\otimes\Yalg)$ started at $Y_0=1$ and $Y_0'=2d+1$, we have
$$
\esp_{0}^\diamond(\tau)=\esp_1^\diamond(T_{d+1})=\esp_{2d+1}^\diamond(T_{d+1})\,,\\
$$
where for any $k\in\Yset$, $T_k:=\inf\{t>0,\;Y_t=k\}$ and $\esp_k^\diamond$ denotes the expectation of the marginal Markov chain $\{Y_t\}$ started at $Y_0=k$. The same coupling for the locally-weighted Markov chain yields $\esp_{0}^{\ast\diamond}(\tau)=\esp_1^{\ast\diamond}(T_{d+1})$. Central to this proof is the fact that a reflection coupling similar to Algorithm \ref{alg:coupling} exists for the Markov chains $\{X_t\}$ and $\{X^\ast_t\}$ and since
the average time to reach the middle of the filament $\Zset$ when starting from one end is the same regardless whether the space is folded or not we have
\begin{equation}
\label{eq:coupling_folded}
\esp_1^\diamond(T_{d+1})=\esp_1(T_{d+1})\,,
\end{equation}
which implies that $\esp_{0}^\diamond(\tau)=\esp_{0}(\tau)$. The same argument holds for the locally-weighted Markov chain $\{X^\ast_t\}$ and its folded version $\{Y^\ast_t\}$.
\begin{algorithm}
\caption{Reflection coupling on the hypercube}\label{alg:coupling}
\begin{algorithmic}[1]
\State Initialise the two Markov chains with $Y_0=1$ and $Y'_0=2d+1$
\State Set $t=0$, $Y=Y_0$ and $Y'=Y_0'$
\While{$Y_t\neq Y'_t$}
\State Draw $U\sim_\iid\unif(0,1)$ and set $U'=1-U$
\State Define $\eta=\{\sum_{i=1}^j Q(Y,i)\}_{j=1}^d$ and $\eta'=\{\sum_{i=1}^jQ(Y',2d+1-i)\}_{j=1}^d$
\State Set $Y=1+\sum_{k=1}^{d-1}\1_{\eta_k<U}$ and $Y'=1+\sum_{k=1}^{d-1}\1_{\eta'_k<U'}$
\State Set $t=t+1$, $Y_t=Y$ and $Y_t'=Y'$
\EndWhile
\State Set $\tau=t$
\ForAll{$t=\tau+1,\tau+2,\ldots$}
\State Simulate $Y_t$  using the steps (4)--(7) with $Y=Y_t$
\State Set $Y_{t}'=Y_t$
\EndFor
\end{algorithmic}
\end{algorithm}

Working on the folded space allows to derive $\esp_1^\diamond(T_{d+1})$ and $\esp_1^{\diamond\ast}(T_{d+1})$ in an easier way and this is the last part of the proof. We take $d$ even so as to make the algebra more immediate. In this case, since $d+1$ is odd, the state $d+1$ corresponds to a vertex. Clearly, $\esp_1^\diamond(T_{d+1})$ is the average time to absorption of a fictitious chain that would contain only the $d+1$ first states, replacing the outwards connections of $d+1$ by a self loop with probability 1. Denoting by $Q_{d+1}$ the transition matrix of this fictitious chain, by $Q_d$ the transition matrix of the $d$ first transient states and by $I_d$ the $d$-dimensional identity matrix, the matrix $I_d-Q_d$ is invertible and its inverse, often known as the fundamental matrix of $Q_{d+1}$, contains information related to the absorption time, see \eg the Chapter 11 in \citet{grinstead2012introduction}. In particular, we have that for any $i<d$ starting position of the chain, then
$$
\esp_i^\diamond(T_{d+1})=\{(I_d-Q_d)^{-1}1_d\}_i\,,
$$
where $1_d$ denotes here the $d$-dimensional $1$ vector. This implies that $\esp_1^\diamond(T_{d+1})$ is simply the sum of the first row of $(I_d-Q_d)^{-1}$. It is possible to calculate analytically the fundamental matrix for each chain $Q_{d+1}$ and $Q^\ast_{d+1}$ and the proof follows from comparing each first row sum.

Using symbolic computation provided by Matlab, we found the following entries for the first row of the $d$-dimensional fundamental matrix of the random-scan
\begin{equation}
v_d=\frac{1}{\alpha(2\alpha+\beta)}\left(\beta\;,\; 2\alpha \;,\;3\beta \;,\; 4\alpha\;,\; \cdots\;,\; (d-1)\beta \;,\; d\alpha\right)\,,
\end{equation}
and
\begin{multline*}
v_d^\ast=\frac{1}{\alpha(2\alpha^\ast+\beta^\ast)}\big(\beta^\ast\;,\; 2\alpha^\ast\;,\; \cdots
\quad (d-3)\beta^\ast \;,\; (d-2)\alpha^\ast\;,\;  \alpha^\ast(2\alpha^\ast+\beta^\ast)\phi_d\;,\;
\alpha^\ast(2\alpha^\ast+\beta^\ast)\psi_d \big)\,,
\end{multline*}
where
$$
\phi_d=\frac{2\beta^\ast}{\alpha^\ast\delta^\ast}\left\{\left(\frac{3d}{2}-1\right)\alpha^\ast+(d-1)\beta^\ast\right\}\,,
\qquad\psi_d=\frac{1}{\delta^\ast}\left\{3d\alpha^\ast+(2d-1)\beta^\ast\right\}\,,
$$
and $\delta^\ast=6{\alpha^\ast}^2+7\alpha^\ast\beta^\ast+2{\beta^\ast}^2$.
Letting $d=2p$ and using the fact that
$$
\sum_{k=1}^{2p}k\1_{\{k\,\text{is odd}\}}=p^2\,,\qquad \sum_{k=1}^{2p}k\1_{\{k\,\text{is even}\}}=p(p+1)\,,
$$
the sum of $v_d$'s elements is
\begin{equation}
\label{eq:proof_ex_1}
%\bar{v}_d
\esp_1^\diamond(T_d)=\frac{1}{\alpha(2\alpha+\beta)}\left\{\beta p^2+\alpha p(p+1)\right\}=\frac{m-1}{4}d^3+\frac{1}{2}d^2\,,
\end{equation}
by definition of $\alpha$ and $\beta$. Using the same argument the sum of $v^\ast_d$'s elements is
\begin{equation}
\label{eq:proof_ex_2}
\esp_1^{\ast\diamond}(T_d)=\frac{m-1}{2}d^2+(3-2m)d+2(m-2)+\phi_d+\psi_d\,.
\end{equation}
By straightforward algebra, we have
\begin{equation*}
\phi_d+\psi_d=2(m-1)d-2(m-2)\,,
\end{equation*}
which plugged into Eq. \eqref{eq:proof_ex_2} yields
\begin{equation}
\label{eq:proof_ex_3}
\esp_1^{\ast\diamond}(T_d)=\frac{m-1}{2}d^2+d\,.
\end{equation}
The proof is completed by comparing Eqs. \eqref{eq:proof_ex_1} and \eqref{eq:proof_ex_3} and using Eq. \eqref{eq:coupling_folded}.
\end{proof}

\subsection{Proof of Proposition \ref{prop:hypercube:2}}
\label{proof2}
\begin{proof}
We consider a version of the locally-weighted kernel $P_\varpi^\ast$ delayed by a factor $\lambda\in(0,1)$:
\begin{equation}
P^\ast_\lambda=\lambda P^\ast+(1-\lambda)\Id\,.
\end{equation}
Here again, we slightly change notation compared to the statement of   Proposition \ref{prop:hypercube:2}. The asymptotic rate of convergence and asymptotic variance of $\pi$-reversible Markov kernels can be assessed by studying their spectral properties. Indeed, defining the absolute spectral gap of a Markov kernel $P$ as
$$
\gap(P):=1-\sup\left\{|\lambda|,\;\lambda\in\Sp(P)\backslash\{1\}\right\}\,,
$$
where $\Sp(P)$ is the spectrum of $P$, Proposition 2 from \cite{rosenthal2003asymptotic} states that
\begin{equation}
\label{eq:spe}
\sup_{x\in\Zset}\lim_{t\to\infty}\frac{1}{t}\log\|\delta_x P^t-\pi\|=\log(1-\gamma(P))\,.
\end{equation}
We recall that since $P$ is a Markov operator, $\Sp(P)\subset(-1,1]$. Hence, the larger the spectral gap, the faster the convergence. In this proof, several elements of the proof of Proposition \ref{prop:hypercube:1} are used and in particular, the folded version of the two Markov chains.  A route to prove Prop. \ref{prop:hypercube:2} can be given by:
\begin{enumerate}
\item  introduce the two unfolded Markov chains with transition kernel $\bP_\lambda^\ast$ and $\bP$ which operate on $\Zset$,
\item  show that $\gap(\bP_\lambda^\ast)=\gap(P_\lambda^\ast)$ and $\gap(\bP)=\gap(P)$,
\item  show that $\gap(\bP_\lambda^\ast)=\gap(Q^\ast_\lambda)$ and $\gap(\bP^\ast)=\gap(Q)$,
\item  show that for $\lambda=2/d$, $\gap(Q)=\gap(Q_\lambda^\ast)$,
\item show that $\gap(P_\lambda^\ast)=\lambda\gap(P^\ast)$.
\end{enumerate}

Getting the analytical expression of $\gamma(P)$ and $\gamma(P_\lambda^\ast)$ is challenging. Instead of calculating the eigenvectors of the transition matrices $P$ and $P_\lambda^\ast$, we resort to the folded versions of those Markov chains in the same spirit as the proof of Proposition \ref{prop:hypercube:1}. Indeed, the resulting transition matrices on the folded space $\Yset=\{1,2,\ldots,2d+1\}$ are pentadiagonal and this facilitates the derivation of their spectrum. This is exploited at step 4 of our proof.

First, let us define the operators $\Gamma$ and $\Omega$ that map $\Zset$ to $\Yset$ and $\Yset$ to $\Zset$, respectively. Using the notation $\Zset=\{\Vset_1,\Eset_1,\ldots,\Vset_{d+1}\}$ defined in the proof of Proposition \ref{prop:hypercube:1}, $\Gamma$ maps a state $x\in\Zset$ to a step $y\in\Yset$ as follows:
\begin{itemize}
\item If there exists $k\in\nset$, such that $x=\Vset_k$, set $y=2(k-1)+1$.
\item If there exists $k\in\nset$, such that $x\in\Eset_k$, set $y=2k$.
\end{itemize}
The operator $\Omega$ maps a state $y\in\Yset$ to a step $x\in\Zset$ as follows:
\begin{itemize}
\item If there exists $k\in\nset$, such that $y=2k+1$, set $x=\Vset_{k+1}$.
\item If there exists $k\in\nset$, such that $y=2k$, pick $x$ uniformly at random in $\Eset_k$.
\end{itemize}
Hence, contrarily to $\Gamma$, $\Omega$ is a stochastic operator. More precisely, $\Omega$ and $\Gamma$ are matrices such that $\Omega\in\mathcal{M}_{d(n-1)+1,2d+1}((0,1))$ and $\Gamma\in\mathcal{M}_{2d+1,d(n-1)+1}((0,1))$ and their construction is detailed at Algorithm \ref{alg:matrices}. %As we shall see, the introduction of $\Gamma$ and $\Omega$ allows to define an alternative processes to $\{X_t\}_t$ but which has the same law and whose spectrum is easier to derive.

\begin{algorithm}
\caption{Construction of the mapping matrices}\label{alg:matrices}
\begin{algorithmic}[1]
\State set $\Omega_{1,\cdot}=\{\delta_{1,j}\}_{j\leq 2d+1}$ and $\Gamma_{\cdot,1}=\{\delta_{1,j}\}_{j\leq 2d+1}$
\State $k\gets 2$
\ForAll{$i=2,\ldots,d(n-1)+1$}
\If{it exists  $\ell\geq 0$ \,s.t. $i=\ell(n-1)+1$}
\State set $k\gets k+1$
\State set $\Omega_{i,\cdot}=\{\delta_{k,j}\}_{j\leq 2d+1}$ and $\Gamma_{\cdot,i}=\{\delta_{k,j}\}_{j\leq 2d+1}$
\State set $k\gets k+1$
\Else
\State set $\Omega_{i,\cdot}=\{\delta_{k,j}\}_{j\leq 2d+1}$ and $\Gamma_{\cdot,i}=(1/(n-2))\{\delta_{k,j}\}_{j\leq 2d+1}$
\EndIf
\EndFor
\end{algorithmic}
\end{algorithm}

To circumvent calculating the eigenvalues of $P$ and $P_\lambda^\ast$, a natural idea is to look at the spectrum of their equivalent transition kernels on the folded space $\Yset$ defined formally as
\begin{equation}
\label{eq:folded}
Q:=\Gamma P \Omega\,\quad \text{and}\quad Q^\ast_\lambda:=\Gamma P^\ast_\lambda \Omega
\end{equation}
and illustrated at Figure \ref{fig:representation} (in the case $\lambda=0$). Unfortunately, those folded Markov chains cannot be directly used since $\Sp(Q)\neq \Sp(P)$ and $\Sp(Q^\ast_\lambda)\neq \Sp(P^\ast_\lambda)$. Indeed, it can be readily checked that
\begin{equation}
\text{Tr}(P)=1+(n-1)(d-1)\neq 2d-1=\text{Tr}(Q)
\end{equation}
and thus, should $\gamma(Q)$ and $\gamma(Q_\lambda^\ast)$ be analytically tractable, one could not call on to Eq. \eqref{eq:spe} to conclude the proof.

The trick is to consider the unfolded kernels stemming from $Q$ and $Q_\lambda^\ast$ and defined as
\begin{equation}
\label{eq:unfolded}
\bP:=\Omega Q \Gamma\,\quad \text{and}\quad \bP^\ast_\lambda:=\Omega Q^\ast_\lambda \Gamma\,.
\end{equation}
Intuitively, while the dynamic of $P$ (resp. $P_\lambda^\ast$) is fundamentally on $\Zset$, $\bP$ (resp. $\bP_\lambda^\ast$) generates a process which fundamentally operates on $\Yset$ via $Q$ (resp. $Q_\lambda^\ast$) and which is then mapped back to $\Zset$. The operator $\Sigma:=\Omega\Gamma$ acts as an operator on $\Zset$ which randomizes the inner part of each edge while leaving untouched the vertices. In fact, $\Sigma$ is a Markov operator which is reversible w.r.t. the uniform distribution since it is symmetric. Therefore, it is easy to see that applying $\Sigma$ after or before $P$ leads to same Markov transition and thus, $P$ and $\Sigma$ commute. The same can be said about $\Sigma$ and $P_\lambda^\ast$. Finally since $\Sigma\Sigma=\Sigma$, we have that $\Sigma$ is an orthogonal projector of $\Zset$.

It can be readily checked that, for $m=3$, $P=\bP$ and thus we only need to prove step 1 of the proof for $m>3$ since in that case, $P\neq \bar{P}$. Combining Eqs. \eqref{eq:folded}, \eqref{eq:unfolded} and using the fact that $P\Sigma=\Sigma P$, we have
\begin{equation}\label{eq:relation}
 \bP=\Sigma P \Sigma=P\Sigma=\Sigma P
 \end{equation}
and similarly for $(\bP_\lambda^\ast,P_\lambda^\ast)$.  Even though they are different, $\bP$ and $\bP_\lambda^\ast$ are still useful for our analysis. For example, Lemma \ref{lem:2_bis} shows that for any $t>0$ and any starting point $x$ in the set of vertices, we have
\begin{equation}
\label{eq:equiv_cv}
\|\delta_xP^t-\pi\|=\|\delta_x\bP^t-\pi\|\quad \text{and}\quad
\|\delta_x{P_\lambda^\ast}^t-\pi\|=\|\delta_x{\bP_\lambda^{\ast\,t}-\pi}\|\,.
\end{equation}
As a consequence, when assessing the efficiency of $P$ one can equivalently study $\bP$ and similarly for $P^\ast$ with $\bP_\lambda^\ast$. We also note that since $\bP$ (resp. $\bP_\lambda^\ast$) is the composition of two $\pi$-reversible Markov kernels which commute, $\bP$ (resp. $\bP_\lambda^\ast$) is also $\pi$-reversible.

Recall that the absolute spectral gap of a $\pi$-reversible Markov kernel $K$ is the quantity $\gap(K)\in[0,1]$ defined by
$$
\gap(K)=\min\left[\gap_L(K),\gap_R(K)\right]
$$
where the left and right spectral gaps are respectively given by
\begin{multline*}
\gap_L(K)=\inf_{f\in\Lb}\pscal{f}{(I+K)f}_\pi\,,\\
\gap_R(K)=\inf_{f\in\Lb}\pscal{f}{(I-K)f}_\pi\,.
  \end{multline*}
If $K$ is a positive operator on $\Lb$, it is easy to see that $\gap_L(K)\geq 1$ and since $\gap_R(K)\in[0,1]$, then  $\gap(K)=\gap_R(K)$. By Lemma \ref{lem:2_bis}, we know that $P$, $\bP$, $P_\lambda^\ast$ and $\bP_\lambda^\ast$ are all positive operators on $\Lb$ and thus we only need to show that
$\gap_R(P)=\gap_R(\bP)$ and $\gap_R(P_\lambda^\ast)=\gap_R(\bP_\lambda^\ast)$. The proof is identical in both cases. We have for any $f\in\Lb$ that
\begin{multline*}
\pscal{f}{(I-P)f}_\pi=\pscal{f}{(I-\bP+\bP-P)f}_\pi\\
=\pscal{f}{(I-\bP)f}_\pi-\pscal{f}{(P-\bP)f}_\pi\leq \pscal{f}{(I-\bP)f}_\pi
\end{multline*}
since by Lemma \ref{lem:2_bis}, $P-\bP$ is a positive operator on $\Lb$. Thus $\gap_R(P)\leq \gap_R(\bP)$ but since, again by Lemma \ref{lem:2_bis}, $\Sp(\bP)\backslash \{0\}\subset \Sp(P)$ we have that $\gap_R(P)=\gap_R(\bP)$, which establishes step 1 of the proof for each $m\geq 3$ and each $d\geq 2$. Lemma \ref{lem:0} proves that $\gap(\bP)=\gap(Q)$ and $\gap(\bP_\lambda^\ast)=\gap(Q_\lambda^\ast)$. Lemma \ref{lem:1} completes the proof by showing that $\gap(Q)=\gap(Q^\ast_{2/d})$. To prove the last point, we note that for each $f\in\Lb$
$$
\pscal{f}{(I-P_\lambda^\ast)f}_\pi=\lambda\pscal{f}{(I-P^\ast)f}
$$
which yields $\gap_R(P_\lambda^\ast)=\lambda\gap_R(P^\ast)$.
Since $\bP_\lambda^\ast$ and $\bP$ are positive operators on $\Lb$, we have that $\gap(P_\lambda^\ast)=\lambda\gap(P^\ast)$.
\end{proof}

\subsection{Proof of Proposition \ref{prop:hypercube:3}}
\label{proof3}
\begin{proof}
The off-diagonal part of the transition kernels is given by
\begin{itemize}
\item for $k\in\{1,\ldots,d\}$, $(x,y)\in\Eset_k$,
$$
P_\varpi^\ast(x,y)=\frac{1}{m}\,,\qquad P_\omega(x,y)=\frac{1}{dm}\,,
$$
\item for $k\in\{1,\ldots,d\}$, $x\in\Eset_k$, $y\in\Vset_{k}\cup\Vset_{k+1}$,
$$
P_\varpi^\ast(x,y)=\frac{1}{2m}+\frac{1}{2m}\1_{(x,y)\in\Eset_d\times\Vset_{d+1}\cup\Eset_1\times\Vset_1}\,,\qquad P_\omega(x,y)=\frac{1}{dm}\,,
$$
\item for $k\in\{2,\ldots,d\}$, $x\in\Vset_k$ and $y\in\Eset_{k}\cup\Eset_{k+1}\cup\Vset_{k-1}\cup\Vset_k$,
$$
P_\varpi^\ast(x,y)=\frac{1}{2m}\,,\qquad P_\omega(x,y)=\frac{1}{dm}\,,
$$
\item for  $x\in\Vset_1$ and $y\in\Eset_1\cup\Vset_2$,
$$
P_\varpi^\ast(x,y)=\frac{1}{2m}+\frac{1}{2m}\1_{(x,y)\in\Vset_1\times\Eset_1}\,,\qquad P_\omega(x,y)=\frac{1}{dm}\,,
$$
and similarly for $x\in\Vset_{d+1}$ and $y\in\Eset_d\cup\Vset_d$.
\end{itemize}
Thus for all $x\neq y$, we have $P_\varpi^\ast(x,y)\geq ({d}/{2}) P_\omega(x,y)$ and  $P_\varpi^\ast$ and $P_\omega$ admit a strong form of Peskun ordering. This implies, see for instance \cite[Theorem 2]{zanella2020informed}, that for any bounded function $f:\Xset\to\rset$,
$$
\var(f,P_\varpi^\ast)\leq \frac{2}{d}\var(f,P_\omega)+\left(\frac{2}{d}-1\right)\var_\pi f(X)\,.
$$
Note that the inequality is tight since for $d=2$ and $m=3$, $P_\varpi^\ast=P_\omega$.
\end{proof}
\section{Technical Lemmas}

\begin{lemma}
\label{lem:2}
Let $P$ be the transition matrix of the random-scan, $Q$ its equivalent representation on the folded state space and $\Omega$ and $\Gamma$ be the two mapping matrices defined at Algorithm \ref{alg:matrices} and let $\bP:=\Omega Q\Gamma$ and $\bPast_\lambda:=\Omega \Qast_\lambda\Gamma$. Then we have for $x=(1\,,1\,,\cdots\,, 1\,,1)$ and all $t>0$
\begin{equation}
\label{eq:proof_ex_5}
\delta_x P^t=\delta_x\bP^t\,.
\end{equation}
Similarly for the locally-weighted algorithm, we have for all $t>0$
\begin{equation}
\label{eq:proof_ex_6}
\delta_x{P_\lambda^\ast}^t=\delta_x\bar{P}_\lambda^{\ast\,t}\,.
\end{equation}
\end{lemma}

\begin{proof}[Proof of Lemma \ref{lem:2}]
We prove Eqs. \eqref{eq:proof_ex_5} and \eqref{eq:proof_ex_6} by induction. For notational simplicity, we present the proof for $\lambda=1$, \ie $\bP_\lambda^\ast\equiv\bP^\ast$. We first establish Eq. \eqref{eq:proof_ex_5}. We use the notation of Proof of Proposition \ref{prop:hypercube:1} and let $\Vset:=\{\Vset_1,\ldots,\Vset_{d+1}\}$. The initialisation follows from noting that $P(x,\,\cdot)=\bP(x,\,\cdot)$, for any $x\in\Vset$. Now, assume that $\delta_xP^t=\delta_x\bP^t$ and note that
\begin{multline}
\label{eq:proof_ex_6b}
\delta_xP^{t+1}=\sum_{i\in\Vset}P^t(x,i)P(i,\,\cdot\,)+\sum_{i\in\Eset}P^t(x,i)P(i,\,\cdot\,)\\
=\sum_{i\in\Vset}\bP^t(x,i)P(i,\,\cdot\,)+\sum_{i\in\Eset}\bP^t(x,i)P(i,\,\cdot\,)\,,\\
=\sum_{i\in\Vset}\bP^t(x,i)\bP(i,\,\cdot\,)+\sum_{k=1}^d\sum_{i\in\Eset_k}\bP^t(x,i)P(i,\,\cdot\,)\,,
\end{multline}
where the first line comes from the recursion assumption and the second follows from the initialisation stage. The second term in the last line of Eq. \eqref{eq:proof_ex_6b} requires a special attention. In particular, Lemma \ref{lem:2_1} shows that for all $x\in\Vset$ and any edge state $i$, $\bP^t(x,i)$ depends only on $i$ through the edge it belongs to. In other words, for all $k\in\{1,\ldots,d\}$, there exists a function $\rho_k^t$ such that $\bP^t(x,i)=\varrho_k^t(x)$ for all $i\in\Eset_k$ and all $x\in\Vset$. Plugging this into Eq. \eqref{eq:proof_ex_6b} yields
\begin{equation}
\label{eq:proof_ex_8}
\delta_xP^{t+1}=\sum_{i\in\Vset}\bP^t(x,i)\bP(i,\,\cdot\,)+\sum_{k=1}^d\varrho_k^t(x)\sum_{i\in\Eset_k}\bP(i,\,\cdot\,)\,.
\end{equation}
Finally, we note that
\begin{equation}
\label{eq:proof_ex_7}
\sum_{i\in\Eset_k}P(i,\,\cdot\,)=\sum_{i\in\Eset_k}\bP(i,\,\cdot\,)\,.
\end{equation}
Indeed, by straightforward algebra, denoting $\Vset_{k-1}$ and $\Vset_k$ the adjacent vertices of $\Eset_k$, it can be readily checked that $\sum_{i\in\Eset_k}P(i,j)=\sum_{i\in\Eset_k}\bP(i,j)=\{(m-2)/dm\} \1_{j\in\{\Vset_{k-1},\Vset_k\}}+(1-2/dm)\1_{j\in\Eset_k}$.
Combining Eqs. \eqref{eq:proof_ex_8} and \eqref{eq:proof_ex_7} finally yields
\begin{multline*}
\delta_xP^{t+1}=\sum_{i\in\Vset}\bP^t(x,i)\bP(i,\,\cdot\,)+\sum_{k=1}^d\varrho_k^t(x)\sum_{i\in\Eset_k}\bP(i,\,\cdot\,)\\
=\sum_{i\in\Vset}\bP^t(x,i)\bP(i,\,\cdot\,)+\sum_{i\in\Eset}\bP(x,i)^t\bP(i,\,\cdot\,)=\delta_x\bP^{t+1}\,,
\end{multline*}
which completes the first part of the proof. To prove Eq. \eqref{eq:proof_ex_6}, we note that the initialisation is straightforward since there is a one-to-one mapping on $\Vset$ between the folded and unfolded representation. The induction is concluded by applying the same reasoning, noting that Lemma \ref{lem:2_1} holds for $\bP_\lambda^\ast$ also and that
\begin{equation}
\label{eq:proof_ex_7_ast}
\sum_{i\in\Eset_k}P_\lambda^\ast(i,\,\cdot\,)=\sum_{i\in\Eset_k}\bP_\lambda^\ast(i,\,\cdot\,)\,.
\end{equation}
Indeed,
\begin{itemize}
\item  for $k=1$,
\begin{itemize}
\item $\sum_{i\in\Eset_1}P_\lambda^\ast(i,j)=(m-2)/m=\sum_{i\in\Eset_1}\bP_\lambda^\ast(i,j)$ if $j=\Vset_1$,
\item    $\sum_{i\in\Eset_1}P_\lambda^\ast(i,j)=1-3/2m=\sum_{i\in\Eset_1}\bP_\lambda^\ast(i,j)$ if $j\in\Eset_1$,
\item    $\sum_{i\in\Eset_1}P_\lambda^\ast(i,j)=(m-2)/2m=\sum_{i\in\Eset_1}\bP_\lambda^\ast(i,j)$ if $j=\Vset_2$,
\item    $\sum_{i\in\Eset_1}P_\lambda^\ast(i,j)=0=\sum_{i\in\Eset_1}\bP_\lambda^\ast(i,j)$ for any $j\in\Vset\backslash\{\Vset_1,\Eset_1,\Vset_2\}$,
\end{itemize}
\item for $1<k<d$,
\begin{itemize}
\item $\sum_{i\in\Eset_k}P_\lambda^\ast(i,j)=(m-2)/2m=\sum_{i\in\Eset_1}\bP^\ast(i,j)$ if $j\in\{\Vset_k,\Vset_{k+1}\}$,
\item $\sum_{i\in\Eset_k}P_\lambda^\ast(i,j)=1-1/m=\sum_{i\in\Eset_k}\bP_\lambda^\ast(i,j)$ if $j\in\Eset_k$,
\item $\sum_{i\in\Eset_k}P_\lambda^\ast(i,j)=0=\sum_{i\in\Eset_k}\bP_\lambda^\ast(i,j)$ for any $j\in\Vset\backslash\{\Vset_k,\Eset_k,\Vset_{k+1}\}$,
\end{itemize}
\item the case $k=d$ is identical to the case $k=1$.
\end{itemize}
\end{proof}

\begin{lemma}
\label{lem:2_1}
In the context of Lemma \ref{lem:2}, for any $x\in\Vset$, for all $k\in\{1,\ldots,d\}$ and $i\in\Eset_k$, the transition probabilities $\bP^t(x,i)$ and $\bar{P}_\lambda^{\ast\,t}(x,i)$ are conditionally independent of $i$ given $i\in\Eset_k$.
\end{lemma}
\begin{proof}
We prove Lemma \ref{lem:2_1} by recursion for $\bP^t(x,i)$ only, the proof for $\bar{P}_\lambda^{\ast\,t}(x,i)$ being identical. The initialisation follows from noting that for any $i$ belonging to an edge connected to $x$, $\bP(x,i)=1/dm$. For any $i$ belonging to an edge not connected to $x$, $\bP(x,i)=0$. As a consequence, for all $i$ belonging to the same edge, $\bP(x,i)$ is independent of $i$. Let us assume that for any $x\in\Vset$, for all $k\in\{1,\ldots,d\}$, for all $i\in\Eset_k$, $\bP^t(x,i)=\varrho_k^t(x)$, \ie $\bP^t(x,i)$ is independent of $i$. We have:
\begin{multline*}
\bP^{t+1}(x,i)=\sum_{j\in\Xset}\bP^t(x,j)\bP(j,i)
=\sum_{j\in\{\Vset_{k-1},\Vset_k\}}\bP^t(x,j)\bP(j,i)+\sum_{j\in\Eset_k}\bP^t(x,j)\bP(j,i)\,,\\
=\sum_{j\in\{\Vset_{k-1},\Vset_k\}}\bP^t(x,j)\slash{dm}+\varrho_k^t(x)\sum_{j\in\Eset_k}\bP(j,i)\,,
\end{multline*}
and since for all $(i,j)\in\Eset_k^2$, $\bP(j,i)$ is independent of $i$, there exists $\rho_k^{t+1}(x)$ such that for all $i\in\Eset_k$, $\bP^{t+1}(x,i)=\rho_k^{t+1}(x)$, which completes the proof.
\end{proof}

\begin{lemma}
 \label{lem:2_bis}
 Let $m>3$. The Markov kernels $\bP$, $\bP_\lambda^\ast$, $P$ and $P_\lambda^\ast$ are all positive operators on $\Lb$. Moreover, they satisfy
 $$
 P_\lambda^\ast-\bP_\lambda^\ast\succeq 0\,,\qquad
  P-\bP\succeq 0
 $$
 and
 $$
 \Sp(\bP_\lambda^\ast)\backslash\{0\}\subset\Sp(P_\lambda^\ast)\,,\qquad  \Sp(\bP)\backslash\{0\}\subset\Sp(P)\,.
 $$
\end{lemma}
\begin{proof}
First, note that for all $x\in\cup \Vset_k$, $\bP_\lambda^\ast(x,\cdot)=P_\lambda^\ast(x,\cdot)$ and for all $x\in\Eset_k$, $k\in\{2,\ldots,d-1\}$ and $y\neq x$
$$
P_\lambda^\ast(x,y)=\frac{1}{m}\1_{y\in \Eset_{k}}+\frac{1}{2m}\1_{y\in\Vset_{k}\cup\Vset_{k+1}}\leq \frac{1}{m}\left[\frac{m-1}{m-2}\right]\1_{y\in \Eset_{k}}+\frac{1}{2m}\1_{y\in\Vset_{k}\cup\Vset_{k+1}}=\bP_\lambda^\ast(x,y)\,,
$$
and for all $x\in\Eset_1$ with $y\neq x$
\begin{multline*}
  P_\lambda^\ast(x,y)=\frac{1}{m}\1_{y\in \Eset_{1}}+\frac{1}{m}\1_{y\in \Vset_{1}}+\frac{1}{2m}\1_{y\in \Vset_{2}}+\\ \leq 
  \frac{1}{m}\left[\frac{m-3/2}{m-2}\right]\1_{y\in \Eset_{1}}+\frac{1}{m}\1_{y\in \Vset_1}+\frac{1}{2m}\1_{y\in\Vset_{2}}
=\bP_\lambda^\ast(x,y)\,,
\end{multline*}
and similarly for $x\in\Eset_d$. Thus $\bP_\lambda^\ast$ dominates $P_\lambda^\ast$ in the Peskun ordering sense (see \cite{peskun1973optimum}) and by \cite[Lemma 3]{tierney1998note}, $P_\lambda^\ast-\bP_\lambda^\ast$ is a positive operator on $\Lb$. Even though the calculations are different, the same result can be shown to hold for $P$ and $\bP$. Second, for $m> 3$ the second and third rows of $\bP$ (resp. $\bP_\lambda^\ast$) are identical and thus the leading principal minors of order 3 and higher of $\bP$ (resp. $\bP_\lambda^\ast$) are null. It can be checked that the leading principal minor of order 2 of $\bP$ and $\bP_\lambda^\ast$ are respectively
\begin{multline*}
\frac{1}{(md)^2}\left[\left(md-(m-1)\right)\left(\frac{md-2}{m-2}\right)-1\right]\\
\geq \frac{1}{m(md)^2}\left[(m-1)(d-1)(md-2)-1\right]\geq 0
\end{multline*}
for $m\geq 3$, $d\geq 2$ and
$$
\frac{1}{m^2}\left[\frac{3}{4}\frac{2m-3}{m-2}-1\right]\geq 0\,.
$$
for any $m\geq 3$. As a consequence, $\bP$ and $\bP_\lambda^\ast$ are semi-definite positive for $m>3$. It is also possible to show that $\bP$ and $\bP_\lambda^\ast$ are semi-definite positive by induction on $d\geq 2$, for $m=3$. Combining $P_\lambda^\ast-\bP_\lambda^\ast\succeq 0$ and $\bP_\lambda^\ast\succeq 0$ yields $P_\lambda^\ast\succeq 0$ and similarly for $P$. To prove the last result on the spectrum, we note that $0\in\Sp(\bP)$ since $\det \Sigma=0$ and $\bP=P\Sigma$. In particular, the null space of $\Sigma$ being of dimension $d(m-3)$, we have that $0\in\Sp(\bP)$ with multiplicity $d(m-3)$. Let $f$ be an eigenfunction of $\bP$ whose corresponding eigenvalue $\rho$ is not zero. Then
$$
\bP f=\rho f\Rightarrow \Sigma P\Sigma f=\rho \Sigma f\Leftrightarrow P\Sigma f=\rho \Sigma f
$$
where we have used the facts that $P$ and $\Sigma$ commute and that $\Sigma$ is a projector of $\Zset$. It comes that $\rho\in \Sp(P)$ and that its corresponding eigenfunction is $\Sigma f$. The same argument holds for $\bP_\lambda^\ast$ and $P_\lambda^\ast$.
\end{proof}

\begin{lemma}
\label{lem:0}
In the context of the proof of Proposition \ref{prop:hypercube:2}, $\gap(\bP)=\gap(Q)$ and $\gap(\bP_\lambda^\ast)=\gap(Q_\lambda^\ast)$.
\end{lemma}
\begin{proof}
Without loss of generality and for notational simplicity, the proof is carried out in the case $\lambda=1$, \ie $\bP^\ast_\lambda\equiv \bP^\ast$. Central to this proof is the fact that $\Gamma\Omega=I_{2d+1}$, where $\Gamma$ and $\Omega$ are the two change of basis matrices from $\Zset$ to its folded counterpart $\Yset$ and conversely, see their formal definition given at Algorithm \ref{alg:matrices}. Indeed, it can be readily checked that $\Omega$ is an injection from $\Yset$ to $\Zset$ and thus admits a left inverse. This left inverse corresponds to the reverse transformation from $\Zset$ to $\Yset$, which is precisely $\Gamma$.

We establish $\gap(\bP)=\gap(Q)$ and  $\gap(\bP^\ast)=\gap(Q^\ast)$ is obtained in the same way. Let $\lambda\in\Sp(Q)$. Then, by definition of $\Sp(\bP)$, there exists a non null vector $y_0\in\rset^{2d+1}$ such that
\begin{equation}
\label{eq:sp2}
Q y_0=\lambda y_0\Leftrightarrow Q \Gamma\Omega y_0=\lambda y_0\Leftrightarrow\Omega Q \Gamma\Omega y_0=\lambda \Omega y_0
\Leftrightarrow \bP\Omega y_0=\lambda\Omega y_0\,.
\end{equation}
Moreover, since $\text{ker}(\Omega)$ is restricted to the null vector $0_{2d+1}$, $\Omega y_0\neq 0_{(n-1)d+1}$ and $\lambda\in\Sp(\bP)$.

Let $\lambda\in\Sp(\bP)$, then, by definition of $\Sp(\bP)$, there exists a non null vector $x_0\in\rset^{(n-1)d+1}$ such that, $\bP x_0=\lambda x_0$. By definition of $\bP$, we have that
\begin{equation}
\label{eq:sp1}
\Omega Q \Gamma x_0=\lambda x_0\Leftrightarrow \Gamma \Omega Q \Gamma x_0=\lambda \Gamma x_0 \Leftrightarrow Q \Gamma x_0=\lambda \Gamma x_0\,.
\end{equation}
Now, $\text{ker}(\Gamma)$ is not restricted to $0_{(n-1)d+1}$. Indeed, it can be readily checked that $x_0:=(0,1,-1,0,\ldots,0)$ belongs to $\text{ker}(\Gamma)$. As a consequence, for any $\lambda\in\Sp(\bP)$ if the eigenvector associated $\lambda$ does not belong to $\ker(\Gamma)$, then $\lambda\in\Sp(Q)$. In contrast, if $x_0\in\ker(\Gamma)$, it cannot be concluded whether or not $\lambda\in\Sp(Q)$. A careful look at the transition matrix $\bP$ shows that the columns of $\bP$ are not linearly independent. In particular, the columns corresponding to states $x_0\in\Zset$ belonging the same edge $\Eset_k$ are all equal. As a consequence, $\text{rank}(\bP)=2d+1$ which implies that $\text{dim}(\text{ker}(\bP))=(m-1)d+1-2d-1=(m-3)d$. This shows that $0\in\Sp(\bP)$ with multiplicity $(m-3)d$ and in fact $\ker(\Gamma)=\ker(\bP)$. Conversely, $\text{rank}(Q)=2d+1$ and thus $\text{dim}(\text{ker}(Q))=0$ which implies that $0\not\in\Sp(Q)$. Combining those different observations yield to
\begin{equation}
\label{eq:lasteq}
\Sp(\bP)=\Sp(Q)\cup 0\,.
\end{equation}
The proof is concluded by noting that from the definition of the absolute spectral gap, we have
\begin{equation}
\gap(\bP)=1-\sup_{\lambda\in\Sp(\bP)}|\lambda|=1-\sup_{\lambda\in\Sp(\bP)\backslash\{0\}}|\lambda|=
1-\sup_{\lambda\in\Sp(Q)}|\lambda|=\gap(Q)\,.
\end{equation}
\end{proof}

\begin{lemma}
\label{lem:1}
In the context of Proposition \ref{prop:hypercube:2} and whenever $d$ is even, we have that
$$
\gap(Q)=\gap(Q_{\lambda}^\ast)\qquad \text{for}\;\lambda=2/d\,.
$$
\end{lemma}

\begin{proof}[Proof of Lemma \ref{lem:1}]
The proof is established from the following series of steps
\begin{itemize}
\item calculating the characteristic polynomial for each matrix:
$$
\chi(\lambda)=\dete(Q-\lambda \Id)\,,\qquad\chi_{2/d}^\ast(\lambda)=\dete(Q_{2/d}^\ast-\lambda \Id)
$$
\item developing the determinant in a specific way, we show that  $\chi$ and $\chi_{2/d}^\ast$ only differ through one factor:
$$
\chi(\lambda)=\{m(d-1)-dm\lambda\}^2\det M_\lambda\,,\qquad \chi_{2/d}^\ast(\lambda)=\{m(d-2)+1-dm\lambda\}^2\det M_\lambda\,.
$$
\item denoting $\Lambda=\{\lambda\in(-1,1)\,,\;\det M_\lambda =0\}$, we have
\begin{multline*}
\Sp(Q)=\left\{1\;,\quad\lambda_0:=\frac{d-1}{d}\;,\quad \Lambda\right\}\,,
\\ \Sp(Q_{2/d}^\ast)=\left\{1\;,\quad\lambda_0^\ast:=\frac{d-2}{d}+\frac{1}{dm}\; ,\quad \Lambda\right\}\,.
\end{multline*}
\item for both cases, the larger eigenvalue smaller than 1 is in $\Lambda$. We first calculate the traces
$$
\tr(Q)=2d-1\,,\qquad
\tr(Q_{2/d}^\ast)=2d+\frac{2}{md}-1-\frac{2}{d}\,,
$$
and since $\tr(Q)=\sum_{\lambda\in\Sp(Q)}\lambda$, we have
\begin{equation}
\label{eq:proof_ex_4}
\tr(Q)=1+\lambda_0+\sum_{\lambda\in\Lambda}\lambda\,,\qquad \tr(Q^\ast)=1+\lambda_0^\ast+\sum_{\lambda\in\Lambda}\lambda
\end{equation}
If $\lambda_0>\sup\Lambda$, then
$$
1+\lambda_0+\sum_{\lambda\in\Lambda}\lambda<1+2d\lambda_0=2d-1=\tr(Q)\,,
$$
which contradicts the LHS of Eq. \eqref{eq:proof_ex_4} and we have $\lambda_0\leq \sup\Lambda$. Similarly, if $\lambda_0^\ast>\sup\Lambda$, then
$$
1+\lambda_0^\ast+\sum_{\lambda\in\Lambda}\lambda<1+2d\lambda_0^\ast=2d-1+2\frac{1-m}{m}<\tr(Q)\,,
$$
which contradicts the RHS of Eq. \eqref{eq:proof_ex_4} and $\lambda_0^\ast\leq \sup \Lambda$. Therefore, regardless whether or not $(\lambda_0,\lambda_0^\ast)\in\Lambda^2$, the largest eigenvalue smaller than 1 of $Q$ and $Q_{2/d}^\ast$ is identical. Since by Lemma \ref{lem:2_bis}, $\bP^\ast_\lambda$ and $\bP$ are positive operators and using \eqref{eq:lasteq}, we have that $Q$ and $Q_{2/d}^\ast$ are also positive operators. As a consequence, the absolute spectral gaps of $Q$ and $Q_{2/d}^\ast$ coincide with their right spectral gap which is $1-\sup\Lambda$ and thus equal.
\end{itemize}
\end{proof}
\begin{funding}
FM's research is supported in part by the Natural Sciences and Engineering Research Council of Canada. He would like to thank the Insight Centre for Data Analytics for funding the post-doc fellowship that allowed to develop this research. The Insight Centre for Data Analytics is supported by Science Foundation Ireland under Grant Number SFI/12/RC/2289.
\end{funding}

\bibliographystyle{imsart-number} % Style BST file
\bibliography{biblio}       % Bibliography file (usually '*.bib')

%% or include bibliography directly:
%\begin{thebibliography}{4}
%
%\bibitem{r1}
%Billingsley, P. (1999). \textit{Convergence of
%Probability Measures}, 2nd ed.
%New York: Wiley.
%
%\bibitem{r2}
%Bourbaki, N.  (1966). \textit{General Topology}  \textbf{1}.
%Reading, MA: Addison--Wesley.
%
%\bibitem{r3}
%Ethier, S. N. and Kurtz, T. G. (1985).
%\textit{Markov Processes: Characterization and Convergence}.
%New York: Wiley.
%
%\bibitem{r4}
%Prokhorov, Yu. (1956).
%Convergence of random processes and limit theorems in probability
%theory. \textit{Theory  Probab.  Appl.}
%\textbf{1} 157--214.
%
%\end{thebibliography}
%
\end{document}

%% file: definitions.tex
\def\Xset{\mathsf{X}}
\def\Zset{\mathsf{Z}}
\def\Xalg{\mathcal{X}}
\def\Yalg{\mathcal{Y}}

\def\proba{\mathbb{P}}

\def\Yset{\mathsf{Y}}
\def\tX{\tilde{X}}
\def\tx{\tilde{x}}

\def\rset{\mathbb{R}}
\def\norm{\mathcal{N}}
\def\T{^{T}}
\def\rmd{\mathrm{d}}

\def\ie{\textit{i.e.\;}}

\def\det{\mathrm{det}}
\def\Id{\mathrm{Id}}
\def\Ltwo{\mathcal{L}^2}

\newcommand{\pscal} [2]{\left\langle #1 , #2\right\rangle}

\def\eps{\epsilon}
\def\Past{P_\omega^{\ast}}
\def\esp{\mathbb{E}}

\def\Sp{\sigma}
\def\1{\mathds{1}}

\def\nset{\mathbb{N}}
\def\inf{\text{inf}}
\def\det{\text{det}}
\def\tr{\text{tr}}
\def\gap{\mathrm{Gap}}
\def\var{\mathrm{ var}}

\def\bP{\bar{P}}

\def\Mast{M^\ast}
\def\Past{P^\ast}
\def\iid{\mathrm{iid}}
\def\unif{\mathrm{unif}}
\def\dete{\mathrm{det}}

\def\Vset{\mathcal{V}}
\def\Eset{\mathcal{E}}
\def\bPast{{\bar{P}^\ast}}
\def\Qast{{Q^\ast}}

\def\Pfrak{\mathcal{P}}
\def\eg{\textit{e.g.\;}}
\def\Sp{\text{Sp}}

\def\Lb{L_b^{(0,1)}(\Zset)}
\def\0{\mathbf{0}}
\def\1{\mathbf{1}} 

%% file: packages.tex
\usepackage{t1enc,graphicx,amssymb,amsmath,latexsym,float,color,enumerate}
\usepackage{dsfont,amsmath,afterpage,textcomp,array,amsthm}
\usepackage{algorithm}
\usepackage{algpseudocode}
\usepackage{algorithmicx}
\usepackage{tikz}
\usetikzlibrary{arrows}
\usepackage[english]{babel}
\usepackage{epstopdf}
\usepackage{hyperref}
\hypersetup{colorlinks,breaklinks,
            urlcolor=[rgb]{0,0.5,0.5},
            linkcolor=[rgb]{0,0.5,0.5},
            citecolor=[rgb]{0,0.2,.8}}
\usepackage{multirow}
\usepackage{natbib}
\usepackage{caption}
\usepackage{amsfonts}
\usepackage{mathrsfs}
\DeclareMathAlphabet{\mathpzc}{OT1}{pzc}{m}{it}
\usepackage{float}
\usepackage{hyperref}
\usepackage{verbatim}